%% file: clean.tex
\pdfoutput=1

\documentclass[11pt]{scrartcl}
\usepackage[utf8]{inputenc}
\usepackage[T1]{fontenc}
\usepackage[a4paper, total={16.5cm, 24cm}]{geometry}
\usepackage{microtype} 
\usepackage{authblk}
\title{The End Justifies the Mean: A Linear Ranking Rule for Proportional Sequential Decisions}
\date{\vspace{-1.5cm}}

\author[1]{Carmel Baharav}
\author[2]{Niclas Boehmer}
\author[1]{Bailey Flanigan}
\author[2]{\\Maximilian T. Wittmann}
\affil[1]{MIT, USA}
\affil[2]{Hasso Plattner Institute, University of Potsdam, Germany}

\usepackage{xcolor}
\usepackage{graphicx}
\usepackage{tikz}
\usetikzlibrary{quotes,angles,arrows.meta}
\usepackage{pgfplots}
\pgfplotsset{compat=1.17}
\usepgfplotslibrary{groupplots}

\usepackage{amsmath, amssymb, amsthm, mathtools, nicefrac, siunitx}
\usepackage{amsfonts}

\usepackage{thmtools}
\usepackage{thm-restate}
\newtheorem{theorem}{Theorem}[section]

\newtheorem{example}[theorem]{Example}

\newtheorem{proposition}[theorem]{Proposition}
\newtheorem{lemma}[theorem]{Lemma}
\newtheorem{corollary}[theorem]{Corollary}
\theoremstyle{definition}
\newtheorem{definition}[theorem]{Definition}
\newtheorem{remark}[theorem]{Remark}

\usepackage[pagebackref]{hyperref}
\hypersetup{
    pdfencoding=auto,
    psdextra,
    colorlinks=true,
    citecolor=green!50!black,
    linkcolor=red!60!black,
}

\usepackage[nameinlink]{cleveref}

\usepackage{natbib}
\usepackage{algorithm}
\usepackage[noend]{algorithmic}

\usepackage[inline]{enumitem}
\usepackage{booktabs}
\usepackage{subcaption}
\usepackage{multirow}

\usepackage{verbatim}

\usepackage[textsize=tiny]{todonotes}

\newcommand{\prooflink}[1]{\marginline{\vspace{0.5cm}\footnotesize \hyperlink{restated#1}{\hypertarget{original#1}{[Proof]}}}}
\newcommand{\restatehere}[1]{%
    \marginline{\vspace{0.6cm}\footnotesize \hyperlink{original#1}{\hypertarget{restated#1}{[Main]}}}%
    \csname #1\endcsname*%
}

\usepackage{colortbl}     
\usepackage{xcolor}       

\usepackage{gensymb}

\DeclareCaptionLabelFormat{fullfig}{Figure \thefigure#2}

\newcommand{\cX}{\mathcal{X}}
\newcommand{\cD}{\mathcal{D}}
\newcommand{\R}{\mathbb{R}}
\newcommand{\iid}{\overset{\mathrm{iid}}{\sim}}
\newcommand{\dang}{d_{\measuredangle}}
\newcommand{\thang}{\theta_{\mathrm{ang}}}
\newcommand{\thar}{\theta_{\mathrm{arith}}}

\newcommand{\prof}{(\boldsymbol{\theta}, \boldsymbol{\alpha})}
\newcommand{\profset}[1]{P_{#1}}
\newcommand{\D}{\mathcal{D}}
\newcommand{\btheta}{\boldsymbol{\theta}}
\newcommand{\balpha}{\boldsymbol{\alpha}}
\newcommand{\thetai}{\theta^{(i)}}

\newcommand{\thetaone}{\theta^{(1)}}
\newcommand{\thetatwo}{\theta^{(2)}}
\newcommand{\alphai}{\alpha^{(i)}}

\newcommand{\vphi}{\varphi}
\newcommand{\agi}{\textsc{IP}_i}
\newcommand{\ip}{\textsc{IP}}

\newcommand{\vps}{\vphantom{\sum_{j = 1}^n}}

\usepackage{xspace}
\newcommand{\pind}{\textsc{long}\text{-}\textsc{IP}\xspace}
\newcommand{\pcol}{\textsc{batch-IP}\xspace}
\newcommand{\symD}{\D_\circ}
\newcommand{\aKT}{a_{KT}}

\newcommand{\Fs}{F_{\text{sin}}}
\newcommand{\Fd}{F_{\text{diff}}}

\DeclareMathOperator*{\argmax}{arg\,max}
\DeclareMathOperator*{\argmin}{arg\,min}

\newcommand{\E}{\mathbb{E}}

\let\sign\relax
\DeclareMathOperator{\sign}{sign}

\usepackage[tt=false, type1=true]{libertine}
\usepackage[varqu]{zi4}
\usepackage[libertine]{newtxmath}

\begin{document}

\maketitle

\bigskip
{\footnotesize\tableofcontents}

\newpage
\begin{abstract}
    \begin{center}
        \textbf{\textsf{Abstract}} \smallskip
    \end{center}
AI alignment and participatory design motivate a new democratic design problem: how to collectively choose a \textit{decision rule to use repeatedly}. We study this problem for \textit{linear ranking rules}, which repeatedly rank items $x_j$ within batches \(X=(x_1,\dots,x_m)\in(\mathbb{R}^d)^m\), where each item's ranking is dictated by its score \(\langle \theta^*,x_j\rangle\) according to a fixed scoring vector $\theta^*$. Given voters' preferred scoring vectors \(\theta^{(1)},\dots,\theta^{(n)}\) and their population fractions \(\alpha^{(1)},\dots,\alpha^{(n)}\), we ask how to choose a collective vector \(\theta^*\) satisfying \textit{individual proportionality (IP)}: every voter type $i$ should agree with the resulting rankings to an $\alpha^{(i)}$-proportional degree, either on average over time (\textit{long-run IP}) or even within each batch (\textit{per-batch IP}).

The default rule, the arithmetic mean of the $\theta^{(i)}$, has been shown to be severely majoritarian; more generally, it is not clear that \textit{any} fixed linear rule can balance many voters' disparate opinions. Our main result is that, surprisingly, there \textit{is} a simple rule that does satisfy long-run IP: the \textit{angular mean}, the spherical analog of the arithmetic mean. We then show that exact per-batch IP is impossible for fixed linear rules, but that the gap between per-batch and long-run IP shrinks quickly with batch size. Experiments on three real-world preference datasets show that all rules perform similarly when voters' preferences are homogeneous, while the angular mean substantially improves proportionality in high-disagreement regimes.
\end{abstract}

\section{Introduction} \label{sec:intro}

In AI alignment and participatory design, there is growing interest in the democratic design of \textit{decision rules}: mappings that are repeatedly used to map inputs $X \in \mathcal{X}$ to outputs $Y \in \mathcal{Y}$ over the course of repeated decision-making. Examples include collectively choosing objective functions \citep{shirali2024participatory}, allocation rules for settings such as kidney exchange \citep{freedman2020adapting} and food donation \citep{lee2019webuildai}, and reward or scoring functions used to guide the alignment of language models \citep{DBLP:conf/nips/Ge0MPSV024}. 

Collectively deciding decision rules requires two steps: \textit{eliciting} individuals' preferences, and \textit{aggregating} them into a single decision rule. The literature uses a variety of approaches to this process: classical LLM alignment pipelines pool individual evaluators' raw elicitation data and then fit one collective model \citep{DBLP:conf/nips/Ge0MPSV024, ouyang2022training}; virtual democracy learns a model for every individual agent, has the models vote, and then aggregates these votes with a voting rule \citep{DBLP:conf/icml/KahngLNPP19}; still other work learns an individual model for each agent and then aggregates these models directly into a final model. In order, these approaches utilize increasingly richer information about individuals' beliefs.

We study the aggregation step in the third approach. We focus on the aggregation step because it is comparatively under-studied: there is extensive existing work on eliciting \textit{individuals'} decision rules \citep{freedman2020adapting,lee2019webuildai,DBLP:conf/aies/BoerstlerKCBCHS24,DBLP:conf/aaai/KeswaniCNCHBS26,DBLP:conf/aies/KeswaniCHBS24,mohsin2022learning}, but even in the few papers that consider how to directly aggregate them, the aggregation method is not the main object of analysis, and they default to taking the arithmetic mean of the parameters \citep{noothigattu2018voting, shirali2024participatory}. We focus on the third approach because it is a best-case benchmark: if high-quality aggregation is impossible even with known individual rules, then lower-information pipelines cannot be expected to recover it; if it \textit{is} possible, this approach will identify the structure those pipelines should preserve.

We specifically consider aggregating over \textit{linear ranking} decision rules, related to the paradigm of linear social choice \citep{DBLP:conf/nips/Ge0MPSV024}. The input to linear ranking rules is a batch of $m$ items $X = (x_1,\dots,x_m)$, where each item $x\in X$ is described by a feature vector $x \in \mathbb{R}^d$. A rule, parameterized by a score vector \(\theta\in\mathbb{R}^d\), assigns each item the score \(s_\theta(x)=\langle \theta,x\rangle\) and returns the ranking \(\succ_\theta^X\) induced by these scores. Linear ranking rules are a natural first case: they are interpretable, technically tractable, and capture the score-then-rank structure underlying many prioritization systems for, e.g., resource allocation, recipient selection, and preference-based ranking of language model outputs.

For linear ranking rules, the aggregation task is then: given $n$ voter types, their respective fractional shares of the population $\alpha^{(1)},\dots,\alpha^{(n)} \in [0,1]$, and their respective individual scoring vectors $\theta^{(1)},\dots, \theta^{(n)} \in \R^d$, we want to choose a scoring vector $\theta^*$. This is a non-standard voting problem because the object being chosen is \textit{reusable}: a decision rule is not evaluated only by whether it reflects the electorate at the moment of aggregation, but by how it treats voters across the sequence of decisions on which it is applied. The central worry is therefore cumulative: a fixed rule may repeatedly agree with some viewpoints and repeatedly override others, producing representational inequality over time.
In fact, \citet{feffer2023moral} document this very problem for the literature's default aggregation approach, the arithmetic mean $\thar = \sum_i \alpha^{(i)} \theta^{(i)}$: with just two voter types and two alternatives, the 
arithmetic mean can be inherently majoritarian, in some cases producing rankings that almost never agree with the preferred ranking of the minority group.

This observation motivates the design of \textit{individually proportional} aggregation mechanisms: mechanisms that produce a $\theta^*$ such that if voter type $i$ constitutes an
$\alpha^{(i)}$ fraction of the electorate, the rankings produced by $s_{\theta^*}$ agree with $i$'s preferred ranking on at least an $\alpha^{(i)}$ fraction of item pairs. From the outset, however, it is not at all obvious that \textit{any} fixed scoring rule $\theta^*$ should be able to satisfy such a requirement: because a fixed rule must prioritize the same features in the same way across all decision instances, it could easily fail to balance so many competing viewpoints. This leads to our research question below; surprisingly, we find that its answer is \textit{yes}.

\begin{center}
\begin{minipage}{0.95\linewidth}
\textbf{Research Question:} \textit{Does there exist an aggregation mechanism that, given any $(\theta^{(i)},\alpha^{(i)})_{i \in [n]}$, produces a $\theta^*$ that ranks items proportionally to all voter types $i$ over the long-term?}
\end{minipage}
\end{center}

\subsection{Approach, Contributions, and Related Work}

In \Cref{sec:model}, we set up the formal problem and define \textit{long-run} individual proportionality (IP). This notion is measured \textit{in expectation over batches of items $X$}, where $X = (x_1,\dots,x_m)$ is repeatedly drawn from some fixed absolutely continuous distribution $\mathcal{D}^m$. For a given $X$, voter type $i$'s level of agreement with the induced ranking $\succ_{\theta^*}^X$ 
is measured by the Kendall-Tau (KT) similarity (i.e., the number of pairwise agreements) between $\succ_{\theta^{*}}^X$ 
and $i$'s own ranking $\succ_{\theta^{(i)}}^X$. Voter $i$'s proportional agreement entitlement, on average over $X \sim \mathcal{D}^m$, is $\alphai{m \choose 2}$, so achieving long-run IP means that the following quantity is at least 1, i.e., all voters receive at least their proportional agreement entitlement.
\[ \text{Long-Run IP}(\theta^*) = \min_{i \in [n]} \mathbb{E}_{X \in \mathcal{D}^m} \left[\text{ KT similarity of } \succ_{\theta^*}^X \text{ and } \succ_{\thetai}^X\right] \, / \, \alphai\tbinom{m}{2}.\]
We then formally re-establish \cite{feffer2023moral}'s computationally-derived finding that the arithmetic mean $\thar$ severely violates Long-Run IP. This motivates our first contribution, a new aggregation mechanism.

\paragraph{Contribution 1: The \textit{angular mean} satisfies long-run proportionality.} Our main result is a proof that for generic $(\thetai,\alphai)_{i \in [n]}$, there is a surprisingly interpretable linear ranking rule that is guaranteed to satisfy long-run IP (\Cref{thm:thang_prop}): the \textit{angular mean} $\thang$. It is the angular analog of the arithmetic mean $\thar$, as illustrated by their analogous defining optimization problems (see \Cref{app:model} for proof that the arithmetic mean can indeed be formulated this way):
    \[\thar := \text{argmin}_{\theta} \ \ \frac{1}{n}\sum_{i \in [n]} \alphai \|\theta - \thetai\|^2 \qquad \thang:= \text{argmin}_{\theta} \ \ \frac{1}{n}\sum_{i \in [n]} \alphai \angle(\theta,\thetai)^2\]
We prove $\thang$ satisfies long-run individual proportionality specifically for \textit{spherically symmetric} distributions $\mathcal{D}$, the distribution class also analyzed by \citep{feffer2023moral}; a restriction we interpret in \Cref{sec:model}.

The main challenge of this result is that the angular mean does not have a known closed-form representation nor is its optimization problem necessarily convex, so reasoning about it is non-trivial. The proof relies on a new key theorem (\Cref{thm:angular_bound}), which establishes a fundamental property of the angular mean in $d$ dimensions. This property may be of independent interest, as it generalizes a similar bound in 2 dimensions from the math literature \citep{hotz2015intrinsic}. To our knowledge, this generalization to higher dimensions --- and even its possibility --- was an open question.

As a bonus, we examine when $\thang$'s proportionality is \textit{approximable} by other scoring vectors. This leads to a general robustness result: as $\theta$ approaches $\thang$ in angular distance, its long-run proportionality accordingly converges to that of $\thang$ (\Cref{rem:robustness},\Cref{lem:agreeapprox}). 

\paragraph{Contribution 2: The possibility of \textit{per-batch} proportionality.} We have so far taken as granted that, given $\theta^{(1)},\dots,\theta^{(n)}$, we need to choose a \textit{fixed} and \textit{linear} ranking rule to be applied consistently across inputs $X$. These two restrictions are well-motivated, serving the interpretability of the decision rule by making it explicit which weights will be applied in every case. However, one may worry that this restriction has major costs to achieving proportionality \textit{per-batch}: i.e., instead of making sure that no group is shortchanged on average in the long run, we might want to ensure that within a given batch, we should not expect \textit{any} group to agree too little with the resulting ranking. Technically speaking, we are simply strengthening the definition above by re-ordering the minimum and the expectation in the above definition of Long-Run IP to define $\text{Per-Batch IP}(\theta^*)$.

In contrast to Long-Run IP, we find that even on spherically symmetric distributions,  \textit{no} fixed linear ranking rule can achieve Per-Batch IP of at least 1 (\Cref{thm:fixed-cpop-UB}). However, we show that for any $\theta^*$ and \textit{any} distribution satisfying minimal conditions, Per-Batch IP rapidly approaches Long-Run IP with increasing batch size (\Cref{thm:fixed-cpop-LB}). Hence, our commitment to a consistent and interpretable rule results in surprisingly little degradation of our ability to achieve per-batch proportionality.

\paragraph{Contribution 3: Empirical Results.} We empirically compare five aggregation mechanisms on three real-world preference datasets \citep{awad2018moral,DBLP:conf/aaai/KeswaniCNCHBS26,lee2019webuildai}: the angular mean, arithmetic mean, geometric median, Borda count, and Proportional Sequential Borda (PSB). We find that on the original datasets, voters' scoring vectors are fairly homogeneous, and all rules easily satisfy individual proportionality. We formalize this finding for $\thar$: as voters' scoring vectors converge, the arithmetic mean converges to the angular mean, and thus enjoys near-optimal proportionality in such cases (\Cref{thm:arith_ang_approx}).

We then induce greater voter heterogeneity by subsampling and restricting these datasets. Here, the rules separate sharply: the angular mean substantially improves long-run IP over the arithmetic mean, geometric median, and often Borda. PSB performs best due to its ability to output rankings not inducible by any scoring vector and to adjust its behavior per-batch, outperforming the angular mean and demonstrating a loss associated with restricting to fixed, linear ranking rules. We also demonstrate robustness of the angular mean to non-spherical batch distributions and confirm empirically that each rule's long-run IP and per-batch IP converges quickly as the batch size grows.

\paragraph{Related Work.} In addition to the work cited above, the most closely-related research is on proportionality in repeated decision-making, including work introducing the PSB rule \citep{DBLP:conf/aaai/Lackner20,DBLP:conf/aaai/Kozachinskiy0S25,chandak2024proportional,DBLP:journals/access/BulteauHPRT21,DBLP:conf/aaai/Lackner023,DBLP:journals/corr/abs-2505-22513}. We draw our conceptual notion of proportionality from these papers (i.e., ``an $\alpha$ fraction of the population deserves an $\alpha$ fraction of agreement'') \citep{chandak2024proportional,DBLP:journals/corr/abs-2508-16177}, but they differ fundamentally in that they typically allow the use of non-fixed, non-linear voting rules. Also from this literature comes the objective \textit{Squared Kemeny,} which was shown to have strong proportionality in repeated voting \cite{chandak2024proportional}; interestingly, we show in \Cref{app:squared_kem_approx} that optimizing this objective converges to the same behavior as $\thang$ as the batch size $m$ grows, but it is potentially more majoritarian at small $m$ with respect to long-run IP. 
We give a broader related work overview in \Cref{app:related}.

\section{Model \& Preliminaries}\label{sec:model}

Let there be $d\in \mathbb{N}_{\geq 2}$ dimensions, and let $S^{d-1}$ be the $d$-dimensional unit sphere in $\mathbb{R}^d$. 
For a vector $z\in \mathbb{R}^d$ and $j\in [d]$, we let $z_j$ be the $j$-th entry of $z$. 
The \textit{angular distance} between two vectors $x,y \in \R^d$ is the (minimal) angle between them, $d_{\measuredangle}(x,y) := \cos^{-1}(\langle x, y \rangle/ (\|x\|_2\|y\|_2))$.

\paragraph{Items and Batches.} 
Let an \textit{item} $x \in \mathbb{R}^d$ be a covariate vector, and let $\mathcal{X}\subseteq \R^d$ be the set of feasible covariate vectors.  
A \emph{batch} of items consists of $m \in \mathbb{N}_{\geq 2}$ items, represented as a tuple $X=(x_1,\dots , x_m)\in \mathcal{X}^m$. 
We assume that batches $X$ are formed by sampling $m$ items i.i.d.~from an \emph{item distribution} $\cD\in \Delta(\cX)$ that is absolutely continuous, and has PDF $f_{\cD}\colon \R^d \rightarrow \R$; we  denote  the induced \textit{batch distribution} as $\cD^m$. Additionally, we assume that $f_{\cD}$ is strictly positive and continuous on some open set in $\R^d$.

For some results, we will additionally assume that $\mathcal{D}$ is spherically symmetric, which we will emphasize notationally as $\cD_{\circ}$. The PDF $f_{D_{\circ}}$ of a spherically symmetric distribution $\cD_{\circ}$ is rotationally invariant and hence can be written as a function solely of an item's Euclidean norm, so $f_{D_{\circ}}(x) = g(\|x\|)$ for some function $g\colon \R \rightarrow \R$, i.e., the PDF assigns the same probability to any two items with the same Euclidean norm, regardless of their rotation in space. This assumption is natural in many practical settings; for instance, items representing population percentiles of independent covariates are inherently spherically symmetric.

\paragraph{Scoring Vectors and Rankings.}
A \textit{scoring vector} $\theta\in S^{d-1}$ is a vector of linear weights to be placed on the covariates in $x$. Accordingly, the \textit{score} of $x$ according to $\theta$ is  $s_\theta(x) := \langle x,\theta \rangle$. These scores naturally induce a ranking over the elements of $\mathcal{X}$: Let $x,y \in \mathcal{X}$, and fix some arbitrary tie-breaking ranking over $\mathcal{X}$ denoted as $\succ_{\text{tiebreak}}$. Then, the ranking over $\mathcal{X}$ induced by scoring vector $\theta\in S^{d-1}$, denoted as $\succ_\theta$, is such that for all $x,y \in \mathcal{X}$ (where $x\succ y$ denotes that $x$ is preferred to $y$),
\[x \succ_\theta y \iff s_\theta(x) > s_\theta(y) \lor (s_\theta(x) = s_\theta(y) \land x \succ_{\text{tiebreak}} y).\]

For any $\theta\in S^{d-1}$ and $X\in \mathcal{X}^m$, $\succ_\theta^X$ denotes the  restriction of the ranking $\succ_\theta$ to the items in $X$. 

\paragraph{Voter Types and Aggregation Mechanisms.\ }
There are $n \geq 2$ distinct \textit{types} of voters $1, \dots, n$, where $i$ has an ideal scoring vector $\theta^{(i)} \in S^{d-1}$ and constitutes an $\alpha^{(i)} \in (0,1]$ fraction of the population (so $\sum_{i \in [n]} \alphai = 1$). We let $\boldsymbol{\theta}=(\theta^{(i)})_{i \in [n]}$ and $\boldsymbol{\alpha}=(\alpha^{(i)})_{i \in [n]}$, and we call $(\boldsymbol{\theta},\boldsymbol{\alpha})$ a \textit{profile}.
Let $P$ be the set of all possible profiles, and $\profset{n, d} \subseteq P$ be all those with $n$ voter types in $d$ dimensions.

An \textit{aggregation mechanism} is a mapping that takes in a profile $(\boldsymbol{\theta},\boldsymbol{\alpha})$  and outputs some $\theta \in S^{d-1}$. A generic aggregation mechanism is denoted as $\theta_f$, so $\theta_f(\boldsymbol{\theta},\boldsymbol{\alpha})$ is the scoring vector produced by some underlying function $f$.\footnote{We remark that requiring the output vector to be normalized (i.e., in $S^{d-1}$) is without loss of generality, since rescaling a scoring vector does not change the rankings it induces.}  
Given a profile $(\boldsymbol{\theta}, \boldsymbol{\alpha}) \in \profset{n,d}$, the two aggregation mechanisms we focus on --- the  (normalized) \emph{arithmetic mean} $\thar$ and \textit{angular mean} $\thang$ --- are respectively defined as
 \[
    \theta_\textrm{arith}(\boldsymbol{\theta}, \boldsymbol{\alpha}) = \argmin_{\theta \in S^{d-1}} \sum_{i=1}^n \alphai \|\theta-\theta^{(i)}\|_2^2 \qquad  \theta_\textrm{ang}(\boldsymbol{\theta}, \boldsymbol{\alpha}) = \argmin_{\theta \in S^{d-1}} \sum_{i=1}^n \alphai d_{\measuredangle}(\theta, \theta^{(i)})^2.
\]
We show that this formulation of $\thar$ is equivalent to the direct arithmetic mean formulation in \Cref{prop:arith_mean_def}. If there are multiple global minimizers for either function, we assume some fixed tie-breaking scheme is applied. 

\vspace{0.2em}
\begin{remark}[\textit{Computability of Angular Mean}]\label{rem:computability_thang}
   How to compute the angular mean is a well-studied problem \citep{buss2001spherical, afsari2013convergence, charlier2013necessary, cazals2021frechet}. Because it has no known closed-form representation, computation generally relies on gradient-based approximation schemes. The objective is not necessarily convex so convergence is not guaranteed \citep{buss2001spherical, afsari2013convergence}; however, in practice, this is often addressed by the fact that the arithmetic mean often serves as a good initial point for optimization (cf. \Cref{app:exp}), in which case there exist strong theoretical guarantees that gradient-based methods will converge \citep{buss2001spherical, afsari2013convergence}.
\end{remark}

\paragraph{\textit{Long-Run} Individual Proportionality.} \
In words, an aggregation mechanism $\theta_f$ is \textit{long-run} individually proportional if, over many random draws of $X \sim \mathcal{D}^m$, voter type $i$ ``agrees with'' at least an $\alphai$ fraction
of pairwise comparisons in the collective ranking $\succ_{\theta_f}^X$.
``Agreement'' is formalized as Kendall-Tau (KT) agreement. For two scoring vectors $\theta, \psi \in S^{d-1}$, their induced rankings on a batch $X \in \cX^m$ have \textit{KT-agreement}
\begin{center}
$a_{KT}(\theta, \psi, X) := |\{x,y\} \in \binom{X}{2} \colon (x \succ_{\theta} y \land x \succ_{\psi} y) \lor (y \succ_\theta x \land y \succ_\psi x)|$. 
\end{center}
An individual's proportional \textit{entitlement} of KT-agreement is then equal to $\alphai\binom{m}{2}$, i.e., an $\alphai$ fraction of all ${m \choose 2}$ pairwise comparisons in a ranking over $m$ items.

For an aggregation mechanism $\theta_f$, a profile $\prof \in P_{n,d}$, a batch $X \in \cX^m$, and voter type $i \in [n]$, \textit{$i$'s individual proportionality (IP) level} is
\[\agi(\theta_{f},X,\prof) := \frac{a_{KT}(\thetai, \theta_f\prof, X)}{\alphai \binom{m}{2}}.\]
Then, the \textit{Long-Run IP Level} (\textsc{Long-IP}) is the minimum \textit{expected} individual proportionality level received by any voter type. 
\begin{definition}[\textbf{Long-Run IP Level}]
    For any batch size $m$, item distribution $\cD$, and profile $\prof \in P_{n,d}$, an aggregation mechanism $\theta_f$ has 
    \begin{align*}
        \text{\textit{Long-Run IP Level}} \ \ \ \ \ \ \ &\pind(\theta_f, \cD^m, \prof) := \min_{i \in [n]} \E_{X \sim \D^m}\left[ \agi(\theta_f, X,\prof)\right]\\
        \text{and \textit{Worst-Case Long-Run IP Level}} \ \ \ \ \ \ \ & \pind^*(\theta_f, \cD^m) := \inf_{\prof \in P} \pind(\theta_f,\cD^m, \prof).
    \end{align*}
\end{definition}
Given that we aim to bound $\pind^*$, we now characterize what is possible: we show that the most proportional aggregation mechanism can achieve $\pind^*$ of no more than 1, and the least can achieve $\pind^*$ no less than 0. The lower bound is by nonnegativity of $\agi$; see \Cref{app:model} for proof of the upper bound. Accordingly, we say that $\theta_f$ \textit{satisfies} \textsc{Long-IP} iff \ $\pind^*(\theta_f, \cD^m) = 1.$ This means that in every profile, $\theta_f$ gives every voter type at least their proportional entitlement.
\prooflink{propindrange}
\begin{restatable}{proposition}{propindrange}\label{prop:propind_range}
    For any $\theta_f, m,$ and $\cD$, $\pind^*(\theta_f, \cD^m) \in [0,1]$.
\end{restatable}

\subsection{Preliminaries}
We now establish that the arithmetic mean has the worst possible worst-case Long-Run IP Level: 
\prooflink{arithmeanpind}
\begin{restatable}[Figure 2, \cite{feffer2023moral}]{theorem}{arithmeanpind}\label{thm:thar_bad}
For any $m$ and item distribution $\cD$, \ $\pind^*(\thar, \cD^m) = 0$.
\end{restatable}
The proof (\Cref{app:model}) is a simple profile with $n=2$, where the majority ($\alpha^{(1)} > 1/2$) and minority ($\alpha^{(2)} < 1/2$) types have diametrically opposed scoring vectors, and thus the arithmetic mean overlaps perfectly with $\theta^{(1)}$. We name this profile The Antipodal Construction for later use (\Cref{fig:antipodal-prof}). 

\begin{figure}[htbp]
    \centering
    \begin{subfigure}[t]{0.49\textwidth}
        \centering
        \begin{tikzpicture}[scale=1.2]

        \draw[black] (0,0) circle (1);

        \fill (0,0) circle (0.02);

        \coordinate (theta1) at (180:1);
        \coordinate (theta2) at (0:1);

        \draw[->, thick] (0,0) -- (theta1);
        \draw[->, thick] (0,0) -- (theta2);

        \node[left] at (theta1) {$\theta^{(1)}$};
        \node[right] at (theta2) {$\theta^{(2)}$};

        \end{tikzpicture}
        \caption{The Antipodal Construction}
        \label{fig:antipodal-prof}
    \end{subfigure}
    \hfill 
    \begin{subfigure}[t]{0.49\textwidth}
        \centering
        \input{.//Figures/angle_agreement}
        \caption{Geometric interpretation of agreement}
        \label{fig:prob_agree}
    \end{subfigure}

\end{figure}

In addition to this critique of the arithmetic mean, \citet{feffer2023moral} also observe a strong relation between two scoring vectors' angular distance and their expected KT-agreement on batches drawn from spherically symmetric distributions when $m =2$:
\begin{proposition}[Proposition 1, \cite{feffer2023moral}]\label{prop:angle_eprop} For any spherically-symmetric item distribution $\D_{\circ}$ and any $\theta,\psi \in S^{d-1}$, we have
    $\E_{X \sim \D_{\circ}^2}[\aKT(\theta, \psi, X)] = \nicefrac{\pi - \dang(\theta, \psi)}{\pi}.$
\end{proposition}
The intuition for the argument will be useful, so we summarize it here. Let $X = (x_1,x_2)$ and define the difference vector $\tilde{x} = x_1 - x_2$. Note that $\succ_{\theta}$ and $\succ_{\psi}$ agree on the ranking of $x_1$ and $x_2$ iff $\sign(\langle \theta, \tilde{x}\rangle) = \sign(\langle \psi, \tilde{x}\rangle)$. Geometrically (see \Cref{fig:prob_agree}), note that the set $H_>^\theta = \{y : \sign(\langle \theta, y \rangle) > 0\}$ forms a halfspace, and so does the similarly defined $H_>^\psi$. The intersection $H_>^\theta \cap H_>^\psi$ then forms a cone representing the set of $\tilde{x}$ for which $\theta$ and $\psi$ agree that $x_1 \succ x_2$. Via the symmetric argument for when $\theta,\psi$ agree that $x_2 \succ x_1$, it follows that the set of all $\tilde{x}$ on which $\theta$ and $\psi$ agree forms a double-wedge. The intersection of that double-wedge with any sphere centered at $0$ takes up a $\nicefrac{\pi - \dang(\theta,\psi)}{\pi}$ fraction of the sphere's total surface area. Finally, observe that the expected agreement between $\theta$ and $\psi$ is equal to the probability of $\tilde{x}$ landing in the double-wedge. By spherical symmetry, $\tilde{x}$ itself is drawn from a spherically symmetric distribution $\tilde{\D}_\circ$, and this probability is thus equal to the size of the double-wedge, concluding the proof. 

\section{The Angular Mean Satisfies Long-Run Individual Proportionality}\label{sec:ang_mean_prop}

The dramatic disproportionality of $\thar$ motivates the question of whether there exists \textit{any} fixed linear decision rule satisfying \textsc{Long-IP}. 
We now answer this in the affirmative for the angular mean $\thang$, assuming a spherically symmetric item distribution. 

\begin{restatable}[\textbf{Main Result}]{theorem}{thangprop}\label{thm:thang_prop}
For any batch size $m$ and spherically symmetric item distribution $D_\circ$, 
\[ \pind^*(\thang, \symD^m) = 1.\]
\end{restatable}

The first step of the proof is to extend \Cref{prop:angle_eprop} to the general $m$-item case, which follows by the intuition above along with linearity of expectation (full proof in \Cref{app:proportionality}):
\prooflink{agreementangle}
\begin{restatable}{lemma}{agreementangle}\label{lem:ktangle}
    For any batch size $m$, spherically symmetric item distribution $D_\circ$, and $\theta, \psi \in S^{d-1}$, 
    \begin{center}
        $\E_{X \sim \D_\circ^m}\big[\aKT(\theta, \psi, X)\big] = \frac{\pi - \dang(\theta, \psi)}{\pi} \binom{m}{2}.$
    \end{center}
\end{restatable}
Plugging in $\thetai$ and $\thang$ for $\theta$ and $\psi$ and writing our desired bound on $\pind^*$ in terms of $\E_{X \sim \D_\circ^m}\big[\aKT(\theta, \psi, X)\big]$, 
we deduce that it suffices to show that $\dang(\theta^{(i)}, \thang) \leq (1 - \alphai) \pi$ for all $i \in [n]$. We now show that this fundamental property of the angular mean holds, i.e.,  the angular mean is guaranteed to simultaneously lie within a $(1-\alphai)\pi$ angle from any $\theta^{(i)}$. To our knowledge, this was unknown, and generalizes the known version on $S^1$ \citep{hotz2015intrinsic}. The full proof is located in \Cref{app:proportionality}.

\prooflink{angbound}
\begin{restatable}[\textbf{Fundamental Property of Angular Mean}]{theorem}{angbound}\label{thm:angular_bound}
    For any profile $(\btheta, \balpha) \in P_{n,d}$,
    \[\dang(\thang\prof, \thetai) \leq (1 - \alphai)\pi \ \ \text{ for all } i \in [n].\]
\end{restatable}
\vspace*{-1em}\begin{proof}[Proof Sketch.] Let $F(\theta) := \sum_{i=1}^n \alphai \dang(\theta, \thetai)^2$ be the value of the angular mean objective for a generic $\theta$. Suppose for contradiction that there exists $i \in [n]$ such that $\dang(\thang, \thetai) > (1 - \alphai)\pi$. We will find another score vector $\theta'$ such that $F(\theta') < F(\thang)$, contradicting the optimality of $\thang$.

We illustrate the intuition with \Cref{fig:angular_bound}. To construct $\theta'$, let $G_i$ be the unique great circle passing through $\thetai$ and $\thang$ (left image). Then, let $\theta'$ be positioned on the great circle $G_i$, but swooped around to be within a $(1-\alpha_i)\pi$ distance of $\theta^{(i)}$ but on the other side. We want to show that
\[
  \textstyle  F(\theta') - F(\thang) = \sum_{j \in [n]} \alpha^{(j)} (\dang(\theta',\theta^{(j)})^2 - \dang(\thang,\theta^{(j)})^2 ) < 0.
\]
The $j = i$ term of this sum will be handled at the end. 
We first bound $\dang(\theta',\theta^{(j)})^2$ for each $j \neq i$, corresponding to bounding the angular length of the yellow arc (left image).
We do this by applying Toponogov's inequality and the Euclidean law of cosines, which together give a  spherical analog of the Pythagorean Theorem showing that the yellow arc is upper bounded by its Euclidean counterpart when the other two sides of the triangle are preserved (middle image). 
As will be important in the next step, one term of this bound tracks whether the move from $\thang$ toward $\theta'$ initially moves toward or away from $\theta^{(j)}$, as captured by $\cos(\gamma_j)$ (see diagram for $\gamma_j$).

With a bound on $\dang(\theta',\theta^{(j)})^2$ for each $j$, we then take the final weighted sum over voters to bound $F(\theta') - F(\thang)$ as above. Here we add one final key ingredient: the local optimality of $\thang$. Because $F$ has zero derivative at $\thang$ in the direction of the move from $\thang$ to $\theta'$, the weighted sum of the directional terms involving $\cos(\gamma_j)$, as carried through our bound above, sum to 0 (right image). The remaining terms then reduce to an expression involving only the move length $\dang(\thang,\theta')$, the original distance $\dang(\thang,\thetai)$, and $\alphai$, which is negative by the choice of $\theta'$ and the assumed violation $\dang(\thang,\thetai)>(1-\alphai)\pi$. We conclude
$F(\theta')< F(\thang)$, our contradiction.
\end{proof}

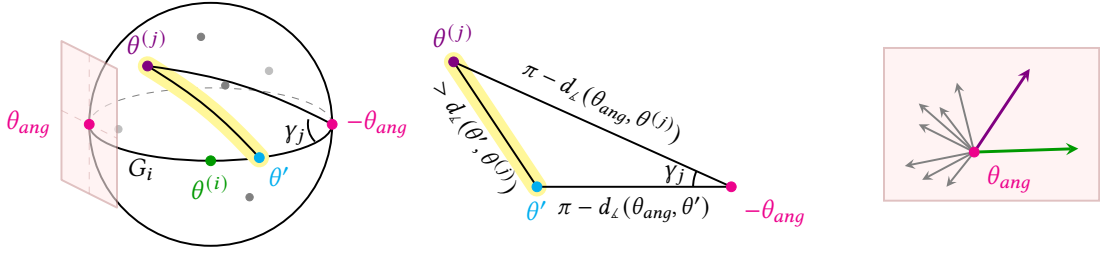
\begin{figure}[htbp]
    \vspace{-1em}
    \centering
    \resizebox{0.9\linewidth}{!}{
    \input{.//Figures/proportionality_proof}
    }
    \caption{Diagrams for proof sketch of \Cref{thm:angular_bound}.}
    \label{fig:angular_bound}
\end{figure}

\begin{remark}[\textbf{Robustness}] \label{rem:robustness}
    A natural question is whether a score vector \textit{close to} $\thang$ will exhibit similar proportionality. The answer is yes: if a score vector is within angular distance $\gamma$ of $\thang$, each $i$'s IP level can decrease (additively) by at most $\gamma/(\pi\alpha^{(i)})$ (\Cref{lem:agreeapprox}).
\end{remark}

\section{The Possibility of Per-Batch Individual Proportionality}\label{sec:cost-of-consistency}

So far, we have shown that $\thang$ ensures that voters will receive their proportional entitlement \textit{in expectation over batches}, and thus on average over sufficient rounds by the Law of Large Numbers.
However, one may also want a stronger guarantee: proportionality on \textit{every item batch}. The strongest notion would guarantee that for any $\prof$ with $n$ voter types and any $m \in \mathbb{N}_{\geq 2}$, $\theta_f$ is such that \[ \textstyle
\min_{X \in \cX^m} \min_{i \in [n]}  \agi(\theta_f, X, \prof) \geq 1.\]
One can immediately see that this too strong for simple divisibility reasons: e.g., consider the Antipodal Construction with $m=2$, where any deterministic rule must output the same ranking each time, thereby giving one group 0 agreement. 
A similar construction can be used to show that this issue can occur even without divisibility issues (i.e., even when all voters' entitlements \(\alpha_i\binom{m}{2}\) are integral): simply let $m$ be large but make the items in the batch collinear, and so again, one group in the Antipodal Construction must end up with 0 agreement (\Cref{ex:batch-prop-impossibility}).

In light of this impossibility, we pursue a weakened version of per-batch proportionality, which takes the \textit{expectation} over item batches instead of the \textit{minimum}:

\begin{definition}(\textbf{Per-Batch IP Level})
For any batch size $m$, item distribution $\cD$, and profile $\prof \in P_{n,d}$, an aggregation function $\theta_f$ has \textit{per-batch proportionality level} of
\[\textstyle \pcol(\theta_f, \cD^m, \prof) := \E_{X \sim \D^m}\left[ \min_{i \in [n]}\agi(\theta_f, X,\prof)\right].\]
\end{definition}
The expectation allows us to avoid over-indexing on the worst-case item batches, while moving the minimum over $i \in [n]$ inside the expectation ensures that no voter type is shortchanged too frequently within each batch (rather than just on average over batches).
The following proposition, proved in \Cref{app:subsec-prop-relationship}, formalizes this intuition that $\pcol$ is a strict strengthening of $\pind$.

\prooflink{proprelationship}
\begin{restatable}{proposition}{proprelationship}\label{prop:prop_relationship}
    For any $\theta_f, \cD^m, \prof$, \  $
    \pcol(\theta_f, \cD^m, \prof) \leq \pind(\theta_f, \cD^m, \prof)$.
\end{restatable}
 Also note that this result --- and all in this section unless otherwise stated --- holds for \textit{all} item distributions (meeting the regularity conditions in \Cref{sec:model}), not just spherically symmetric ones.

Unfortunately, unlike $\pind$, we show that there exists \textit{no} fixed linear ranking rule that can perfectly satisfy $\pcol$. This can be shown trivially via divisibility issues, but we prove a stronger impossibility requiring that in the constructed profile, all voter type entitlements are integral: 
\prooflink{thmcpopUB}
\begin{restatable}[\textbf{Impossibility of $\pcol$}]{theorem}{thmcpopUB}\label{thm:fixed-cpop-UB}
    For any batch size $m \geq 3$\footnote{Note that for $m=2$, there is no profile satisfying the integral proportionality entitlement condition.}, there exists a profile $\prof \in P_{n,d}$ with $\alphai \binom{m}{2} \in \mathbb{Z}$ $ \forall i \in [n]$ such that, for any $\theta_f$ and any $\cD$,  \ $\pcol(\theta_f, \cD^m,\prof) < 1.$
\end{restatable}
We prove this result in \Cref{app:subsec-cpop-UB} via a similar intuition as used above for \Cref{ex:batch-prop-impossibility}: we show that on batches $X$ that are \textit{nearly} collinear (and thus occur with nonzero probability), at least one voter type receives $\agi(\theta_{f},X,\prof)  = 0$. We then upper-bound this type's individual proportionality level in all other batches, demonstrating the bound. 
\vspace{0.2em}
\begin{remark}
    This impossibility is universal only over fixed linear ranking rules: the Proportional Sequential Borda voting rule introduced by \citet{DBLP:journals/corr/abs-2508-16177} is guaranteed to achieve $\pcol=1$ on all profiles with integer entitlements.
\end{remark}

The bad news is that the above remark concretely establishes a proportionality cost of restricting to fixed linear ranking rules. The good news, which we show next, is that this cost is surprisingly small and diminishes quickly in the batch size $m$. Intuitively, the $\pcol$ improves in $m$ because as $m$ grows, the likelihood that all $x \in X$ are collinear decreases, causing $\pcol$ to approach $\pind$. The full proof with explicit constants is in \Cref{app:subsec-cpop-LB}. 
\prooflink{thmcpopLB}
\begin{restatable}[\textbf{Main Result}]{theorem}{thmcpopLB}\label{thm:fixed-cpop-LB}
    For all $\theta_f, m, \mathcal{D}$, and $\prof \in P_{n,d}$, 
    \[ \textstyle
    \pcol(\theta_f, \cD^m,\prof) \geq \pind(\theta_f, \cD^m, \prof) - \frac{1}{\alpha_{min}} \cdot O\left(\sqrt{\frac{\min\{\log(n),d\log(m/d)\}}{m}}\right).
    \]
\end{restatable}
\begin{proof}[Proof Sketch]
The fundamental goal in this proof is to show that $\agi(\theta_f, X,\prof)$ concentrates around its expectation over $X \sim \cD^m$ as $m$ grows, in which case by definition $\pcol$ converges to $\pind$. We do this by first rewriting $i$'s IP level as
\[\textstyle \agi(\theta_f, X,\prof) = \frac{1}{\alpha^{(i)}\binom{m}{2}}\sum_{1\leq j < \ell\leq m} \mathbb{I}\left(\text{sign}(\langle \thetai, x_j- x_{\ell}\rangle) = \text{sign}(\langle \theta_f\prof, x_j- x_{\ell}\rangle)\right),
    \]
the intuition being that agreement between $\thetai$ and $\thang$ on a pairwise comparison of items $x_j,x_\ell$ just comes down to whether the two scoring vectors have the same signed dot product with the item difference vector. This is now the sum of indicators, which is ripe for a concentration bound except that these indicators have a particular kind of dependence: while the $x_j$ are individually independent, their \textit{differences} are not. To handle this dependence, we reframe individual proportionality level as a U-statistic: we rewrite this sum as the average of $m!$ sums of indicators acting on $\lfloor m/2 \rfloor$ mutually exclusive pairs of items, which lets us show that the sum is sub-Gaussian with variance inversely proportional to $\lfloor m/2 \rfloor$. We use this setup to get both convergence rates in the minimum: The $\log(n)$ rate comes from taking the maximum over the $n$ sub-Gaussian random variables corresponding to $\agi$ for each $i \in [n]$, and the $d$-dependent bound is via the Rademacher complexity, the intuition being that in low dimensions, even an arbitrary number of voters must be ``close.''
\end{proof}
\vspace{-0.5em}
As a byproduct of this proof, we are also able to give tail bounds showing that each $i$'s IP level is close to its expectation with high probability (\Cref{cor:concentration-bound}). We can additionally apply this bound to show that $\thang$ converges to the optimal $\pcol$ as $m$ grows, noting that by \Cref{prop:prop_relationship}, $\pcol$ of 1 is the best we can hope for. 

\begin{corollary}
    For spherically symmetric $\cD_\circ$,  \Cref{thm:thang_prop} gives that $\pind(\thang, \cD^m_\circ, \prof) = 1$; \Cref{thm:fixed-cpop-LB} then immediately implies that
    \[\pcol(\thang, \cD^m_\circ,\prof) \xrightarrow{m \to \infty} 1 \quad \text{\upshape{  (with a known convergence rate for finite $m$)}}.\]
\end{corollary}

\begin{table}[t]
\centering
\caption{Heterogeneity of voter types' scoring vectors and pairwise angular distances between aggregation rules. All values are in degrees.}
\label{tab:heterogeneity-main-t}
\begin{tabular}{l|cc|cc|cc}
 & \multicolumn{2}{c|}{\textsc{MMachine}} & \multicolumn{2}{c|}{\textsc{FoodRescue}} & \multicolumn{2}{c}{\textsc{Kidney}} \\
 & Orig. & 2D & Orig. & 2D & Orig. & 2D \\[0.25em]
\hline\hline
\multicolumn{7}{l}{\textbf{Heterogeneity of $\theta^{(i)}$}} \\
max pairwise angle between $\theta^{(i)}$ & 53.5 & 180.0 & 109.6 & 178.8 & 147.3 & 180.0 \\
avg of pairwise angular distances between $\theta^{(i)}$ & 20.7 & 80.4 & 61.7 & 85.7 & 57.1 & 77.5 \\
stdev of pairwise angular distances between $\theta^{(i)}$ & 8.3 & 52.7 & 17.9 & 56.7 & 20.6 & 52.5 \\
\hline
\multicolumn{7}{l}{\textbf{Rule coincidence}} \\
$d_\measuredangle(\theta_{\mathrm{ang}}, \theta_{\mathrm{arith}})$ & 0.03 & 9.70 & 1.10 & 21.86 & 1.18 & 2.85 \\
$d_\measuredangle(\theta_{\mathrm{ang}}, \theta_{\mathrm{med}})$ & 0.93 & 15.57 & 4.23 & 32.93 & 5.69 & 5.62 \\
$d_\measuredangle(\theta_{\mathrm{arith}}, \theta_{\mathrm{med}})$ & 0.91 & 5.87 & 3.16 & 11.07 & 4.54 & 2.76 \\
\end{tabular}
\end{table}

\section{Empirical Results}
\label{sec:empirical}

Finally, we compare aggregation mechanisms on
real and semi-synthetic data. Wherever an experiment requires
an expectation $\mathbb{E}_{X \sim \D^m}[\,\cdot\,]$, we estimate it
by Monte Carlo: we draw $2000$ batches
 independently and average. Unless otherwise stated, we use batch size $m = 10$ and sample each batch's items i.i.d.~uniformly from $S^{d-1}$, so $\D_\circ^m$ is spherically
symmetric. We test robustness to non-spherically symmetric batch distributions in \Cref{app:acg};  $\thang$ remains proportional throughout, highlighting its appeal beyond the regime covered in \Cref{sec:ang_mean_prop}.

\paragraph{Rules.} We consider three fixed linear ranking rules: $\thar, \thang$, and the \emph{geometric median}
$\theta_{\mathrm{med}}:= \argmin_{\theta\in S^{d-1}} \sum_i \alpha^{(i)}\,\lVert\theta- \theta^{(i)}\rVert_2$, where the median is a reasonable candidate for its insensitivity to outliers.
We also consider two rules that are not linear or fixed, and thus may circumvent the per-batch impossibility
of Theorem~\ref{thm:fixed-cpop-UB}. 
The first is \emph{Borda}, the
canonical aggregator of predicted rankings in virtual democracy
\citep{DBLP:conf/icml/KahngLNPP19,lee2019webuildai}: given a
batch $X$, each voter type $i$ awards $\alpha^{(i)}(m-j)$ points to the
item ranked at position $j$ of $\succ^{X}_{\theta^{(i)}}$;  Borda then returns a ranking over items by total score. 
The
second is \emph{Proportional Sequential Borda} (PSB), introduced by
\citet{DBLP:journals/corr/abs-2508-16177}, which iteratively selects
the Borda winner and then reduces each voter type's weight in proportion to its contribution to the winner's score. PSB satisfies a strong batch-level guarantee: on any $X$, every $i$ agrees with the computed ranking on at least $\lfloor \alpha^{(i)} \binom{m}{2}\rfloor$ pairwise comparisons.

\begin{table*}[t]
\centering
\footnotesize
\setlength{\tabcolsep}{5pt}
\setlength{\extrarowheight}{2pt}
\caption{Proportionality levels in original data. $m = 10$. Cell shading is proportional to $\log_{10}(\text{value})$.}
\label{tab:overview-full-t}
\begin{tabular}{c l>{\centering\arraybackslash}p{0.65cm}>{\centering\arraybackslash}p{0.65cm}>{\centering\arraybackslash}p{0.65cm}>{\centering\arraybackslash}p{0.65cm}>{\centering\arraybackslash}p{0.65cm}@{\hspace{14pt}}c l>{\centering\arraybackslash}p{0.65cm}>{\centering\arraybackslash}p{0.65cm}>{\centering\arraybackslash}p{0.65cm}>{\centering\arraybackslash}p{0.65cm}>{\centering\arraybackslash}p{0.65cm}}
\toprule\\[-1.75em]
& dataset & arith. & ang. & med. & Borda & PSB & & dataset & arith. & ang. & med. & Borda & PSB \\
\midrule
\multirow{3}{*}{\rotatebox[origin=c]{90}{\textsc{long-IP}}}
& \textsc{FoodRescue} & \cellcolor{black!22}11  & \cellcolor{black!22}11  & \cellcolor{black!22}11  & \cellcolor{black!22}11  & \cellcolor{black!22}11  &
\multirow{3}{*}{\rotatebox[origin=c]{90}{\textsc{batch-IP}}}
& \textsc{FoodRescue} & \cellcolor{black!21}9.7 & \cellcolor{black!21}9.8 & \cellcolor{black!21}9.6 & \cellcolor{black!21}9.9 & \cellcolor{black!21}10  \\
& \textsc{Kidney}     & \cellcolor{black!33}134 & \cellcolor{black!33}135 & \cellcolor{black!33}132 & \cellcolor{black!33}133 & \cellcolor{black!33}135 &
& \textsc{Kidney}     & \cellcolor{black!32}112 & \cellcolor{black!32}112 & \cellcolor{black!32}110 & \cellcolor{black!32}112 & \cellcolor{black!32}115 \\
& \textsc{MMachine}   & \cellcolor{black!32}111 & \cellcolor{black!32}111 & \cellcolor{black!32}111 & \cellcolor{black!32}111 & \cellcolor{black!32}111 &
& \textsc{MMachine}   & \cellcolor{black!32}98  & \cellcolor{black!32}98  & \cellcolor{black!32}98  & \cellcolor{black!32}98  & \cellcolor{black!32}98  \\
\bottomrule
\end{tabular}
\end{table*}

\begin{table*}[t]
\centering
\footnotesize
\setlength{\tabcolsep}{5pt}
\setlength{\extrarowheight}{2pt}
\caption{Proportionality levels in 2D data. $m = 10$. Cell shading is proportional to $\log_{10}(\text{value})$.}
\label{tab:overview-minimal-body}
\resizebox{\linewidth}{!}{\begin{tabular}{c l>{\centering\arraybackslash}p{0.65cm}>{\centering\arraybackslash}p{0.65cm}>{\centering\arraybackslash}p{0.65cm}>{\centering\arraybackslash}p{0.65cm}>{\centering\arraybackslash}p{0.65cm}@{\hspace{14pt}}c l>{\centering\arraybackslash}p{0.65cm}>{\centering\arraybackslash}p{0.65cm}>{\centering\arraybackslash}p{0.65cm}>{\centering\arraybackslash}p{0.65cm}>{\centering\arraybackslash}p{0.65cm}}
\toprule\\[-1.75em]
& dataset & arith. & ang. & med. & Borda & PSB & & dataset & arith. & ang. & med. & Borda & PSB \\
\midrule
\multirow{3}{*}{\rotatebox[origin=c]{90}{\textsc{long-IP}}}
& \textsc{FoodRescue-2D} & \cellcolor{black!14}1.9 & \cellcolor{black!15}2.8 & \cellcolor{black!9}0.7  & \cellcolor{black!13}1.7 & \cellcolor{black!19}5.7 &
\multirow{3}{*}{\rotatebox[origin=c]{90}{\textsc{batch-IP}}}
& \textsc{FoodRescue-2D} & \cellcolor{black!14}1.9 & \cellcolor{black!15}2.6 & \cellcolor{black!9}0.7  & \cellcolor{black!13}1.7 & \cellcolor{black!18}5.4 \\
& \textsc{Kidney-2D}     & \cellcolor{black!6}0.3  & \cellcolor{black!16}3.3 & \cellcolor{black!18}4.5 & \cellcolor{black!22}11  & \cellcolor{black!31}76  &
& \textsc{Kidney-2D}     & 0.1                    & \cellcolor{black!14}2.1 & \cellcolor{black!15}2.6 & \cellcolor{black!19}6.6 & \cellcolor{black!30}66  \\
& \textsc{MMachine-2D}   & \cellcolor{black!6}0.4  & \cellcolor{black!14}1.9 & \cellcolor{black!10}0.9 & \cellcolor{black!18}4.5 & \cellcolor{black!26}30  &
& \textsc{MMachine-2D}   & \cellcolor{black!4}0.2 & \cellcolor{black!10}1.0 & \cellcolor{black!7}0.5  & \cellcolor{black!15}2.6 & \cellcolor{black!26}26  \\
\bottomrule
\end{tabular}}
\end{table*}

\paragraph{Warm-up: Two-Voter Case.} In \Cref{app:2voters}, we run synthetic experiments in the $n = 2$ case. As expected, we find that as $\dang(\theta^{(1)},\theta^{(2)})$ approaches 180$^\circ$, $\thar$ (along with $\theta_{\mathrm{med}}$ and \textit{Borda}) become extremely majoritarian, while $\thang$ and PSB maintain expected $\agi$ of at least 1 (\Cref{fig:exp2-synth}).

\paragraph{Real-World Data.} We study three main datasets (with two extra in \Cref{app:kid}), each including a learned $\theta^{(i)} \in S^{d-1}$
per voter (see \Cref{app:dd} for details). They are \textit{Moral Machine} (\textsc{MMachine})~\citep{noothigattu2018voting,awad2018moral} ($n = 137$, $d = 24$); \textit{Kidney Exchange} (\textsc{Kidney})~\citep{DBLP:conf/aaai/KeswaniCNCHBS26} ($n = 404$, $d = 8$); and 412 Food Rescue (\textsc{FoodRescue})~\citep{lee2019webuildai} ($n = 19$, $d = 7$). We use uniform voter weights
$\alpha^{(i)} = 1/n$, and show that this choice is robust to randomly varied weights in \Cref{app:weights}. 

\textit{Results.} In this data, we find that voters' score vectors are all highly similar (see \Cref{tab:heterogeneity-main-t}): in degrees, the \textit{mean, standard deviation, and max} of the pairwise angles between any two $\theta^{(i)}$ are $21,8.3, 54$ in \textsc{MMachine}; $62, 17.9, 108$ in \textsc{FoodRescue}; and $57, 21, 147$ in \textsc{Kidney}. Empirically, this translates to all three linear ranking rules $\thang,\thar,\theta_{\mathrm{med}}$ being highly coincident, falling within 1, 5, and 6 degrees of one another (pairwise) across the respective datasets. We formally explain this finding in \Cref{thm:arith_ang_approx}, proving that coincident $\thetai$'s lead to coincident $\thar$ and $\thang$. Per our robustness results in \Cref{rem:robustness}, this means that these three rules --- and in fact, empirically all five --- have almost identical proportionality behavior on the raw datasets (see \Cref{tab:overview-full-t}).

\paragraph{2D Data.} To investigate what happens under greater heterogeneity --- the regime where the rule choice becomes nontrivial --- we construct \textit{2D} versions of the raw datasets
(denoted \textsc{MMachine-2D}, etc.) by retaining only the pair of
features that maximize the resulting variance of the pairwise angular
distances between voters. 
This produces voter populations that span a wide arc of $S^1$ (see \Cref{fig:exp9-subset-viz}). We report the resulting (larger) angular distances among the $\theta^{(i)}$'s and $\thang,\thar,\theta_{\mathrm{med}}$ in \Cref{tab:heterogeneity-main-t}.

\textit{Results. }\Cref{tab:overview-minimal-body} reports \pind (left) and \pcol (right) across rules and datasets. We see that $\thar$ and $\theta_{\mathrm{med}}$ both violate voters' proportional entitlements substantially; on \textsc{Kidney-2D} the
\pind of the arithmetic mean drops as low as $0.3$, and expectedly, the \pcol is at $0.1$ even lower. In contrast, $\thang$ satisfies both $\pcol$ and $\pind$ (i.e., attains levels of at least 1) on all three datasets. As expected, the more powerful non-fixed, non-linear rules perform better, with PSB being 3-10x better than Borda. Looking across the two tables, one may also notice that $\pcol$ and $\pind$ are highly similar, as consistent with \Cref{thm:fixed-cpop-LB}. \Cref{fig:convergence} shows just how quickly $\pcol$ approaches $\pind$, empirically: for most rules, convergence happens at extremely small $m$. As expected, these quantities converge the slowest for the non-fixed, non-linear rules, as these rules are making batch-level adjustments. 

\paragraph{Heterogeneity-Filtered 2D Data.} For a final stress test, we subsample voters within each 2D dataset to filter for further heterogeneity: for various $n$, we draw $n$ voters uniformly at random and accept the sample only if the standard deviation of pairwise angular distances in the subsample is $\geq65^\circ$. This creates instances with few, highly separated voter \textit{clusters}  --- the hardest regime for proportionality.

\textit{Results.} \Cref{fig:subsample} shows $\pind$ on these filtered subsamples as $n$ grows (per-batch IP
behaves similarly; see \Cref{fig:subsample-pcol}). Across datasets, the rules separate
into two distinct groups: $\thang$
and PSB scale roughly linearly with $n$, while $\thar$, $\theta_{\mathrm{med}}$, and Borda hover near and oftentimes below $1$. Note that these distinct groups do not materialize as clearly in \Cref{tab:overview-minimal-body} and \Cref{fig:convergence}, where there is less voter heterogeneity. This reflects that when proportionality is more demanding, $\thang$'s dedication to proportionality becomes far more important than Borda's ability to be non-fixed / non-linear. 

\definecolor{exp8color0}{rgb}{0.1216,0.4667,0.7059}
\definecolor{exp8color1}{rgb}{1.0000,0.4980,0.0549}
\definecolor{exp8color2}{rgb}{0.1725,0.6275,0.1725}
\definecolor{exp8color3}{rgb}{0.8392,0.1529,0.1569}
\definecolor{exp8color4}{rgb}{0.5804,0.4039,0.7412}

\begin{figure*}[t]
  \centering

  \begin{tikzpicture}
    \node[font=\footnotesize, inner sep=0pt] {%
      \tikz[baseline=-0.5ex]\draw[exp8color0, line width=1pt] (0,0) -- (0.35,0);~arithmetic mean\quad
      \tikz[baseline=-0.5ex]\draw[exp8color1, line width=1pt] (0,0) -- (0.35,0);~angular mean\quad
      \tikz[baseline=-0.5ex]\draw[exp8color2, line width=1pt] (0,0) -- (0.35,0);~geometric median\quad
      \tikz[baseline=-0.5ex]\draw[exp8color3, line width=1pt] (0,0) -- (0.35,0);~Borda\quad
      \tikz[baseline=-0.5ex]\draw[exp8color4, line width=1pt] (0,0) -- (0.35,0);~PSB
    };
  \end{tikzpicture}

  \vspace{4pt}
    \input{.//Figures_final/double-right}
\caption{$\pind$ (solid) and $\pcol$ (dashed) on 2D data.}
\label{fig:convergence}
    \label{fig:exp10}
\end{figure*}

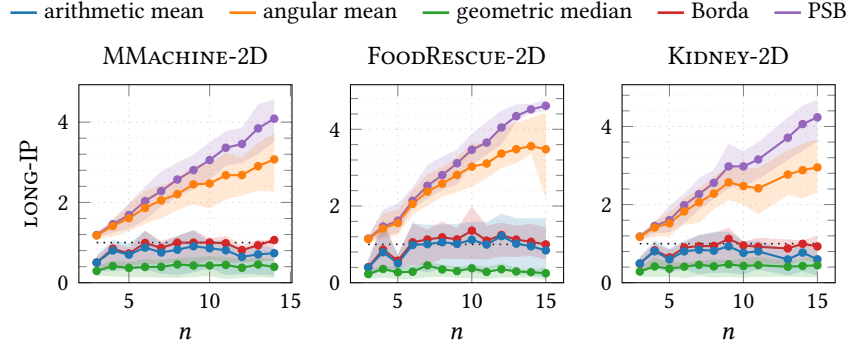
\begin{figure*}[t]
  \centering

  \begin{tikzpicture}
    \node[font=\footnotesize, inner sep=0pt] {%
      \tikz[baseline=-0.5ex]\draw[exp8color0, line width=1pt] (0,0) -- (0.35,0);~arithmetic mean\quad
      \tikz[baseline=-0.5ex]\draw[exp8color1, line width=1pt] (0,0) -- (0.35,0);~angular mean\quad
      \tikz[baseline=-0.5ex]\draw[exp8color2, line width=1pt] (0,0) -- (0.35,0);~geometric median\quad
      \tikz[baseline=-0.5ex]\draw[exp8color3, line width=1pt] (0,0) -- (0.35,0);~Borda\quad
      \tikz[baseline=-0.5ex]\draw[exp8color4, line width=1pt] (0,0) -- (0.35,0);~PSB
    };
  \end{tikzpicture}

  \vspace{4pt}

    \centering
    \input{.//Figures_final/double-left}
    \caption{$\pind$ on filtered 2D data. Shading shows inter-quartile range across $100$ subsamples.}
\label{fig:subsample}
\end{figure*}

\section{Discussion}\label{sec:discussion}
Zooming out, there are many ways to collectively decide a fixed linear ranking rule: one could pool raw elicitation data to fit a collective model, as in alignment; learn noisy $\hat{\theta}^{(i)}$ that are correct up to rankings, and then simulate a vote, as in virtual democracy; directly aggregate noisy $\hat{\theta}^{(i)}$ with robust methods; or collect enough individual data that more precise $\theta^{(i)}$ can be learned and aggregated, as we study here. Given our finding that there exists a linear ranking rule $\thang$ that achieves $\pind$ and quickly converges in $m$ to achieving $\pcol$, a natural next question is how well linear ranking rules learned through each of these less information-rich regimes can approximate these notions of proportionality. Many applications may also merit extending these ideas to nonlinear fixed rules and downstream decision tasks other than ranking.

\appendix

\newpage
\section{General-Use Appendix Results}
In this portion of the Appendix, we put examples and small technical results that we reference throughout the remainder of the Appendix.

\subsection{Tied Scores}

\begin{proposition}(Score ties occur with probability 0)\label{prop:no-ties}
For any item distribution $\cD$, finite $d, m \in \mathbb{N}$, and finite set of scoring vectors $V$ \[
\Pr_{X \sim \cD^m}\left(\exists i, j \in [m], i \neq j \text{ and } \theta \in V \colon s_{\theta}(x_i) = s_{\theta}(x_j)\right) = 0
\]

\begin{proof}
We can union bound the expression and get the following: \begin{align*}
    \Pr_{X \sim \cD^m}\left(\exists i, j \in [m], i \neq j \text{ and } \theta \in V \colon s_{\theta}(x_i) = s_{\theta}(x_j)\right) &\leq \sum_{1 \leq i < j \leq m} \sum_{\theta \in V} \Pr_{x_i,x_j \iid \cD}(s_{\theta}(x_i) = s_{\theta}(x_j))
\end{align*}
    As the total number of terms, $\binom{m}{2} \cdot |V|$, is finite by our specification of finite $m, |V|$, it suffices to show that for any $\theta \in V$, $\Pr_{x,y \iid \cD}(s_{\theta}(x) = s_{\theta}(y)) = 0$. Equivalently, $\Pr_{x_i,x_j \iid \cD}(\langle \theta, x_i \rangle = \langle \theta, x_j \rangle ) = 0$. Fix some $\theta \in V$. By conditioning on a specific realization of $x_j$, we can rewrite our expression as \[
    \Pr_{x_i, x_j \iid \cD}(\langle \theta, x_i \rangle = \langle \theta, x_j \rangle ) = \int_{\R^d} \Pr_{x_i \sim \cD}(\langle \theta, x_i\rangle = \langle \theta, c \rangle \vert x_j = c) f_{\cD}(c) dc
    \]
    Now observe that for $\langle \theta, x_i \rangle$ to be equal to $\langle \theta, c \rangle$ for some fixed $c \in \R^d$, it must be that $x_i$ is in a specific hyperplane $H_{\theta, c} = \{z \in \R^d \colon \langle \theta, z \rangle = \langle \theta, c \rangle\}$. $H_{\theta, c}$ an affine subspace of dimension $d-1$, and hence has $d$-dimensional Lebesgue measure 0: $\lambda^d(H_{\theta, c}) = 0$. The item distribution $\cD$ is absolutely continuous with respect to the Lebesgue measure, which tells us that sets with Lebesgue measure 0 have probability measure 0. Putting it all together \[
    \Pr_{x_i \sim \cD}(\langle \theta, x_i\rangle = \langle \theta, c \rangle \vert y = c) = \Pr_{x_i \sim \cD}(x_i \in H_{\theta, c}) = 0
    \]
    Plugging back into our integral, we get that \[
    \int_{\R^d} \Pr_{x_i \sim \cD}(\langle \theta, x_i\rangle = \langle \theta, c \rangle \vert x_j = c) f_{\cD}(c) dc = \int_{\R^d} 0 \cdot f_{\cD}(c) dc = 0 
    \]
    As $\Pr_{x_i, x_j \iid \cD}(s_{\theta}(x_i) = s_{\theta}(x_j)) = 0$ holds for all $\theta \in V$, we have that the expression derived from the union bound sums to 0. Hence, for a finite set of vectors, and for a finite item batch size\footnote{The proof can be extended to countably infinite item batch sizes if desired.} and dimension, ties between object scores occur with probability 0. 
\end{proof}

\begin{corollary} (Expectation conditioned on no ties)\label{cor:no-ties}
Given a profile $\prof \in P_{n,d}$ (with finite $n,d \in \mathbb{N}$), aggregation mechanism $\theta_f$, item distribution $\cD$, and finite item batch size $m \in \mathbb{N}$, we say that for an item batch $X \in \cX^m$, $\mathcal{E}_{\text{no-ties}}$ is the event where none of the scoring rules are tied on any items in $X = (x_1, \dots, x_m)$: \[
\mathcal{E}_{\text{no-ties}} = \bigcap_{\theta \in (\boldsymbol{\theta} \cup \{\theta_f\})} \bigcap_{1 \leq i < j \leq m}\{s_{\theta}(x_i) \neq s_{\theta}(x_j)\}
\]
Then, for any function $h$ that can be written as a function of these variables, the expectation over random draws of $X \sim \cD^m$ is equivalent to the expectation conditioned on $\mathcal{E}_{\text{no-ties}}$: \[
\mathbb{E}_{X \sim \cD^m} \left[h(X, \prof, \theta_f)\right] = \mathbb{E}_{X \sim \cD^m} \left[h(X, \prof, \theta_f) \vert \mathcal{E}_{\text{no-ties}} \right]
\]
\end{corollary}

\begin{proof}
    Observe that by \Cref{prop:no-ties}, we have that $\Pr(\overline{\mathcal{E}_{\text{no-ties}}}) = 0$. Then, by law of total expectation: 
    \begin{align*}
        &\mathbb{E}_{X \sim \cD^m} \left[h(X, \prof, \theta_f)\right]\\
        &=\mathbb{E}_{X \sim \cD^m} \left[h(X, \prof, \theta_f) \vert \mathcal{E}_{\text{no-ties}} \right] \Pr(\mathcal{E}_{\text{no-ties}}) + \mathbb{E}_{X \sim \cD^m} \left[h(X, \prof, \theta_f) \vert \overline{\mathcal{E}_{\text{no-ties}}} \right] \underbrace{\Pr(\overline{\mathcal{E}_{\text{no-ties}}})}_0\\
        &= \mathbb{E}_{X \sim \cD^m} \left[h(X, \prof, \theta_f) \vert \mathcal{E}_{\text{no-ties}} \right] \Pr(\mathcal{E}_{\text{no-ties}})
    \end{align*}
\end{proof}

\end{proposition}

\subsection{The Antipodal Construction}
This profile, called the Antipodal Construction, is introduced in \Cref{fig:antipodal-prof}. This profile can be realized in any number of dimensions for $d \geq 2$, where $\thetaone = (-1, \protect\underbrace{0, \dots, 0}_{d-1})$ and $\thetatwo = (1, \protect\underbrace{0, \dots, 0}_{d-1})$. By default we will assume $d=2$ in its instantiations, and we leave the $\boldsymbol{\alpha}$ up to specification. We use it in numerous proofs, so we place it here again for reference and show a simple fact about the profile. 

\begin{figure}[ht!]
\centering
\begin{tikzpicture}[scale=1]

\draw[gray] (0,0) circle (1);

\fill (0,0) circle (0.02);

\coordinate (theta1) at (180:1);
\coordinate (theta2) at (0:1);

\draw[->] (0,0) -- (theta1);
\draw[->] (0,0) -- (theta2);

\node[left] at (theta1) {$\theta^{(1)} \scriptscriptstyle(\alpha_1)$};
\node[right] at (theta2) {$\theta^{(2)} \scriptscriptstyle(\alpha_2)$};

\end{tikzpicture}
\end{figure}

\begin{proposition}\label{prop:antipodal-disagreement}
    For $\thetaone$ and $\thetatwo$ in the Antipodal Construction, and any additional scoring vector $\theta^* \in S^{d-1}$, for any $X \in \cX^m$ such that there are \textit{no} scoring vector ties, \[
    a_{KT}(\thetatwo, \theta^*, X) = \binom{m}{2} - a_{KT}(\thetaone, \theta^*, X)
    \]
\end{proposition}

\begin{proof}
    Observe that $\thetaone = -\thetatwo$, so for any $\{x, y\} \in \binom{X}{2}$, $s_{\thetaone}(x) > s_{\thetaone}(y) \iff s_{\thetatwo}(x) < s_{\thetatwo}(y)$. With no tied scores, this means that $\succ_{\thetaone}^X$ and $\succ_{\thetatwo}^X$ disagree on all $\binom{m}{2}$ pairwise comparisons. Therefore, for any pair $\{x, y\} \in \binom{X}{2}$, $\theta^*$ must agree with exactly one of $\thetaone$ or $\thetatwo$ on the ranking.
\end{proof}
\section{Extended Related Work} \label{app:related}

We discuss related work on incorporating public input into the design of decision rules. Existing work differs mainly along three dimensions. First, the elicited input may take the form of a scoring vector for each voter, or of ordinal preference information such as pairwise comparisons or rankings. Second, the aggregation procedure may output a single collective scoring vector that is applied repeatedly, or it may aggregate preferences separately for each decision instance, thereby allowing decisions that are not necessarily induced by the same rule across rounds. Third, the downstream decision task may be to rank or to select candidates.

\paragraph{Virtual democracy and participatory design.}
The most closely related literature is \emph{virtual democracy}, a
paradigm that uses models of each voter's preferences to make
collective decisions on their behalf at decision time
\cite{DBLP:conf/icml/KahngLNPP19}.
Closest to our work,
\citet{noothigattu2018voting} and \citet{feffer2023moral}
study the problem of aggregating individual scoring vectors into a
collective one.
\citet{freedman2020adapting} take an alternative approach: they learn
a collective scoring vector directly by pooling pairwise comparisons
across voters before fitting a single linear model.
\citet{DBLP:conf/icml/KahngLNPP19} and \citet{lee2019webuildai} pursue
a third paradigm in which, at every decision instance, each voter's
preferences over the current alternatives are inferred from a learned
model and aggregated via a single-winner voting rule, producing decisions
that need not be consistent across instances.

\paragraph{Aggregating scoring vectors in social choice.}
Aggregating voters' scoring vectors into a single collective vector
has been studied extensively from a different perspective in divisible
participatory budgeting, also known as portioning, where each voter
proposes a budget split over projects
\citep{DBLP:journals/ai/ElkindGLST26,DBLP:journals/jet/FreemanPPV21,DBLP:journals/corr/abs-2405-20303,DBLP:journals/ai/CaragiannisCP24,DBLP:conf/sigecom/0001GSS24,FreemanS24}.
In this literature, the arithmetic mean is the canonical proportional
aggregator, and the central concern is its lack of strategyproofness.
In contrast to our work, however, the downstream task is a one-shot
cardinal task rather than a repeated ordinal one; under the
proportionality notion we study, the arithmetic mean no longer retains
the same guarantees.

\paragraph{Proportionality in repeated decision making.}
A separate line of work studies individual and group-level proportional
fairness in repeated decision making
\citep{DBLP:conf/aaai/Lackner20,DBLP:conf/aaai/Kozachinskiy0S25,chandak2024proportional,DBLP:journals/access/BulteauHPRT21,DBLP:conf/aaai/Lackner023,DBLP:journals/corr/abs-2505-22513}.
These works typically allow the use of a different decision rule at
each decision step, and often retain memory across rounds to balance
representation over time. Moreover, the downstream task is usually to
select one or several candidates rather than to output a ranking induced
by a fixed scoring vector.

\paragraph{Proportional rank aggregation.}
Individual and group-level proportionality in rank aggregation, i.e.,
the task of aggregating a given set of rankings into a single collective
one, has recently been studied by \citet{lederer2024squared} and
\citet{DBLP:journals/corr/abs-2508-16177}. In particular,
\citet{DBLP:journals/corr/abs-2508-16177} introduces the Proportional
Sequential Borda (PSB) rule, which guarantees that any
$\alpha$-fraction of the electorate agrees with roughly an
$\alpha$-fraction of the pairwise comparisons in the collective
ranking. PSB can be applied to our setting per batch, but the
resulting decisions are neither linear nor consistent across batches.
We benchmark against PSB empirically in \Cref{sec:empirical}.

\paragraph{RLHF and alignment from a social-choice perspective.}
In reinforcement learning from human feedback, the reward model plays
the role of a typically nonlinear scoring rule, and standard practice
pools annotator data before fitting
\citep{christiano2017deep,ouyang2022training,rafailov2023direct}.
A growing line of work studies alignment as a collective
decision-making problem
\citep{DBLP:conf/nips/Ge0MPSV024,hollender2026enforcing,DBLP:conf/iclr/SiththaranjanLH24,DBLP:conf/icml/ProcacciaS025,DBLP:journals/corr/abs-2501-19266,sorensen2024position,conitzer2024social}.
Closest to us, \citet{DBLP:conf/nips/Ge0MPSV024} give an axiomatic
analysis of rules that aggregate individual rankings into a collective
weight vector under a linear functional form for the reward model;
their axioms, however, target majoritarian and efficiency properties
rather than proportionality.
\citet{shirali2025direct} argue that the arithmetic mean of per-type
reward functions is a principled aggregator when the downstream task
uses cardinal scores. Our angular mean offers a viable alternative in
the ordinal regime, and can also be applied to combine the linear heads
of multi-objective reward models that capture diverse user preferences
or objectives
\citep{DBLP:conf/acl/WangLXYDQZZ24,DBLP:conf/emnlp/00030X0024}.

\section{Appendix for \Cref{sec:model}} \label{app:model}

\begin{proposition}[Arithmetic Mean Definition Equivalence] \label{prop:arith_mean_def}
    In the event that $\|\sum_{i=1}^n \alphai \thetai\|_2 > 0$, we have that \[
    \thar = \argmin_{\theta \in S^{d-1}} \sum_{i \in [n]} \alphai \|\theta - \thetai \|_2^2 = \frac{\sum_{i \in [n]} \alphai \thetai}{\|\sum_{i \in [n]} \alphai \thetai\|_2}
    \]
\end{proposition}
\begin{proof}
Call $\argmin_{\theta \in S^{d-1}} \sum_{i \in [n]} \alphai \|\theta - \thetai \|_2^2$ the minimizing form and $\frac{\sum_{i \in [n]} \alphai \thetai}{\|\sum_{i \in [n]} \alphai \thetai\|_2}$ the closed form. Equality when the minimizing form minimizes over all $\theta \in \R^d$ and when the closed form is unnormalized is commonly known, but this variation is less common. We start by expanding the expression we're minimizing: \[
\sum_{i \in [n]} \alphai \|\theta - \thetai \|_2^2 = \sum_{i\in[n]} \alphai (\underbrace{\|\theta\|_2^2}_1 + \underbrace{\|\thetai\|_2^2}_1 - 2 \langle \theta , \thetai\rangle)  = 2\left(\sum_{i \in [n]} \alphai\right) - 2\langle \theta, \sum_{i \in [n]} \alphai \thetai \rangle 
\]
Then, as $\sum_{i \in [n]} \alphai = 1$, we have that the minimizing form simplifies to $\argmin_{\theta \in S^{d-1}} 2- 2\langle \theta, \sum_{i \in [n]} \alphai \thetai \rangle$. This is then equivalent to $\argmax_{\theta \in S^{d-1}} \langle \theta, \sum_{i \in [n]} \alphai \thetai \rangle$. Let $V = \sum_{i \in [n]} \alphai \thetai$ and $\phi$ be the angle between $\theta$ and $V$. We can rewrite \[
 \langle \theta, V \rangle = \|\theta\|_2\|V\|_2\cos(\phi)
\]
Observe that because we require $\theta \in S^{d-1}$, we have that $\|\theta\|_2 = 1$. So $\langle \theta, V \rangle = \|V\|_2 \cos(\phi)$. Under our assumption that $\|V\|_2 > 0$, this expression is uniquely maximized when $\cos(\phi) = 1$, which implies that $\phi = 0$. Hence it must be the case that $\thar$ is in the same direction as $V$. In combination with the knowledge that $\thar$ is a unit vector, this gives us everything we need to write it down in closed form: \[
\thar = \frac{V}{\|V\|_2} = \frac{\sum_{i \in [n]} \alphai \thetai}{\|\sum_{i \in [n]} \alphai \thetai\|_2}
\]
\end{proof}

\restatehere{propindrange}
\begin{proof}
The lower bound follows because for all $i \in [n]$, $i$'s individual proportionality level is always nonnegative. 

As $\pind^*(\theta_f, \cD)$ minimizes over all profiles, we can upper bound it by constructing a specific profile, $\prof$, and evaluating $\pind(\theta_f, \cD^m \prof)$. Let $\prof$ be the Antipodal Construction shown in \Cref{fig:antipodal-prof} with $\alpha^{(1)} = \alpha^{(2)} = 1/2$.  
\begin{align*}
    &\pind(\theta_f, \cD^m, \prof) \\
    &= \min \left(\mathbb{E}_{X \sim \cD^m} \left[\ip_1(\theta_f, X, \prof)\right], \mathbb{E}_{X \sim \cD^m} \left[\ip_2(\theta_f,X,\prof)\right]\right)\\
    &= \min \left(\mathbb{E}_{X \sim \cD^m} \left[\frac{a_{KT}(\thetaone, \theta_f\prof, X)}{\alpha^{(1)} \binom{m}{2}}\right], \mathbb{E}_{X \sim \cD^m} \left[\frac{a_{KT}(\thetatwo, \theta_f\prof, X)}{\alpha^{(2)} \binom{m}{2}}\right]\right)
    \intertext{Now we can make two simplifications. The first is that $\alpha^{(1)} = \alpha^{(2)}$, so we can just take the denominator out of the expectations and outside of the min.}
    &= \frac{1}{\alpha^{(1)} \binom{m}{2}} \min\left(\mathbb{E}_{X \sim \cD^m} \left[a_{KT}(\thetaone, \theta_f\prof, X)\right], \mathbb{E}_{X \sim \cD^m} \left[a_{KT}(\thetatwo, \theta_f\prof, X)\right]\right)
    \intertext{The other simplification is that via $\Cref{cor:no-ties}$ we can assume no scoring vector ties between $\theta_f, \thetaone, \thetatwo$ on $X$ in both of our expectations, and by \Cref{prop:antipodal-disagreement}, this means that $a_{KT}(\thetatwo, \theta_f\prof, X) = \binom{m}{2} - a_{KT}(\thetaone, \theta_f\prof, X)$. Then we can rewrite the expression as:}
    &= \frac{1}{\alpha^{(1)} \binom{m}{2}} \min\left(\mathbb{E}_{X \sim \cD^m} \left[a_{KT}(\thetaone, \theta_f\prof, X)\right], \mathbb{E}_{X \sim \cD^m} \left[\binom{m}{2} - a_{KT}(\thetaone, \theta_f\prof, X)\right]\right)\\
    &= \frac{1}{\alpha^{(1)} \binom{m}{2}} \min\left(\mathbb{E}_{X \sim \cD^m} \left[a_{KT}(\thetaone, \theta_f\prof, X)\right], \binom{m}{2} -\mathbb{E}_{X \sim \cD^m} \left[a_{KT}(\thetaone, \theta_f\prof, X)\right]\right)
\end{align*}
Observe that this minimum is now just a function of $\mathbb{E}_{X \sim \cD^m} \left[a_{KT}(\thetaone, \theta_f\prof, X)\right]$, and is maximized when the two expressions are equal. $\mathbb{E}_{X \sim \cD^m} \left[a_{KT}(\thetaone, \theta_f\prof, X)\right] = \binom{m}{2} - \mathbb{E}_{X \sim \cD^m} \left[a_{KT}(\thetaone, \theta_f\prof, X)\right] \implies \mathbb{E}_{X \sim \cD^m} \left[a_{KT}(\thetaone, \theta_f\prof, X)\right] = \frac{1}{2} \cdot \binom{m}{2}$. Plugging this all in along with our value of $\alpha^{(1)}$ gives that \[
\pind(\theta_f, \cD^m, \prof) \leq \frac{1}{\alpha^{(1)}\binom{m}{2}} \cdot \frac{\binom{m}{2}}{2} = 1
\]
And hence we can conclude that \[
\pind^*(\theta_f, \cD^m) \leq 1
\]

\end{proof}

\restatehere{arithmeanpind}

\begin{proof}
    As with the proof of $\Cref{prop:propind_range}$, we can upper bound $\pind^*(\theta_f, \cD)$ by evaluating $\pind(\theta_f, \cD^m, \prof)$ on a given profile $\prof$. Once again, let $\prof$ be the Antipodal Construction in \Cref{fig:antipodal-prof}, this time with $\alpha^{(1)} = 0.3$ and $\alpha^{(2)} = 0.7$. Then we have that: \[
    \thar\prof = \frac{\alpha^{(1)} \thetaone + \alpha^{(2)} \thetatwo}{\|\alpha^{(1)} \thetaone + \alpha^{(2)} \thetatwo\|} = \frac{(-0.3, 0) + (0.7, 0)}{\|(-0.3, 0) + (0.7, 0)\|} = (1,0) = \thetatwo
    \]
    Hence, once we normalize, we have that $\thar\prof = \thetatwo$. As discussed in \Cref{prop:propind_range}, we have that for $X \sim \cD^m$, with no tied scores, $\succ_{\thetaone}^X$ and $\succ_{\thetatwo}^X$ will disagree on all pairwise comparisons. As $\thar\prof = \thetatwo$, we know that $\succ_{\thetaone}^X$ and $\succ_{\thar\prof}^X$ will disagree on all pairwise comparisons. Hence, by \Cref{cor:no-ties} we can assume no ties, and we have that \[
    \mathbb{E}_{X \sim \cD^m}[a_{KT}(\thetaone, \thar\prof, X)] = 0
    \]
    Therefore, \begin{align*}
        \pind^*(\thar, \cD^m) &\leq \pind(\thar, \cD^m, \prof)\\
        &= \min_{i \in \{1,2\}} \mathbb{E}_{X \sim \cD^m}[a_{KT}(\thetai, \thar\prof, X)]\\
        &= 0
    \end{align*}
    
\end{proof}

\section{Appendix for \Cref{sec:ang_mean_prop}} \label{app:proportionality}

\subsection{Proof of \Cref{thm:thang_prop}}

\thangprop*

\begin{proof}
    We can apply the definition of $\pind^*$, as well as \Cref{lem:ktangle} and \Cref{thm:angular_bound} to obtain:
    \begin{align*}
    \pind^*(\thang, \D_\circ^m) &= \inf_{\prof \in P} \pind(\thang, \D^m, \prof)\\
                                &= \inf_{\prof \in P} \min_{i \in [n]} \E_{X \sim \D^m} [\agi(\thang, X, \prof)] \\
                                &= \inf_{\prof \in P} \min_{i \in [n]} \E_{X \sim \D^m} \left[\frac{a_{KT}(\thetai, \thang\prof, X)}{\alphai \binom{m}{2}}\right] \\
                                &= \inf_{\prof \in P} \min_{i \in [n]} \E_{X \sim \D^m} \left[\frac{(\pi - \dang(\thetai, \thang)) \binom{m}{2}}{\pi \alphai  \binom{m}{2} }\right] \\
                                &\geq \inf_{\prof \in P} \min_{i \in [n]} \E_{X \sim \D^m} \left[\frac{\pi - (1 - \alphai)\pi \binom{m}{2}}{\pi \alphai  \binom{m}{2} }\right] \\
                                &= \inf_{\prof \in P} \min_{i \in [n]} \E_{X \sim \D^m} \left[\frac{\pi \alphai \binom{m}{2}}{\pi \alphai  \binom{m}{2} }\right] = 1.
    \end{align*}
    Due to \Cref{prop:propind_range}, we also get that $\pind^*(\thang, \D_\circ^m) \leq 1$, so $\pind^*(\thang, \D_\circ^m) = 1$.
\end{proof}

\subsection{Proof of \Cref{lem:ktangle}.}

\restatehere{agreementangle}

\begin{proof}

Recall the relation between angles and expected agreement from \Cref{prop:angle_eprop}.
\[
\E_{x \sim \D, y \sim \D}\big[\mathbf{1}(x \succ_\theta y \land x \succ_\psi y)\big] = \E_{X \sim \D_{\circ}^2}[\aKT(\theta, \psi, X)] = \frac{\pi - \dang(\theta, \psi)}{\pi}.
\]
For $X = (x_1, \dots, x_m)$ we then apply linearity of expectation to $\E_{X \sim \D^m}[\aKT(\theta, \psi, X)]$ and obtain
\begin{align*}
    \E_{X \sim \D^m}\big[\aKT(\theta, \psi, X)\big] & = \E_{X \sim \D_{\circ}^m}\left[\Bigl|\{x_i,x_j\} \in \tbinom{X}{2}: (x_i \succ_\theta x_j \land x_i \succ_\psi x_j) \lor  (x_j \succ_\theta x_i \land x_j \succ_\psi x_j)\Bigr|\right] \\
    & = \E_{X \sim \D_{\circ}^m}\left[\sum_{\{x_i, x_j\} \in \binom{X}{2}} \mathbf{1}\big((x_i \succ_\theta x_j \land x_i \succ_\psi x_j) \lor  (x_j \succ_\theta x_i \land x_j \succ_\psi x_j)\big)\right] \\
    & = \sum_{\{i,j\} \in \binom{[m]}{2}} \E_{X \sim \D_{\circ}^m}\left[\mathbf{1}\big((x_i \succ_\theta x_j \land x_i \succ_\psi x_j) \lor  (x_j \succ_\theta x_i \land x_j \succ_\psi x_j)\big)\right] \\
    & = \sum_{\{i,j\} \in \binom{[m]}{2}} \E_{x_i \sim \D_{\circ}, x_j \sim \D_{\circ}}\left[\mathbf{1}\big((x_i \succ_\theta x_j \land x_i \succ_\psi x_j) \lor  (x_j \succ_\theta x_i \land x_j \succ_\psi x_j)\big)\right] \\
    & = \sum_{\{i,j\} \in \binom{[m]}{2}} \frac{\pi - \dang(\theta, \psi)}{\pi} \\
    &= \frac{\pi - \dang(\theta, \psi)}{\pi} \binom{m}{2}.
\end{align*}
\end{proof}

\subsection{Proof of \Cref{thm:angular_bound}}

In this section we provide the proof details for \Cref{thm:angular_bound}.

\restatehere{angbound}

The primary idea behind the proof is to show our main auxiliary Lemma,  \Cref{lem:circle_angular_bound}. Essentially, it states that the distance bound from \Cref{thm:angular_bound} remains valid for $\thetai$, even if we may only select a minimizer $\theta^*$ of $F(\theta) = \sum_{i = 1}^n \alpha^{(i)} \dang(\theta, \thetai)^2$ from some great circle $G_i$ that contains $\thetai$.

\begin{restatable}{lemma}{circleangbound}
    \label{lem:circle_angular_bound}
    For a given profile $(\btheta, \balpha)$, let $F$ be the objective function of the angular mean,
    \[
        F(\theta) = \sum_{i \in [n]} \alpha^{(i)} \dang(\theta, \thetai)^2.
    \]
    Let $G_i \subseteq S^{d}$ be any great circle containing $\thetai$, and write $F_{|G_i} : G_i \to \mathbb{R}$ for the restriction of $F$ to $G_i$.

    Every minimizer $\theta^* \in G_i$ of $F_{|G_i}$ fulfills $\dang(\thetai, \theta^*) \leq (1 - \alpha^{(i)})\pi$.
\end{restatable}

Let us first formally establish why \Cref{thm:angular_bound} follows from \Cref{lem:circle_angular_bound}.

\begin{proof}[Proof of \Cref{thm:angular_bound}]
Let $\prof$ be a profile in $S^d$.  Consider any voter $i \in [n]$ and let $G_i$ be a great circle passing through both $\thetai$ and $\thang$. Because $\thang$ is a global minimizer of $F$, it also minimizes $F$ on $G_i$. Thus, by \Cref{lem:circle_angular_bound},
\[
\dang(\thang, \thetai) \leq (1 - \alpha^{(i)})\pi.
\]
\end{proof}

The full proof of \Cref{lem:circle_angular_bound} is more technical. The high-level idea is a proof by contradiction. We will show that, for a global minimizer $\theta^*$ of $F_{|G_i}$ and a vector $\thetai \in G_i$, that if $\dang(\theta^*, \thetai) > (1-\alpha^{(i)}) \pi$ then we can construct another vector $\theta'$ with a lower objective value with $\dang(\theta', \thetai) < \dang(\theta^*, \thetai)$. In the proof we will leverage a distance bound inspired by the Euclidean law of cosines and Topongov's Inequality (\Cref{lem:fake_toponogov}) and a first order local optimality condition on $\theta^*$ 
Employing these two tools allows us to greatly simplify the the expression $F(\theta') - F(\theta)$, which then makes it easy to verify that $\theta'$ does indeed achieve an improved objective value.

To use the first order optimality condition on the derivative of $F_{|G_i}$ at $\theta^*$, we must first establish that $F$ is differentiable at $\theta^*$ in the first place. Because each $d_i(\theta) = \dang(\theta, \thetai)^2$ is differentiable everywhere other than $-\thetai$, we just need to show that $\theta^*$ cannot be antipodal to any of the $\thetai$.

\begin{lemma}\label{lem:circle_antipodal}
    Let $G_i$ and $F_{|G_i}$ be as in \Cref{lem:circle_angular_bound}. No minimizer $\theta^* \in G_i$ of $F_{|G_i}$ is antipodal to any of the $\thetai$.
\end{lemma}

\begin{proof}
    First observe that the map $d_i(\theta) = \dang(\theta, \thetai)$ is differentiable everywhere other than $-\thetai$, where only the directional derivatives exist. So suppose that $\theta^*$ is antipodal to some of the $\thetai$. Let $I = \{i \in [n]: \thetai = - \theta^*\}$ be the set of indices of these antipodal $\thetai$.

    Let $x : \mathbb{R}/2\pi\mathbb{Z} \to G_i$ be a unit-speed parametrization of $G_i$ such that $x(0) = \theta^*$. We will decompose $F$ into a singular and non-singular part as  $F(\theta) = F_{\text{sin}}(\theta) + F_{\text{diff}}(\theta)$, with
    \[
    \Fs = \sum_{i \in I} \alpha^{(i)} \dang(\theta, \thetai) \quad \text{ and } \quad  \Fd = \sum_{i \notin I} \alpha^{(i)} \dang(\theta, \thetai)
    \]
    
    The right-sided derivative of $F(x(\vphi))$ at $0$ is then given by
    \[\left.\frac{d}{d \vphi}F(x(\vphi))\right|_{\vphi = 0^+} = \left(\left.\frac{d}{d \vphi}\Fs(x(\vphi))\right|_{\vphi = 0^+} \right) + \Fd(x(\vphi))' = \Fs^+ + \Fd(x(\vphi))',\]
    while the left-sided derivative is given by
    \[\left.\frac{d}{d \vphi}F(x(\vphi))\right|_{\vphi = 0^-} = \left(\left.\frac{d}{d \vphi}\Fs(x(\vphi))\right|_{\vphi = 0^-} \right) + \Fd(x(\vphi))' = \Fs^- + \Fd(x(\vphi))'.\]
            Now, because $\thetai$ is antipodal to $\theta^* = x(0)$ for all $i \in I$, we get that $\Fs(\vphi)$ = $\Fs(-\vphi)$ and hence the left and right-sided derivative of $\Fs(x(\vphi))$ at $0$ are identical up to their sign, i.e. $\Fs^+ = - \Fs^-$. In particular $\Fs^+ < 0$ and $\Fs^- > 0$ , as  $\Fs$ only decreases when moving away from $x(0)$. $x(0)$ can only be a minimum when the right-side derivative non-negative and the left-sided derivative is non-positive, so
            \begin{align*} 
                0 \le F_{\text{sin}}^+ + F_{\text{diff}}(x(0))' &\implies -F_{\text{diff}}(x(0))' \le F_{\text{sin}}^+ \\
                0 \ge F_{\text{sin}}^- + F_{\text{diff}}(x(0))' &\implies F_{\text{diff}}(x(0))' \le -F_{\text{sin}}^-
            \end{align*}
            Combining these two inequalities then yields
            $$0 < F_{\text{sin}}^- \le F_{\text{diff}}(x(0))' \le -F_{\text{sin}}^- < 0,$$
            a contradiction. Therefore $\theta^*$ cannot be antipodal to $\theta^{(i)}$.

\end{proof}

Next, we will establish how we upper-bound spherical distances in the proof. Let us give a quick introduction to the aforementioned Toponogov's Inequality. Intuitively it states that spherical triangles are "tighter" than flat euclidean triangles, because the edges can curve towards each other.
More formally, we compare a spherical triangle and a euclidean triangle that are both constructed using the same two side lengths and one angle. Toponogov's Inequality then states that the last side of the spherical triangle is shorter than that of the euclidean triangle.

\begin{restatable}{theorem}{toponogov}[\cite{toponogov1964riemannian, meyer1989toponogov}]\label{thm:toponogov}
    Let $M$ be an $m$-dimensional Riemannian manifold with sectional curvature $K$ satisfying $K \geq \delta$. Let $pqr$ be a geodesic triangle, i.e. a triangle whose sides are geodesics, in $M$, such that the geodesic $pq$ is minimal and if $\delta > 0$, the length of the side $pr$ is less than $\pi / \sqrt{\delta}$. Let $p'q'r'$ be a geodesic  triangle in the model space $M_\delta$, i.e. the simply connected space of constant curvature $\delta$, such that the lengths of sides $p'q'$ and $p'r'$ are equal to that of $pq$ and $pr$ respectively and the angle at $p'$ is equal to that at $p$. Then
\[
d(q, r) \leq d(q', r').
\]
\end{restatable}

We can combine this theorem with the euclidean law of cosines to bound the side lengths of spherical triangles.

\begin{restatable}{corollary}{cosbound}\label{cosbound}
    Consider three points $p, q, r \in S^n$ forming a geodesic spherical triangle. Then
    \[ \dang(q, r)^2 \leq \dang(p, q)^2 + \dang(p, r)^2 - 2 \dang(p, q) \dang(p, r) \cdot \cos(\gamma), \]
    where $\gamma$ is the angle between $\overline{pq}$, and $\overline{pr}$ at $p$.
\end{restatable}

\begin{proof}
    We use the law of cosines in conjunction with Toponogov's theorem. First, note that the sectional curvature $K$ of the unit sphere is constantly equal to $1$.
    Applying Toponogov's theorem with lower curvature bound $\delta = 1$, we obtain $d(q, r) \leq d(q', r')$, where $p'q'r'$ is the comparison triangle in the model space $M_1 = S^n$ with the same two side lengths and included angle $\gamma$ at $p'$. Applying the Euclidean law of cosines to this Euclidean triangle then yields
    \[
    d(q, r)^2 \leq d(q', r')^2 = d(p, q)^2 + d(p, r)^2 - 2d(p, q)d(p, r)\cos(\gamma).
    \]
\end{proof}

The specific bound we use in our main proof is related, but slightly stronger. The distance bound should also apply in case we do not consider strictly geodesic triangles. The precise formulation is as follows.

\begin{restatable}{lemma}{faketoponogov}\label{lem:fake_toponogov}
    Let $\theta^*, \theta^{(j)} \in S^d$ with spherical distance $l_j = \dang(\theta^*, \theta^{(j)})$. Let $x: \mathbb{R}/2\pi\mathbb{Z} \to S^d$ be a unit-speed geodesic with $x(0) = \theta^*$, and let $\gamma_j$ be the angle between the geodesic $\overline{\theta^* \theta^{(j)}}$ and the tangent $x'(0)$ in the positive direction of $x$.
    
    For any $\vphi \in [-2\pi, 0]$, $x_j(\vphi)^2 = \dang(x(\vphi), \theta^{(j)})^2$ satisfies \[ x_j(\vphi)^2 \leq l_j^2 + \vphi^2 - 2 l_j \vphi \cos(\gamma_j). \]
\end{restatable}

\begin{proof}
    Since $\vphi \leq 0$, $\vphi$ can be seen to travel backward when compared to the forward facing tangent $x'(0)$. 
    We consider two cases for $\vphi$:
    \paragraph{Case 1: $\vphi \in(-\pi, 0]$.}
    In this case the geodesic $\overline{\theta^*x(\vphi)}$ runs directly opposite to $x'(0)$, so the angle between the geodesics $\overline{x(\vphi)\theta^*}$ and $\overline{\theta^*\theta^{(j)}}$ is equal to $\pi - \gamma_j$. The lengths of these geodesics are $\dang(x(\vphi), \theta^*) = |\vphi| = -\vphi$ and $\dang(\theta^*, \theta^{(j)}) = l_j$, so by applying \Cref{cosbound} to these two geodesics with the angle $\pi - \gamma_j$ we obtain

    \begin{align*}
        x_j(\vphi^2) &\leq l_j^2 + (-\vphi)^2 - 2 l_j (-\vphi) \cos(\pi - \gamma_j) \\
                     &= l_j^2 + \vphi^2 - 2 l_j (-\vphi) (-\cos(\gamma_j)) \\
                     &= l_j^2 + \vphi^2 - 2 l_j \vphi \cos(\gamma_j),
    \end{align*}
    which is exactly the desired bound for $\vphi \in [-\pi, 0)$.

    \paragraph{Case 2: $\vphi \in (-\pi, -2\pi]$.} Here, we employ \Cref{cosbound} for a different triangle. We replace the vertex $\theta^*$  by $-\theta^*$ when compared to the spherical triangle in case 1. So we will consider the two geodesics $\overline{-\theta^*x(\vphi)}$ and $\overline{-\theta^* \theta^{(j)}}$. Note that the angle between them is also $\gamma_j$, because it sits at the unique other intersection between the great circle parametrized by $x$ and the geodesic $\overline{\theta^*, \theta^{(j)}}$. 
    The lengths are then $\dang(-\theta^*, x(\vphi)) = -\vphi - \pi$ and $\dang(-\theta^*, \theta^{(j)}) = \pi - l_j$. Applying \Cref{cosbound}

    \begin{align*}
        x_j(\vphi^2) &\leq (\pi - l_j)^2 + (-\vphi - \pi)^2 - 2 (\pi - l_j) (- \vphi - \pi) \cos(\gamma_j) \\
                     & \leq (\pi - l_j)^2 + (-\vphi - \pi)^2 + 2 (\pi - l_j) ( \vphi + \pi) \cos(\gamma_j) \\
                     &= (\underbrace{\pi^2 - 2\pi l_j}_{(B)} + \underbrace{l_j^2}_{(A)}) + (\underbrace{\vphi^2}_{(A)} + \underbrace{2 \pi \vphi + \pi^2}_{(B)}) + 2 (\underbrace{\pi^2 - l_j \pi + \vphi \pi}_{(C)}  \underbrace{- l_j \vphi}_{(A)}) \cos(\gamma_j) \\
                     &= \underbrace{l_j^2 + \vphi^2 - 2 l_j \vphi \cos(\gamma_j)}_{(A)} + \underbrace{2\pi^2 - 2\pi l_j + 2\pi \vphi}_{(B)} + 2(\underbrace{pi^2 - l_j \pi + \vphi \pi)}_{(C)} \cos(\gamma_j) \\
                     &= (l_j^2 + \vphi^2 - 2 l_j \vphi \cos(\gamma_j)) + 2\pi (\pi - l_j + \vphi + \cos(\gamma_j)(\pi -l_j + \vphi)) \\
    \end{align*}
    The expression in the first parentheses is what we want to obtain as an upper bound, so the statement is proven, if we can show the remainder is at most $0$. We can rewrite this remainder via
    \[ 2\pi (\pi - l_j + \vphi + \cos(\gamma_j)(\pi -l_j + \vphi)) =  2\pi (\pi - l_j + \vphi)(\cos(\gamma_j) + 1).\]
        Note that $\cos(\gamma_j) + 1 \geq 0$, so we only need to certify that $\pi - l_j + \vphi \leq 0$. This is true, since $\vphi < -\pi$ and $l_j \geq 0$ by definition.
\end{proof}

Having collected all the ingredients we can now get into the full proof of \Cref{lem:circle_angular_bound}.

\begin{proof}[Proof of \Cref{lem:circle_angular_bound}]
    Let $\theta^*$ be a global minimizer of $F_{|G_i}$.
    Because $\thetai \in G_i$, we can define a unit-speed parametrization $x : \mathbb{R}/2\pi\mathbb{Z} \to G_i$ of $G_i$ such that $x(0) = \theta^*$ and there is a unique $l_i \in [0,\pi]$ such that $x(l_i) = \thetai$. Since $G_i$ is a geodesic, for $\vphi$ in a neighborhood of $0$ we have $\dang(\theta^*, x(\vphi)) = |\vphi|$, and in particular $\dang(\theta^*, \thetai) = l_i$.

To reason about the minima of $F(x(\vphi))$, we first want to analyze the derivative of $F(x(\vphi))$ at $0$. By \Cref{lem:circle_antipodal}, there is no point antipodal to $\theta^*$, hence the derivative of $F(x(\vphi))$ at $0$ is well defined.
    If we write $l_j \coloneq \dang(\theta^*, \theta^{(j)})$ and define $x_j(\vphi) = \dang(x(\vphi), \theta^{(j)})$, then the spherical law of cosines states that
    \[ \cos(x_j(\vphi)) = \cos(l_j)\cos(\vphi) + \sin(l_j)\sin(\vphi)\cos(\gamma_j), \]
    where $\gamma_j$ is the angle between the two geodesics $\overline{\theta^* \theta^{(j)}}$ and $\overline{\theta^* x(\vphi)}$. Taking the derivative on both sides gives
    \begin{align*}
        -\sin(x_j(\vphi)) \cdot x_j(\vphi)' & = - \cos(l_j) \sin(\vphi) + \sin(l_j)\cos(\vphi)\cos(\gamma_j) \\
        \implies x_j(\vphi)' & = \frac{\cos(l_j) \sin(\vphi) - \sin(l_j) \cos(\vphi) \cos(\gamma_j)}{\sin(x_j(\vphi))}.
    \end{align*}
    Since $l_j = x_j(0)$ by definition, plugging in $\vphi = 0$ yields $x_j'(0) = -\cos(\gamma_j)$. Furthermore, because $\theta^* = x(0)$ is a global optimum of $F$ along $G_i$, first-order optimality implies
    \[
    0 = \frac{1}{2}F'(x(0)) = \sum_{j = 1}^n \alpha^{(j)} x_j(0)x_j'(0) = -\sum_{j = 1}^n \alpha^{(j)} l_j \cos(\gamma_j),
    \]
    and hence
    \[
    0 = \sum_{j = 1}^n \alpha^{(j)} l_j \cos(\gamma_j) = \sum_{j \neq i} \alpha^{(j)} l_j \cos(\gamma_j) + \alpha^{(i)} l_i,
    \]
    where the last equality follows from $\gamma_i = 0$ so $\cos(\gamma_i) = 1$.

    In addition, if we apply \Cref{lem:fake_toponogov} to $\theta^*$, $x(\vphi)$, and $\theta^{(j)}$  then we obtain
    \[x_j(\vphi)^2 \leq l_j^2 + \vphi^2 - 2 l_j \vphi \cos(\gamma_j).\]
    For $x_i(\vphi)$ specifically, we know that both $\thetai$ and $x(\vphi)$ are on the circle $G_i$, thus the distance is given by $x_i(\vphi) = \dang(\thetai, x(\vphi)) = \min_{k \in \mathbb{Z}}|(l_i - \vphi - 2k\pi)|$. Thus the squared distance $x_i(\vphi)^2$ can be bounded by
    \[ x_i(\vphi)^2 = \min_{k \in \mathbb{Z}}(l_i - \vphi - 2k\pi)^2 \leq (l_i - \vphi - 2\pi)^2.\]

    With all these bounds established, we will now derive a contradiction. Comparing the objective value of $F(x(0))$ to that of another proposed solution $F(x(\vphi))$ results in
    \begin{align*}
         & F(x(\vphi)) - F(x(0)) \\
         & \quad = \sum_{j \neq i} \alpha^{(j)} x_j(\vphi)^2 + \alpha^{(i)} x_i(\vphi)^2 - \sum_{j = 1}^n \alpha^{(j)} l_j^2 \\
                             & \quad \leq \sum_{j \neq i} \alpha^{(j)} \left(l_j^2 + \vphi^2 -2 l_j \vphi \cos(\gamma_j) \right) + \alpha^{(i)} (l_i - \vphi - 2\pi)^2 - \sum_{j = 1}^n \alpha^{(j)} l_j^2 \\
                             & \quad = \sum_{j \neq i} \alpha^{(j)} \Bigl(\underbrace{\vps l_j^2}_{(A)} + \underbrace{\vps \vphi^2}_{(B)} -\underbrace{\vps 2 l_j \vphi \cos(\gamma_j)}_{(C)} \Bigr) + \alpha^{(i)} \Bigl(\underbrace{\vps l_i^2}_{(A)} + \underbrace{\vps \vphi^2}_{(B)} + \underbrace{\vps 4\pi^2}_{(D)} - \underbrace{\vps 2l_i \vphi}_{(C)} - \underbrace{\vps 4 \pi l_i + 4\pi \vphi}_{(D)}\Bigr) - \underbrace{\vps \sum_{j = 1}^n \alpha^{(j)} l_j^2}_{(A)} \\
                             & \quad = \underbrace{\vps \sum_{j = 1}^n \alpha^{(j)} \left(l_j^2 - l_j^2\right)}_{(A)} + \underbrace{\vps \vphi^2 \sum_{j = 1}^n \alpha^{(j)}}_{(B)} - 2 \vphi \underbrace{\vps \left(\sum_{j \neq i}^n \alpha^{(j)} l_j  \cos(\gamma_j) + \alpha^{(i)} l_i \right)}_{(C)} + \underbrace{\vps \alpha^{(i)} \left(4\pi^2 + 4\pi \vphi - 4\pi l_i\right)}_{(D)} \\
    \end{align*}
     Using $l_j^2 - l_j^2 = 0$, $\sum_{j = 1}^n \alpha^{(j)} = 1$, and the local optimality condition $\sum_{j \neq i} \alpha^{(j)} l_j \cos(\gamma_j) + \alpha^{(i)} l_i = 0$, this reduces to
     $F(x(\vphi)) - F(x(0)) = \vphi^2 + \alpha^{(i)} 4 \pi (\pi + \vphi - l_i).$
     
    By using the assumption that $l_i > (1 - \alpha^{(i)}) \pi$ and setting $\vphi = - 2\pi \alpha^{(i)}$, we then get
    \begin{align*}
        F(x(\vphi)) - F(x(0)) &= \vphi^2 + \alpha^{(i)} 4 \pi (\pi + \vphi - l_i) \\
                              &< 4 \pi^2 (\alpha^{(i)})^2 + \alpha^{(i)} 4 \pi (\pi - 2\pi \alpha^{(i)} - (1 - \alpha^{(i)}) \pi) \\            
                              &= 4 \pi^2 (\alpha^{(i)})^2 + 4\pi \alpha^{(i)}(- \alpha^{(i)} \pi) \\
                              & = 4\pi^2 (\alpha^{(i)})^2 - 4\pi^2 (\alpha^{(i)})^2 = 0.
    \end{align*}
    But this means that $F(x(\vphi)) < F(x(0)) = F(\theta^*)$, thus contradicting our initial assumption that $\theta^*$ minimizes $F$ on $G_i$.
    \end{proof}

\subsection*{Approximation}

\subsection{Squared Kemeny Objective}\label{app:squared_kem_approx}

\citet{chandak2024proportional} present the Squared Kemeny rule as a means to obtaining a proportional ranking. Their setting is different --- they receive voter rankings on every batch, and directly aggregate the rankings themselves. That is, given a profile $\prof$, and batch $X \in \cX^m$, their rule would output $\succ_{SqK}$ as follows\[
\succ_{\mathrm{sqK}}^X = \argmin_{\succ^X} \sum_{i=1}^n \alphai \cdot \text{KT-dist}(\succ^X, \succ_{\thetai}^X)^2
\]
where KT-dist$(\succ^X, \succ_{\thetai}^X) = \sum_{1 \leq j < k \leq m} \mathbb{I}((x_j \succ x_k \land x_k \succ_{\thetai}^X x_j) \lor (x_k \succ x_j \land x_j \succ_{\thetai}^X x_k))$ is the Kendall-Tau distance between $\succ^X$ and $\succ_{\thetai}^X$.

As we are looking to optimize for a similar notion of proportionality, what if we select our linear scoring rule based off of minimizing the expected Squared Kemeny objective?

If we optimize for the Squared Kemeny objective on a known item distribution $\cD$, we have 
\[
\theta_{\mathrm{sqK}}\prof \coloneqq \argmin_{\theta \in S^{d-1}} \mathbb{E}_{X \sim \cD^m}\left[\sum_{i =1}^n \alphai\left(\binom{m}{2}-a_{KT}(\thetai, \theta, X)\right)^2\right]
\]
Interestingly, we find that on spherically symmetric item distributions, $\theta_{\mathrm{sqK}}$ approaches $\thang$ for sufficiently large batch sizes. 

\begin{theorem}
    For spherically symmetric $\cD_{\circ}^m$, we have that \[
    \theta_{\mathrm{sqK}}\prof = \argmin_{\theta \in S^{d-1}} \left(\Theta(m^2) \sum_{i \in [n]} \dang(\thetai, \theta) + \Theta(m^3) c(\theta) + \Theta(m^4) \sum_{i \in [n]} \dang(\thetai, \theta)^2\right)
    \]
    where $c(\theta)$ is a value that does not depend on $m$ and is upper bounded by $\sum_{i \in [n]} \dang(\thetai, \theta)$. Importantly this means that as $m$ increases, $\theta_{\mathrm{sqK}}\prof \rightarrow \thang$.
\end{theorem}

\begin{proof}

    We simply rewrite the objective that $\theta_{\mathrm{sqK}}$ minimizes. All expectations are with respect to $X \sim \cD_{\circ}^m$. First for any $i \in [n]$, and $j, k\in [m]$ such that $j < k$, let $D_{jk}$ be the event that $\thetai$ and $\theta$ disagree on the ranking of $x_j$ and $x_k$, i.e. $D_{jk} = \mathbb{I}\{(x_j \succ_{\thetai}^X x_k \land x_k \succ_{\theta}^X x_j) \lor (x_k \succ_{\thetai}^X x_j \land x_j \succ_{\theta}^X x_x)\}$. Then observe that for any $X \in \cX^m$, \[
    \binom{m}{2} - a_{KT}(\thetai, \theta, X) = \sum_{1 \leq j < k \leq m} D_{jk}
    \]
    Substituting this into our objective, we have that $\mathbb{E}\left[\sum_{i =1}^n \alphai\left(\binom{m}{2}-a_{KT}(\thetai, \theta, X)\right)^2\right]$:\begin{align*}
        &= \mathbb{E}\left[\sum_{i =1}^n \alphai\left(\sum_{1 \leq j < k \leq m} D_{jk}\right)^2\right] =  \sum_{i=1}^n \alphai\mathbb{E}\left[\alphai\left(\sum_{1 \leq j < k \leq m} D_{jk}\right)^2\right]\\
        \intertext{When we expand this squared term, there are 3 different types of terms: those in which there is an indicator variable squared, those in which the two indicator variables being multiplied together share exactly one index, and those in which the indicator variables refer to disjoint sets of items. Informally we can write this in the following way:}
        &= \sum_{i=1}^n \alphai\mathbb{E}\left[\sum_{j,k} D_{jk}^2 + \sum_{j,k, \ell} D_{jk}D_{k\ell} + \sum_{j,k,\ell, p} D_{jk}D_{\ell p}\right]\\
        \intertext{How many of each term is there? There are $\binom{m}{2}$ of the first type, and then $\binom{m}{3} \cdot 3 \cdot 2$ of the second type (picking three distinct indices, picking which is repeated, and picking which set of two corresponds to the first indicator), and $\binom{m}{4} \cdot \binom{4}{2}$ of the last type (picking four distinct indices, and picking which correspond to the first indicator). As the $x_i$ are drawn iid, the indicators do not depend on the specific choice of indices (e.g. $D_{jk}$ is the same for any choice of $j \neq k$). So by linearity of expectation, for some choice of distinct $j, k, \ell, p$:}
        &= \sum_{i=1}^n \alphai \left(\binom{m}{2}\mathbb{E}[D_{jk}^2] + 6 \binom{m}{3}\mathbb{E}[D_{jk}D_{k\ell}] + 6 \binom{m}{4}\mathbb{E}[D_{jk}D_{\ell p}]\right)\\
        \intertext{Observe that as $D_{jk}$ is an indicator variable, $\mathbb{E}[D_{jk}^2] = \mathbb{E}[D_{jk}]$. This is just $1-\{$ the probability that $\thetai$ and $\theta$ agree on a pairwise ranking$\}$, so  $1 - \frac{\pi - \dang(\thetai, \theta)}{\pi}$ by \Cref{prop:angle_eprop}. Hence $\mathbb{E}[D_{jk}] = \frac{\dang(\thetai, \theta)}{\pi}$. Additionally, $D_{jk}$ is independent of $A_{\ell p}$ when $\{j,k\}$ and $\{\ell, p\}$ are disjoint, so $\mathbb{E}[A_{jk}A_{\ell p}] = \mathbb{E}[D_{jk}]\mathbb{E}[D_{\ell p}] = \frac{\dang(\thetai, \theta)^2}{\pi^2}$. So we have that}
        &= \left(\sum_{i=1}^n \alphai \binom{m}{2}\frac{\dang(\thetai, \theta)}{\pi}\right) + \left(\sum_{i=1}^n \alphai 6 \binom{m}{3}\mathbb{E}[D_{jk}D_{k\ell}]\right) + \left(\sum_{i=1}^n \alphai 6 \binom{m}{4}\frac{\dang(\thetai, \theta)^2}{\pi^2}\right)
        \intertext{Let $c(\theta) = \sum_{i=1}^n \alphai \mathbb{E}[D_{jk}D_{k\ell}]$ and substitute in to get the final expression:}
        &= \binom{m}{2}\frac{1}{\pi}\left(\sum_{i=1}^n \alphai \dang(\thetai, \theta)\right) + 6 \binom{m}{3}c(\theta) + \frac{6}{\pi^2} \binom{m}{4}\left(\sum_{i=1}^n \alphai \dang(\thetai, \theta)^2\right)
    \end{align*}
    
    Note that for $m=2$, we have that $\theta_{\mathrm{sqK}}\prof = \argmin_{\theta \in S^{d-1}} \sum_{i=1}^n \alphai \dang(\thetai, \theta)$. Let $\prof$ be the Antipodal Construction instantiated with $\alpha^{(1)} > 0.5$ and $\alpha^{(2)}  = 1-\alpha^{(1)} < 0.5$. Minimizing the weighted sum of \textit{unsquared} angular distances recovers exactly the majority vector (in this case $\thetaone$). Hence, $\theta_{\mathrm{sqK}}\prof$ gives voter type 2 individual proportionality level of 0. 

    However, for large $m$, the final term dominates, and minimizing this objective is equivalent to finding $\thang$. We know then that in the regime of large $m$, $\theta_{\mathrm{sqK}}\prof \approx  \thang\prof$. By \Cref{lem:agreeapprox}, this tells us that $\theta_{\mathrm{sqK}}\prof$ inherits the good proportionality guarantees of $\thang$ for spherically symmetrical item distributions.

    Interestingly, this points to a potential ``phase shift'' where minimizing the expected Squared Kemeny for small batch sizes yields an overly majoritarian scoring vector, but doing so for large batch sizes yields a vector that gives strong proportionality guarantees.
    
    \end{proof}

\subsection{Robustness} \label{app:robustness}
\begin{restatable}{proposition}{agreeapprox}\label{lem:agreeapprox} For any $m \in \mathbb{N}$, any spherically symmetric batch distribution $D_\circ^m$, any profile $(\btheta, \balpha)$, $\theta, \psi \in S^d$ such that $\dang(\theta, \psi) \leq \gamma$, 
\[\E_{X \sim \D_\circ^m}\left[\agi(\theta,  X,\prof)\right] \geq \E_{X \sim D_\circ^m}[\agi(\psi,X,\prof)] - \gamma/(\pi \alphai) \qquad \forall i \in [n].\]
\end{restatable}
\begin{proof}
    By using the identity established in \Cref{lem:ktangle} and the reverse triangle inequality we get
    \begin{align*}
         \E_{X \sim \D^m}\left[\agi(\theta, X, \prof)\right] &= \E_{X \sim \D^m} \left[ \frac{\aKT(\thetai, \theta, X)}{\alpha^{(i)}\binom{m}{2}}\right] \\
         &=  \frac{1}{\alpha^{(i)} \binom{m}{2}}\E_{X \sim \D^m}\left[ \aKT(\thetai, \theta, X)\right] \\
         &=  \frac{1}{\alpha^{(i)} \binom{m}{2}}\E_{X \sim \D^m}\left[ \frac{\pi - \dang(\thetai, \theta)}{\pi} \cdot \binom{m}{2}\right] \\
         &\geq \frac{1}{\alpha^{(i)} \binom{m}{2}}\E_{X \sim \D^m}\left[\frac{\pi - \dang(\thetai, \psi) - \gamma}{\pi} \cdot \binom{m}{2}\right] \\
         &= \frac{1}{\alpha^{(i)} \binom{m}{2}}\E_{X \sim \D^m}\left[\frac{\pi - \dang(\thetai, \psi) }{\pi} \cdot \binom{m}{2}\right] - \frac{\gamma}{\pi \alpha^{(i)}} \\
         &=  \E_{X \sim \D^m} [\agi(\psi, X, \prof)] - \frac{\gamma}{\pi \alpha^{(i)}}.
     \end{align*}
\end{proof}
\section{Appendix for \Cref{sec:cost-of-consistency}} \label{app:cost-of-consistency}

\subsection{Failure of Stronger Per-Batch Proportionality}

\begin{example}[\textbf{Impossibility of IP Level $\geq$ 1 on some $\mathbf{X}$}]\label{ex:batch-prop-impossibility}
    Fix any aggregation mechanism $\theta_f$, and let $\prof$ be the Antipodal Construction  (\Cref{fig:antipodal-prof}) instantiated with any weights. Let $X \in \cX^m$ be a batch where all items are co-linear on some vector $v$ that is not orthogonal to any of $\thetaone,\thetatwo,\theta_f\prof$. Then, a ranking over $X$ by a linear scoring rule is entirely determined by the sign of its inner product with $v$. As $\thetaone = -\thetatwo$, they have oppositely signed inner products with $v$. Without loss of generality, assume that $\text{sign}(\langle \theta_f\prof, v \rangle) = -\text{sign}(\langle \thetaone, v \rangle)$, then $\agi(\theta_f, X, \prof) = 0$.
\end{example}

\subsection{Proof of \Cref{prop:prop_relationship}}\label{app:subsec-prop-relationship}

\restatehere{proprelationship}
\begin{proof}
Let $\prof \in P_{n,d}$, and observe that for any $X \in \cX^m$ 
\begin{align*}
    \min_{i \in [n]} \agi(\theta_f, X, \prof) &\leq \ip_j(\theta_f, X, \prof) \quad \forall j \in [n]
    \intertext{Because this is true for all $X \in \cX^m$, this inequality is preserved when we take the expectation on both sides over $X \sim \cD^m$:}
    \mathbb{E}_{X \sim \cD^m}\left[\min_{i \in [n]} \agi(\theta_f, X, \prof)\right] &\leq \mathbb{E}_{X \sim \cD^m}\left[\ip_j(\theta_f, X, \prof)\right] \quad \forall j \in [n]
    \intertext{As this inequality is true for all choices of $j \in [n]$ on the RHS, it is also true for the minimizing choice of $j$. Then we have that}
    \mathbb{E}_{X \sim \cD^m}\left[\min_{i \in [n]} \agi(\theta_f, X, \prof)\right] &\leq \min_{j \in [n]}\mathbb{E}_{X \sim \cD^m}\left[\ip_j(\theta_f, X, \prof)\right]
    \intertext{Hence, }
    \pcol(\theta_f, \cD^m, \prof) &\leq \pind(\theta_f, \cD^m, \prof)
\end{align*}
\end{proof}

\subsection{Proof of \Cref{thm:fixed-cpop-UB}}\label{app:subsec-cpop-UB}
\restatehere{thmcpopUB}
\begin{proof}
    
    As with the other upper bound proofs, we let $\prof$ be the Antipodal Construction ($\Cref{fig:antipodal-prof}$) and let $\alpha^{(1)} = \frac{1}{\binom{m}{2}}, \alpha^{(2)} = \frac{\binom{m}{2}-1}{\binom{m}{2}}$.  Denote $\theta_f\prof$ as $\theta_f$ for the remainder of this proof. For notational concision, we define $Z(X)$ to be the minimum of voter 1 and 2's individual proportionality level \begin{align*}
        Z(X) &\coloneqq \min \left(\frac{a_{KT}(\thetaone, \theta_f, X)}{ \alpha^{(1)} \binom{m}{2}}, \frac{a_{KT}(\thetatwo, \theta_f, X)}{ \alpha^{(2)} \binom{m}{2}}\right)\\
        &= \min\left(a_{KT}(\thetaone, \theta_f, X), \frac{a_{KT}(\thetatwo, \theta_f, X)}{\binom{m}{2}-1}\right)
    \end{align*}
    Then, \[
    \pcol(\theta_f, \cD^m, \prof) = \mathbb{E}_{X \sim \cD^m} [Z(X)]
    \]
   Recall from $\Cref{cor:no-ties}$ that we can evaluate this expectation assuming that there are no scoring ties for any $\theta \in \boldsymbol{\theta}$ or $\theta_f$ over the item batch sampled.

    \textbf{Deterministic upper bound.}
    First, we will show that for any $X \in \cX^m$, with no scoring ties, we have that $Z(X) \leq 1$. Recalling via \Cref{prop:antipodal-disagreement} that $a_{KT}(\thetatwo, \theta_f, X) = \binom{m}{2} - a_{KT}(\thetaone, \theta_f, X)$ \[
    \min \left(a_{KT}(\thetaone, \theta_f, X), \frac{\binom{m}{2} - a_{KT}(\thetaone, \theta_f, X)}{\binom{m}{2}-1}\right) \leq 1
    \]
    Observe that $a_{KT}(\thetaone, \theta_f, X)$ is a nonnegative integer. If $a_{KT}(\thetaone, \theta_f, X) = 0$, then the minimum of the two terms is 0. If $a_{KT}(\thetaone, \theta_f, X) \geq 1$, then $\frac{\binom{m}{2} - a_{KT}(\thetaone, \theta_f, X)}{\binom{m}{2}-1} \leq 1$, so the minimum of the two terms is at most 1.

    Now that we know that $Z(X) \leq 1$ \textit{deterministically} for all $X$ with no ties, we will show that with strictly positive probability, $Z(X) = 0$.

    \textbf{Showing positive probability on a near-collinear batch.}
    First observe that if $\theta_f = \thetatwo$, then $a_{KT}(\thetaone, \theta_f, X) = 0$ for all $X$ with no ties, and we immediately have that $Z(X) = 0$ with positive probability (in fact, with probability 1).  
    
    Assume then that $\theta_f \neq \thetatwo$. This means that $\theta_f$ and $\thetaone$ are not antipodal -- and as they are both unit vectors, their sum is nonzero. So, if we let $v$ be the unit vector bisecting them, we have that \[
    v = \frac{\thetaone + \theta_f}{\|\thetaone + \theta_f\|}
    \]
    and it follows that \[
    \langle \thetaone, v \rangle = \frac{1 + \langle \thetaone, \theta_f\rangle}{\| \thetaone + \theta_f \|} = \langle \theta_f, v \rangle.
    \]
    We remark that here and throughout this proof, whenever there is a norm being used (also to define a ball), it is the 2-norm. Let $\gamma \coloneqq  \langle \thetaone, v \rangle = \langle \theta_f, v \rangle$ and observe that $\gamma > 0$ because $\langle \thetaone, \theta_f \rangle > -1$. 

    The high-level idea is diagrammed in \Cref{fig:cprop-UB-diagram}, and is as follows: if we sample $X = (x_1, \dots, x_m)$ such that they are all along $v$ (or some translation/scaling thereof), then $\theta_f$ and $\thetaone$ will agree on every pairwise comparison in $X$. Consequently, $\thetatwo$ will have KT-agreement of 0 with $\theta_f$ on $X$, and $Z(X) = 0$. However, this outcome has probability measure 0. Therefore, we have to allow for some wiggle room around all of the points lying exactly on $v$ so that such a batch is sampled with positive probability. We construct this wiggle room in such a way that for any items in it, their pairwise difference vectors are still within the region of agreement between $\thetaone$ and $\theta_f$.

    \begin{figure}[htbp]
        \centering
        \resizebox{0.6\linewidth}{!}{
            \input{.//Figures/batch_prop_UB_diagram}
        }
        \caption{Proof visualization for \Cref{thm:fixed-cpop-UB}.}
        \label{fig:cprop-UB-diagram}
    \end{figure}

    By our assumption on $\cD$, there exists an open set $U \subseteq \mathbb{R}^d$ in which $f_{\cD}$ is strictly positive and continuous. Fix some point $y \in U$ and a radius $R > 0$ such that $\overline{B}_{R}(y) \subset U$. Let $\eta = \min_{x \in \overline{B}_{R}(y)} f_{\cD}(x) > 0$ be the minimum probability density in this ball. Now we define an ideal series of points for the $m$ items that are perfectly aligned along $v$. For $k \in \{1, \dots, m\}$, define $x_k^*$: \[
    x_k^* = y + \left(\frac{m-2k}{2m}\right)R \cdot v
    \]
     Now let $\varepsilon = \frac{R\gamma}{3m}$. We will show two claims: (1) for any batch $X=(x_1, \dots, x_m)$ such that $x_i \in B_{\varepsilon}(x^*_k)$ for all $k \in \{1, \dots, m\}$, $a_{KT}(\thetatwo, \theta_f, X) = 0$ and consequently $Z(X) = 0$ and(2) the probability of drawing such a batch is strictly positive.

     To show claim (1), fix some $X=(x_1, \dots, x_m)$ such that $x_i \in B_{\varepsilon}(x^*_k)$ for all $k \in \{1, \dots, m\}$ and consider $i, j \in [m]$ with $i < j$. We have that $x_i = x_i^* + \delta_i$ and $x_j = x_j^* + \delta_j$ for some $\delta_i, \delta_j \in \mathbb{R}^d$ with $\|\delta_i\|, \|\delta_j\| \leq \varepsilon$. Then: \begin{align*}
        \langle \thetaone, x_i - x_j \rangle &= \langle \thetaone, x_i^* - x_j^* \rangle + \langle \thetaone, \delta_i - \delta_j\rangle\\
        \intertext{By observing that for any $i < j$, $x_i^* - x_j^* = \frac{(j-i)R}{m} \cdot v$ and applying Cauchy-Schwarz to the second term, we can rewrite as follows. }
        &\geq \frac{(j-i)R}{m}\langle \thetaone, v\rangle - \underbrace{\|\thetaone \|}_1 \underbrace{\|\delta_i - \delta_j\|}_{\leq 2\varepsilon}\\
        &\geq \frac{(j-i) R \gamma}{m} - 2\varepsilon\\
        &= \frac{(j-i) R \gamma}{m} - \frac{2R\gamma}{3m} \geq \frac{R\gamma}{3m} > 0
    \end{align*}
    By the exact same computation, we get that $\langle \theta_f, x_i - x_j \rangle > 0$ --- all we needed is that $\|\thetaone\| = 1$ and $\langle \thetaone, v \rangle = \gamma$, and both hold when we substitute in $\theta_f$ instead. Hence, we have that for any $i < j$, both $x_i \succ_{\thetaone}^X x_j$ and $x_i \succ_{\theta_f}^X x_j$. Therefore, $a_{KT}(\thetaone, \theta_f, X) = \binom{m}{2}$, so $a_{KT}(\thetatwo, \theta_f, X) = 0$ and $Z(X) = 0$.
     
    To show claim (2), first we'll argue that for all $k \in [m]$, $B_{\varepsilon}(x^*_k) \subset \overline{B}_R(y)$. Fix some $k \in [m]$, and some $x_k \in B_{\varepsilon}(x_k^*)$. By the triangle inequality, we have that \[
    \|x_k - y\| \leq \|x_k - x_k^*\| + \|x_k^* - y\| \leq \frac{R\gamma}{3m} + \left\|\left(\frac{m-2k}{2m}\right)R \cdot v\right\| \leq \frac{R}{3m} + R\left(\frac{|m-2k|}{2m}\right)
    \] 
    where we used the fact that $\|v\| =1$ and $\gamma \leq 1$. Now observe that $|m-2k| \leq m$ (achieved when $k = m$), and so we have that $\|x_k - y\| \leq \frac{R}{3m} + \frac{R}{2} < R$ as $m \geq 3$. Therefore, $x_k \in \overline{B}_R(y)$. As this was true for arbitrary $k$ and arbitrary $x_k \in B_{\varepsilon}(x_k^*)$, we can say that for all $k \in \{1, \dots, m\}$, $B_{\varepsilon}(x^*_k) \subset \overline{B}_R(y)$. Now recall that $f_{\cD}(x) \geq \eta > 0$ for all $x \in \overline{B}_R(y)$ and hence all $x \in B_{\varepsilon}(x^*_k)$ for some $k \in [m]$. Then we have that \[
    \Pr_{X \sim \cD^m}\left(x_k \in B_{\varepsilon}(x^*_k) \quad \forall k \in [m]\right) = \prod_{k \in [m]} \Pr_{x_k \sim \cD} \left(x_k \in B_{\varepsilon}(x^*_k)\right) \geq \eta^m (\text{Vol}(B_{\varepsilon}))^m > 0
    \]

    \textbf{Putting it all together.} By conditioning on whether or not we witness a ``clustered'' item batch, we can derive the desired upper bound on the expectation. Let $C$ be the event that $x_k \in B_{\varepsilon}(x^*_k) \quad \forall k \in [m]$. Then, \begin{align*}
        \mathbb{E}_{X \sim \cD^m} [Z(X)]
        &= \mathbb{E}_{X \sim \cD^m} [Z(X) \vert C] \Pr(C) + \mathbb{E}_{X \sim \cD^m} [Z(X) \vert \overline{C}] \Pr(\overline{C})\\
        &= 0 \cdot \Pr(C) + 1 \cdot (1-\underbrace{\Pr(C)}_{> 0}) < 1
    \end{align*}
    
\end{proof}

\subsection{Proof of \Cref{thm:fixed-cpop-LB}}\label{app:subsec-cpop-LB}

\restatehere{thmcpopLB}

First, define a ``deviation quantity'' for voter type $i$:\[
W_i \coloneqq \mathbb{E}_{X \sim \cD^m}\left[\agi(\theta_f, \prof, X)\right] - \agi(\theta_f, \prof, X).
\]
The $W_i$ captures how much the normalized individual agreement for voter type $i$ falls short of the expected normalized individual agreement. Observe that we can then rewrite $\min_{i \in [n]} \agi(\theta_f, \prof, \cD)$ in terms of expected normalized individual agreement minus this deviation term as follows:
\begin{align*}
    \min_{i \in [n]} \agi(\theta_f, \prof, X)
    &= \min_{i \in [n]} \left(\mathbb{E}_{X \sim \cD^m}[\agi(\theta_f, \prof, X)] - W_i\right)\\
    &\geq \min_{i \in [n]} \mathbb{E}_{X \sim \cD^m}[\agi(\theta_f, \prof, X)] - \max_{i \in [n]}W_i
    \intertext{Taking the expectation of both sides, applying linearity of expectation, and observing that $\mathbb{E}[\min_{i \in [n]} \mathbb{E}[\agi(\theta_f, \prof, X)]] = \min_{i \in [n]} \mathbb{E}[\agi(\theta_f, \prof, X)]$ as the inner expectation is a constant yields:}
    \pcol(\theta_f, \prof, \cD^m) &\geq \pind(\theta_f, \prof, \cD^m) - \mathbb{E}[\max_{i \in [n]}W_i]
\end{align*}

Hence, the remaining task is upper bounding $\mathbb{E}_{X \sim \cD^m}[\max_{i \in [n]} W_i]$. We show two different upper bounds, one that has a dependence on $n$ and the other that has a dependence on $d$, which results in the minimum term in the final expression. 

Both upper bounds rely on reformulating normalized individual agreement as a U-statistic (with an addition $\frac{1}{\alpha_i}$ factor). Recall that \[
\agi(\theta_f, \prof, X) = \frac{a_{KT}(\thetai, \theta_f, X, \prof)}{\alpha_i\binom{m}{2}}
\] 
As $\theta_f$ and $\prof$ are fixed, we omit the function inputs in the rest of this proof for concision, and just write $\agi$ to refer to the quantity above. 
Define a kernel function $h_{\thetai} \colon \R^d \rightarrow \{0,1\}$ that takes in a difference vector (between two items) and indicates whether $\thetai$ and $\theta_f\prof$ have the same sign of dot product with this difference vector \[
h_{\thetai}(z) \coloneqq \mathbb{I}\{\text{sign}(\langle \thetai, z\rangle) = \text{sign}(\langle \theta_f\prof, z\rangle)\}
\]
Recalling that agreement between scoring vectors on a pairwise comparison between items just comes down to whether the two scoring vectors have the same signed dot product with the item difference vector, we can rewrite the $\agi$ as \[
    \agi = \frac{1}{\alpha_i\binom{m}{2}}\sum_{1 \leq j < \ell \leq m} h_{\thetai}(x_j- x_{\ell})
    \]
    For all permutations $\sigma \in \Pi_m$, define \[
    V_{i, \sigma} = \frac{1}{\lfloor m/2 \rfloor}\left(h_{\thetai}(x_{\sigma(1)} - x_{\sigma(2)}) + \dots + h_{\thetai}(x_{\sigma(2\lfloor m/2 \rfloor-1)}- x_{\sigma(2\lfloor m/2 \rfloor)})\right)
    \]
    Then we can rewrite $\agi$ again, this time as a function of the $V_{i, \sigma}$. 
    \[
    \agi = \frac{1}{\alpha_i \cdot m!} \sum_{\sigma \in \Pi_m} V_{i, \sigma}
    \]
    The key observation is that every $V_{i, \sigma}$ is the sum of $\lfloor m/2 \rfloor$ \textit{independent} variables, because each term is the application of the kernel function to a disjoint pair of items. Therefore, we can more easily apply standard results to show that the all of the $V_{i, \sigma}$ concentrate around their expectation quickly, and translate that into bounds on the deviation of $\agi$ from its expectation ($W_i$).

Now the two upper bound methods diverge. For the first bound, resulting in a dependence on $n$, we use the reformulation of $\agi$ to show that the $W_i$ are all sub-Gaussian in \Cref{lem:W_sub_Gauss}. Then, we apply a known bound on the maximum of sub-Gaussian random variables in \Cref{lem:max_fluc}. 

The second method, found in \Cref{lem:dev_ub_d}, bounds the maximum $W_i$ by taking the supremum over the deviation quantity for \textit{any} scoring vector in $S^{d-1}$ and bounding this using Rademacher complexity. This bound supports the intuition that even if there are many different voter types, the diversity among scoring vectors is limited by the dimension. Scoring vectors that are close together in space have positively correlated deviation quantities --- on any given batch, they are likely experience similar deviations from their expected individual proportionality.

\begin{lemma}\label{lem:W_sub_Gauss}
    For all $i \in [n]$, $W_i = \mathbb{E}[\agi] - \agi$ is sub-Gaussian with parameter $\sigma^2 = \frac{1}{4\lfloor m/2 \rfloor\alpha_i^2}$.
\end{lemma}

\begin{proof}
    In order to show sub-Gaussianity, we want to show that $\mathbb{E}[e^{\lambda(W_i - \mathbb{E}[W_i])}] \leq e^{\sigma^2\lambda^2/2}$ for all $\lambda \in \R$. Observe that $\mathbb{E}[W_i] = 0$.
    To bound the MGF,
    \begin{align*}
        \mathbb{E}[e^{\lambda(W_i - \mathbb{E}[W_i])}] &= \mathbb{E}[e^{\lambda(\mathbb{E}[\agi] - \agi)}]\\
        &= \mathbb{E}\left[\exp\left(\frac{\lambda}{\alpha_i}\sum_{\sigma \in \Pi_m} \frac{1}{ m!} ( \alpha_i \mathbb{E}[\agi] - V_{i, \sigma})\right)\right]
        \intertext{By Jensen's inequality:}
        &\leq \mathbb{E}\left[\sum_{\sigma \in \Pi_m} \frac{1}{m!} \exp\left(\frac{\lambda}{\alpha_i}(\alpha_i \mathbb{E}[\agi] - V_{i, \sigma})\right)\right]
        \intertext{Then by linearity of expectation and the observation that all $V_{i, \sigma}$ have the same distribution (so the statement holds for any permutation $\sigma$):}
        &= \mathbb{E}\left[\exp\left(\frac{\lambda}{\alpha_i}(\alpha_i \mathbb{E}[\agi] - V_{i, \sigma})\right)\right]\\
        &= \mathbb{E}\left[\exp\left(\frac{-\lambda}{\alpha_i}(V_{i, \sigma} - \alpha_i \mathbb{E}[\agi])\right)\right]
        \intertext{Note that $\alpha_i  \mathbb{E}_{x_1, \dots, x_m}[\agi] = \mathbb{E}_{x_1, x_2}[h_{\thetai}(x_1, x_2)] = \mathbb{E}_{x_1, \dots, x_m}[V_{i, \sigma}]$. So the expression $V_{i, \sigma} - \alpha_i \mathbb{E}[\agi]$ can be rewritten as $V_{i, \sigma} - \mathbb{E}[V_{i, \sigma}]$. Now we can see that the task has come down to bounding the MGF of $V_{i, \sigma}$. Recalling the definition of $V_{i, \sigma}$ and plugging in yields:}
         &= \mathbb{E}\left[\exp\left(\frac{-\lambda}{\alpha_i}\left(\sum_{j=1}^k \frac{h_{\thetai}(x_{\sigma(2j-1)}- x_{\sigma(2j)})}{k} - \frac{\mathbb{E}[V_{i, \sigma}]}{k}\right)\right)\right]\\
        &= \mathbb{E}\left[\prod_{j=1}^k\exp\left(\frac{-\lambda}{\alpha_i}\left(\frac{h_{\thetai}(x_{\sigma(2j-1)}- x_{\sigma(2j)})}{k} - \frac{\mathbb{E}[V_{i, \sigma}]}{k}\right)\right)\right]
        \intertext{Now we can utilize the independence of $h_{\thetai}(x_{\sigma(2j-1)}, x_{\sigma(2j)})$ for different values of $j$. Recall that these are independent because each of the $\lfloor m/2\rfloor$ terms are functions of disjoint subsets of the randomly drawn $x_1, \dots, x_m$.}
        &= \prod_{j=1}^k \mathbb{E}\left[\exp\left(\frac{-\lambda}{\alpha_i}\left(\frac{h_{\thetai}(x_{\sigma(2j-1)}- x_{\sigma(2j)})}{k} - \frac{\mathbb{E}[V_{i, \sigma}]}{k}\right)\right)\right]
        \intertext{Moving the factor of $1/k$ outside and recalling that $\mathbb{E}[V_{i, \sigma}] = \mathbb{E}\left[h_{\thetai}(x_{\sigma(2j-1)}, x_{\sigma(2j)})\right]$ gives us:}
        &= \prod_{j=1}^k \mathbb{E}\left[\exp\left(\frac{-\lambda}{\alpha_i k}\left(h_{\thetai}(x_{\sigma(2j-1)}- x_{\sigma(2j)}) - \mathbb{E}[V_{i, \sigma}]\right)\right)\right]
        \intertext{Then, we can apply Hoeffding's lemma to each of these $k$ terms, which states that for a random variable $X$ that is bounded in $[a,b]$ w.p. 1, we have that $\mathbb{E}[\exp(c(X-\mathbb{E}[X]))] \leq e^{c^2(b-a)^2/8}$ for any $c \in \R$. Observe that $h_{\thetai}(x_{\sigma(2j-1)}- x_{\sigma(2j)}) \in [0,1]$, hence:}
        &\leq \prod_{j=1}^k e^{\lambda^2/(8\alpha_i^2k^2)}= e^{\lambda^2/(8\alpha_i^2k)}
    \end{align*}
Putting it altogether and plugging in for $k$, we have that:
    \begin{align*}
        \mathbb{E}[e^{\lambda(W_i - \mathbb{E}[W_i])}] &\leq e^{\frac{\lambda^2}{8\alpha_i^2\lfloor m/2 \rfloor}}
    \end{align*}
    So, $W_i$ is sub-Gaussian with parameter $\sigma^2 = \frac{1}{4\alpha_i^2\lfloor m/2 \rfloor}$.
\end{proof}

\begin{lemma}\label{lem:max_fluc}
    We can bound the expected max fluctuation in individual proportionality as follows: \[
    \mathbb{E}_X[\max_{i \in [n]} W_i] = \mathbb{E}_X\left[\max_{i \in [n]}\left(\mathbb{E}[\agi] - \agi\right)\right] \leq \frac{1}{\alpha_{min}} \sqrt{\frac{\log (n)}{2\lfloor m/2 \rfloor}}
    \]
\end{lemma}

\begin{proof}
    In combination with our result showing that the individual $W_i$ are sub-Gaussian in \Cref{lem:W_sub_Gauss}, we can directly apply a known bound on the maximum of  (not necessarily independent) sub-Gaussian RVs (e.g. \citet{rigollet2023high} Thm 1.14). We can upper bound the variance parameter of all of these individual variables by $\max_{i \in [n]} \sigma^2_i \leq \max_{i \in [n]} \frac{1}{4\lfloor m/2 \rfloor\alpha_i^2} = \frac{1}{4\lfloor m/2 \rfloor\alpha_{min}^2}$. So we have that \[
    \mathbb{E}_X[\max_{i \in [n]} W_i] \leq \sqrt{2\sigma^2 \log(n)} \leq \frac{1}{\alpha_{min}} \sqrt{\frac{\log (n)}{2\lfloor m/2 \rfloor}}
    \]
\end{proof}

\begin{lemma}\label{lem:dev_ub_d}
We can bound the expected maximum deviation in individual proportionality in terms of $d$ as follows: \[
\mathbb{E}_{X \sim \cD^m}\left[\max_{i \in [n]} W_i\right] \leq \frac{1}{\alpha_{min}}\sqrt{\frac{d \cdot 8\log(e\lfloor m / 2 \rfloor/d)}{\lfloor m/2 \rfloor}}
\]
\end{lemma}

\begin{proof}
    Observe that we can upper bound the maximum deviation experienced by any voter type, by considering the maximum deviation that would be experienced for \textit{any} scoring vector on the sphere, if it were associated with the minimum weight $\alpha_{min} = \min_{i \in [n]} \alpha_i$. \begin{align*}
         \mathbb{E}\left[\max_{i \in [n]} W_i\right]
         &= \mathbb{E}\left[\max_{i \in [n]} \left(\mathbb{E}\left[\frac{a_{KT}(\thetai, \theta_f, X, \prof)}{\alpha_i \binom{m}{2}}\right] - \frac{a_{KT}(\thetai, \theta_f, X, \prof)}{\alpha_i \binom{m}{2}}\right)\right]\\
         &\leq \mathbb{E}\left[\sup_{\theta \in S^{d-1}} \left(\mathbb{E}\left[\frac{a_{KT}(\theta, \theta_f, X, \prof)}{\alpha_{min} \binom{m}{2}}\right] - \frac{a_{KT}(\theta, \theta_f, X, \prof)}{\alpha_{min} \binom{m}{2}}\right)\right]
    \end{align*}
    We can use our method of rewriting individual proportionality for an arbitrary vector $\theta \in S^{d-1}$: \[
    \frac{a_{KT}(\thetai, \theta_f, X, \prof)}{\binom{m}{2}} = \frac{1}{ m!} \sum_{\sigma \in \Pi_m} \frac{1}{\lfloor m /2 \rfloor}\sum_{j=1}^{\lfloor m/2 \rfloor} h_{\theta}(x_{\sigma(2j-1)}- x_{\sigma(2j)})
    \]
    recalling that $h_{\theta}(x-y)$ is the kernel function, and is an indicator function for whether $\theta$ and $\theta_f\prof$ agree on the ranking of $x$ and $y$. Observe that all of the $x_i$ are iid, so for any permutation $\sigma$, all of the difference vectors $x_{\sigma(2j-1)}- x_{\sigma(2j)}$ are also identically distributed. Then, by linearity of expectation $\mathbb{E}\left[\frac{1}{\lfloor m /2 \rfloor}\sum_{j=1}^{\lfloor m/2 \rfloor} h_{\theta}(x_{\sigma(2j-1)}- x_{\sigma(2j)})\right] = \mathbb{E}[h_{\theta}(x_1-x_2)]$. We can continue in our bound on $\mathbb{E}\left[\max_{i \in [n]} W_i\right]$ follows: \begin{align*}
        &\leq \frac{1}{\alpha_{min}}\mathbb{E}\left[\sup_{\theta \in S^{d-1}} \left(\mathbb{E}\left[\frac{a_{KT}(\theta, \theta_f, X, \prof)}{\binom{m}{2}}\right] - \frac{a_{KT}(\theta, \theta_f, X, \prof)}{ \binom{m}{2}}\right)\right]\\
        &= \frac{1}{\alpha_{min}}\mathbb{E}\left[\sup_{\theta \in S^{d-1}} \frac{1}{ m!} \sum_{\sigma \in \Pi_m}\left(\mathbb{E}[h_{\theta}(x_1-x_2)] -  \frac{1}{\lfloor m /2 \rfloor}\sum_{j=1}^{\lfloor m/2 \rfloor} h_{\theta}(x_{\sigma(2j-1)}- x_{\sigma(2j)})\right)\right]
        \intertext{The permutation average can be thought of as the expectation over a permutation drawn uniformly at random. Then, by the convexity of supremum and Jensen's inequality, we can upper bound this expression by pulling the supremum inside the permutation average:}
        &\leq \frac{1}{\alpha_{min}}\mathbb{E}\left[\frac{1}{ m!} \sum_{\sigma \in \Pi_m} \sup_{\theta \in S^{d-1}} \left(\mathbb{E}[h_{\theta}(x_1-x_2)] -  \frac{1}{\lfloor m /2 \rfloor}\sum_{j=1}^{\lfloor m/2 \rfloor} h_{\theta}(x_{\sigma(2j-1)}- x_{\sigma(2j)})\right)\right]
        \intertext{Note that this supremum is the same for all permutations, as all of them share the same distribution. So, we can fix any permutation $\sigma$, and write our expression in terms of it:}
        &= \frac{1}{\alpha_{min}}\mathbb{E}\left[\sup_{\theta \in S^{d-1}} \left(\mathbb{E}[h_{\theta}(x_1-x_2)] -  \frac{1}{\lfloor m /2 \rfloor}\sum_{j=1}^{\lfloor m/2 \rfloor} h_{\theta}(x_{\sigma(2j-1)}- x_{\sigma(2j)})\right)\right]
        \intertext{If we let $\mathcal{H} = \{h_{\theta} \colon \theta \in S^{d-1}\} = \{z \mapsto \mathbb{I}(\text{sign}(\langle z, \theta \rangle) = \text{sign}(\langle z, \theta_f\prof \rangle)\colon \theta \in S^{d-1}\}$ be our function class, then it is a known result that we can upper bound this expected supremum by twice the Rademacher complexity of $\mathcal{H}$ (e.g. \citep{wainwright2019high} Prop 4.11).}
        &\leq \frac{2}{\alpha_{min}} \mathcal{R}_{\lfloor m/2 \rfloor}(\mathcal{H})
    \end{align*}
    In order to bound $\mathcal{R}_{\lfloor m/2 \rfloor}(\mathcal{H})$, we first bound the VC-dimension of $\mathcal{H}$ and then use a known connection between the two. 
    
    Consider the class of linear halfspaces that pass through the origin, $\mathcal{F} = \{ z \mapsto \text{sign}(\langle \theta, z \rangle) \mid \theta \in S^{d-1} \}$. Because our function class $\mathcal{H}$ is formed by checking equality against a fixed boolean label pre-assigned by $\theta_f\prof$, $\mathcal{H}$ cannot shatter any set of points that $\mathcal{F}$ cannot shatter. While the VC dimension of general affine halfspaces in $\mathbb{R}^d$ is $d+1$ (\citep{mohri2018foundations}, Example 3.12 and Thm 3.13), our class $\mathcal{F}$ has the additional restriction that the halfspaces must pass through the origin, making it geometrically isomorphic to the class of general affine halfspaces in $\mathbb{R}^{d-1}$. Therefore, $\text{VCdim}(\mathcal{F}) = d$. Thus, $\text{VCdim}(\mathcal{H}) \le \text{VCdim}(\mathcal{F}) = d$.

    Finally, combining Massart's lemma and Sauer's lemma (and their implications in \citep{mohri2018foundations} Corollaries 3.8 and 3.18), we get that \[
    \frac{2}{\alpha_{min}}\mathcal{R}_{\lfloor m/2 \rfloor}(\mathcal{H}) \leq \frac{2}{\alpha_{min}}\sqrt{\frac{2\log((e\lfloor m/2 \rfloor/d)^d)}{\lfloor m/2 \rfloor}} = \frac{1}{\alpha_{min}}\sqrt{\frac{d \cdot 8\log(e\lfloor m / 2 \rfloor/d)}{\lfloor m/2 \rfloor}}
    \]
\end{proof}

Returning to the proof of \Cref{thm:fixed-cpop-LB}, we can now plug in our upper bounds on $\mathbb{E}[\max_{i \in [n]} W_i]$ from \Cref{lem:W_sub_Gauss} and \Cref{lem:dev_ub_d}, and recover the desired bound: \[
\pcol(\theta_f, \prof, \cD^m) \geq \pind(\theta_f, \prof, \cD^m) - \frac{1}{\alpha_{min}}\sqrt{\frac{\min(\log(n), d \cdot 16\log(e\lfloor m / 2 \rfloor/d))}{2\lfloor m/2 \rfloor}}
\]

\subsection{Tail Bound on $i$'s Individual Proportionality Level}
\begin{corollary} \label{cor:concentration-bound}
    For any $\theta_f, \cD^m, \prof, i \in \prof$, and any $t > 0$ we have that \[
\Pr_{X \sim \cD^m}\left(\agi(\theta_f, X, \prof) < \mathbb{E}[\agi(\theta_f, X, \prof)] - t\right) \leq e^{-2\lfloor m/2 \rfloor\alpha_i^2t^2}
\]
\end{corollary}

\begin{proof}
    We can reformulate this as a  tail bound on $W_i$: \begin{align*}
        &\Pr_{X \sim \cD^m}\left(\agi(\theta_f, X, \prof) < \mathbb{E}[\agi(\theta_f, X, \prof)] - t\right)\\
        &= \Pr_{X \sim \cD^m}\left(\mathbb{E}[\agi(\theta_f, X, \prof)] - \agi(\theta_f, X, \prof) > t\right)\\
        &= \Pr_{X \sim \cD^m}(W_i > t)\\
        \intertext{We know from \Cref{lem:W_sub_Gauss} that $W_i$ is sub-Gaussian with variance parameter $\sigma^2 = \frac{1}{4\lfloor m/2 \rfloor \alpha_i^2}$. This immediately gives us the following tail bound (e.g. \citep{wainwright2019high} Equation 2.9):}
        &\leq e^{\frac{-t^2}{2\sigma^2}}\\
        &= e^{-1\lfloor m/2 \rfloor \alpha_i^2t^2}
    \end{align*}
\end{proof}

\section{Appendix for \Cref{sec:empirical}}\label{app:exp}

Note that, on a 12-core machine, running all experiments takes approximately one hour. No special computing resources are needed.

\subsection{Approximation of $\thar$ to $\thang$ Under High Agreement} \label{app:approx-thar}
\begin{restatable}{proposition}{arithangapprox}\label{thm:arith_ang_approx}
    Fix any profile $(\btheta, \balpha)$ and $R \in [0, \pi/2)$  for which there exists $\theta \in S^{d-1}$ such that for all $i \in [n]$, $\dang(\theta, \theta^{(i)}) \leq R$. Then, $\dang(\thang, \thar) \leq \arctan \left( \frac{2R - \sin(2R)}{\cos(R)} \right).$
\end{restatable}

\begin{proof}
    Without loss of generality, suppose $\thang = e_1$ where $e_1 = (1, 0, \dots, 0) \in S^{d-1}$. 
    If we write $d_i \coloneq \dang(\thang, \thetai)$ we get that $\cos(d_i) = \langle \thang,  \thetai\rangle = \thetai_1$ for every $i \in \mathbb{N}$.
    Next we set $w_i \coloneq \thetai - e_1 \cos(d_i)$, and $v_i := w_i / ||w_i||$. Then $w_i = v_i \sin(d_i)$ only consists of the last $d-1$ components of $\thetai$, and $v_i$ is $w_i$ normalized. We use this to rewrite
    \[ \thetai = \begin{pmatrix}
        \cos(d_i) \\
        v_i \sin(d_i)
    \end{pmatrix}.\]

    The unnormalized arithmetic mean can then be written as
    \[\thar^* = \sum_{i = 1}^n \alpha^{(i)}\thetai = \sum_{i = 1}^n   \alpha^{(i)}\begin{pmatrix} \cos(d_i) \\ v_i \sin(d_i) \end{pmatrix} = \begin{pmatrix} \sum_{i = 1}^n \alpha^{(i)} \cos(d_i) \\ \sum_{i = 1}^n \alpha^{(i)} v_i \sin(d_i) \end{pmatrix}. \]
    Because $\thar$ is now just the normalized vector of $\thar^*$, the distance $\dang(\thar, \thang)$ is the angle between $\thang$ and $\thar^*$.

    \begin{figure}
        \centering
        \input{Figures/arith_ang_helper}
        \caption{\emph{Left:} $\thetai$ decomposes into $(\cos(d_i), v_i \sin(d_i))$, \emph{Right:} Through considering the right angled triangle formed through the origin, $\thar$ and the projection of $\thar$ to $\thang$, we can rewrite $\dang(\thang, \thar) = \arctan(\|B\| / A)$.}
        \label{fig:arith_ang_helper}
    \end{figure}

    We write $A = \sum_{i = 1}^n \alpha^{(i)} \cos(d_i)$, and $B = \sum_{i = 1}^n \alpha^{(i)} v_i \sin(d_i)$. From \Cref{fig:arith_ang_helper}, we can then see that we can therefore determine the angle $\dang(\thang, \thar)$ as $\arctan\left( \frac{||B||}{A} \right)$. Because $\arctan$ is monotonic, we will separately bound $A$ from below and $||B||$ from above.

\paragraph{Bounding $||B||$.} We will prove that $||B|| \leq (2R) - \sin(2R)$. To see this observe that
\begin{multline*} ||B|| = \left|\left| \sum_{i = 1}^n \alpha^{(i)} v_i \sin(d_i) \right|\right| = \left|\left| \sum_{i = 1}^n \alpha^{(i)} v_i (\sin(d_i) - d_i) \right|\right| \\ \leq \sum_{i = 1}^n \alpha^{(i)} ||v_i|| (d_i - \sin(d_i)) \leq \sum_{i = 1}^n \alpha^{(i)} (2R - \sin(2R)) = 2R - \sin(2R).
\end{multline*}
Where we use the local optimality condition in the second equality, as well as the fact that the function $f(x) = x - \sin(x)$ and $d_i = \dang(\thang, \thetai) \leq 2R$, since $\thang$ must be in the spherically convex hull of the $\thetai$ \citep{buss2001spherical}.
    
\paragraph{Bounding $A$.}
We will show that $A = \sum_{i = 1}^n \alpha^{(i)} \cos(d_i) \geq \cos(R)$. If we define $r:[0, \pi] \to \R, x \mapsto \cos(x) - \cos(R)$ and $r_i \coloneq \cos(d_i) \cos(R)$, then this is equivalent to proving
\begin{equation}\label{eqn:denominator}
    \sum_{i = 1}^n \alpha^{(i)} r_i > 0, 
\end{equation}
Now observe that $r$ is a decreasing function as $\cos(x)$ is decreasing on $[0, \pi]$. Due to Chebyshev's Sum Inequality,  it suffices to prove $\sum_{i = 1}^n \alpha^{(i)} w(d_i) r_i = \sum_{i = 1}^n \alpha^{(i)} w(d_i) r(d_i) \geq 0$ for some increasing function $w:[0, \pi] \to \R$ to show \Cref{eqn:denominator}. We will prove that this holds for the function $w(x) = x / \sin(x)$ (defining $w(0) = 1$) which is increasing on $[0, \pi]$. We write $w_i \coloneq w(d_i)$ as a shorthand.

It holds that
\begin{align*}
    \sum_{i = 1}^n \alpha^{(i)} w_i \thetai & = \left( \sum_{i = 1}^n \alpha^{(i)} w_i \cos(d_i)) \right) \thang + \sum_{i = 1}^n \alpha^{(i)} \sin(d_i) w_i v_i \\ 
                                         &= \left( \sum_{i = 1}^n \alpha^{(i)} w_i \cos(d_i)) \right) \thang + \sum_{i = 1}^n \alpha^{(i)} d_i v_i = \thang \sum_{i = 1}^n \alpha^{(i)} w_i \cos(d_i).
\end{align*}
Taking the scalar product with $u$ on both sides and using that $\dang(\thetai, u) < R$, so $\langle \thetai, u \rangle > \cos(R)$, as well as $\langle \thang, u \rangle \leq 1$ we get
\[  \sum_{i = 1}^n \alpha^{(i)} w_i \cos(R)  < \sum_{i = 1}^n \alpha^{(i)} w_i \langle \thetai, u \rangle = \langle \thang, u \rangle \sum_{i = 1}^n \alpha^{(i)} w_i \cos(d_i) \leq \sum_{i =1}^n \alpha^{(i)} w_i \cos(d_i), \]
and henceforth
\[ \sum_{i =1}^n \alpha^{(i)} w_i < \alpha^{(i)} w_i \cos(d_i) \iff 0 < \sum_{i =1}^n \alpha^{(i)} w_i (\cos(d_i) - \cos(R)) = \sum_{i = 1}^n \alpha^{(i)} w_i r_i .\]
\end{proof}

\subsection{Warm-up: Expected Proportionality in the Two-Voter Case} \label{app:2voters}
We begin in the simple setting that originally exposed the
disproportionality of the arithmetic mean: $n=2$ voter types in $d=2$
dimensions, with a majority group and a minority group
\citep{feffer2023moral}. We let voter type $1$ be the majority and
sweep its weight over
$\alpha_1 \in \{0.51, 0.60, 0.70, 0.80, 0.90, 0.95\}$; we let
$\phi := d_{\measuredangle}(\theta^{(1)}, \theta^{(2)}) \in [0\degree, 180\degree]$
denote the angle between the two scoring vectors. Together $(\alpha_1,
\phi)$ parametrize every instance of this 2-voter, 2-dimensional
problem up to rotation.
\Cref{fig:exp2-synth} shows the expected individual proportionaltiy level $\mathbb{E}_X[\textsf{IP}_i]$ for the majority (solid) and minority
(dashed) voter type under each of the five rules. We observe that for
$\phi \leq 90\degree$, every rule serves both voter types above their
entitlement and the choice of rule has little impact; an observation in line with
our approximation results from  \Cref{thm:arith_ang_approx}. As $\phi \to
180\degree$, the picture changes sharply: the arithmetic mean, Borda,
and the geometric median all move towards the majority voter type's
preferred ranking, and the minority's individual proportionality level
falls toward $0$. In contrast, as guaranteed, the angular mean and PSB provide a individual proportionality level of at least $1$ for both voter types 
across the entire range of $\phi$.

\begin{figure}
  \centering
  \resizebox{\textwidth}{!}{\input{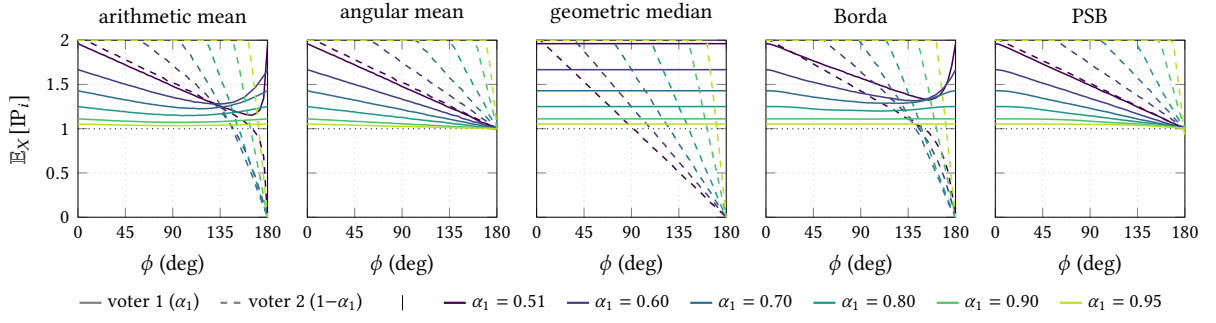}}
  \caption{Each voter type's expected individual proportionality level $\mathbb{E}_X[\textsf{IP}_i]$ as a function of the angle $\phi$ between $\theta^{(1)}$ and $\theta^{(2)}$, on the two-voter profile for different weight splits.}
  \label{fig:exp2-synth}
\end{figure}

Note that \Cref{fig:exp2-synth} reports voter type's individual proportionality
level $\textsf{IP}_i = a_{KT} / \alpha^{(i)} \binom{m}{2}$, which
counts every pairwise comparison the rule agrees with voter type $i$
on.  This includes pairs on which the two voter types already agree
with each other, a free ``win'' for any reasonable aggregation method that
inflates individual proportionality levels.  We now define a disagreement-restricted variant
$\widetilde{\textsf{IP}}_i$ that ignores these free wins and credits
the rule only for its agreement on \emph{contested} pairs.
For a batch $X = (x_1, \dots, x_m) \in \mathcal{X}^m$, a
pair of items $\{x_j, x_k\} \subseteq X$ is \emph{contested} iff the
two voter types rank it differently:
$
  (x_j \succ_{\theta^{(1)}}^X x_k) \;\Leftrightarrow\;
  (x_k \succ_{\theta^{(2)}}^X x_j).
$
Let $s_i(\theta_f, X)$
be the number of contested pairs on which $\theta_f$'s ranking sides
with voter type $i$.
Voter type $i$'s
disagreement-restricted individual proportionality level $\widetilde{\textsf{IP}}_i$ is the
ratio
\[
  \widetilde{\textsf{IP}}_i(\theta_f, X) \;=\;
  \frac{s_i(\theta_f, X)}{\alpha^{(i)} \cdot \bigl|\bigl\{\{x_j, x_k\} \subseteq X : \{x_j, x_k\}
  \text{ is contested}\bigr\}\bigr|}.
\]
The denominator is voter type $i$'s proportional share of the contested
pairs in this batch.  The numerator is the number of
contested pairs the rule actually awards to voter type $i$.
The expected value of $\widetilde{\textsf{IP}}_i$ is plotted in \Cref{fig:exp2-synth-disagreement}.

\begin{figure}[t]
  \centering
  \resizebox{\textwidth}{!}{\input{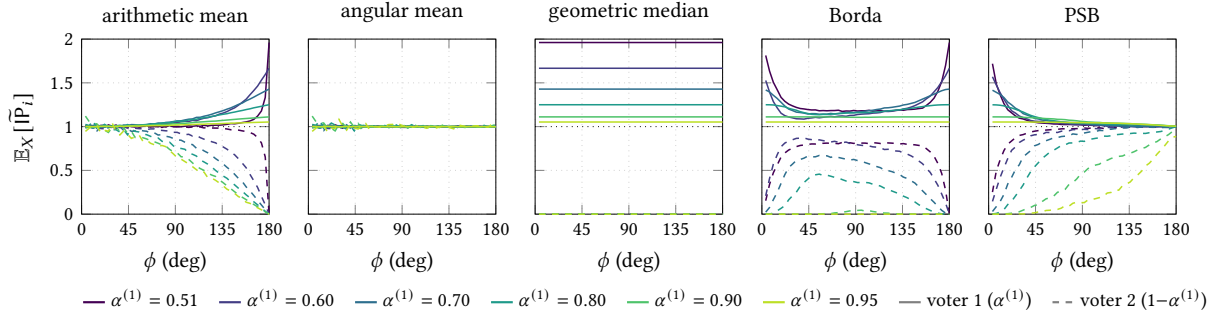}}
  \caption{Experiment on two voter types.  Setup as \Cref{fig:exp2-synth}, but with an adapted definition of the individual proportionality level: each voter type demands  an $\alpha$ fraction of the pairs on which the two voter types disagree.  }
  \label{fig:exp2-synth-disagreement}
\end{figure}

\subsection{Additional Information on Datasets}\label{app:dataset_info}

\begin{figure*}[h]
  \centering
  \begin{tikzpicture}
    \node[font=\footnotesize, inner sep=0pt] {%
      \tikz[baseline=-0.5ex]\draw[exp8color0, line width=1pt] (0,0) -- (0.35,0);~arithmetic mean\quad
      \tikz[baseline=-0.5ex]\draw[exp8color1, line width=1pt] (0,0) -- (0.35,0);~angular mean\quad
      \tikz[baseline=-0.5ex]\draw[exp8color2, line width=1pt] (0,0) -- (0.35,0);~geometric median\quad
      \tikz[baseline=-0.5ex]\draw[exp8color3, line width=1pt] (0,0) -- (0.35,0);~Borda\quad
      \tikz[baseline=-0.5ex]\draw[exp8color4, line width=1pt] (0,0) -- (0.35,0);~PSB
    };
  \end{tikzpicture}

  \vspace{4pt}

  \resizebox{0.8\textwidth}{!}{\input{Figures_final/exp8_multi_subset_propcol}}
  \caption{$\pcol$ on heterogeneity-filtered voter
    subsamples (std of pairwise angular distances $\geq 65^\circ$),
    as a function of subsample size $n$, on the 2D variants of the
    three main-body datasets.  Counterpart of \Cref{fig:subsample}.
    Shaded bands show the inter-quartile range across $100$
    subsamples per $n$.}
  \label{fig:subsample-pcol}
\end{figure*}
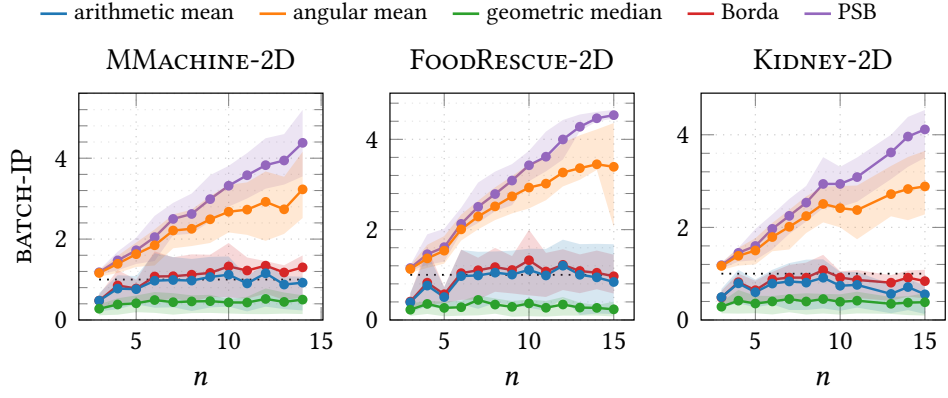

\subsubsection{Dataset Description}\label{app:dd} 

Across all five datasets, each voter's scoring vector is learned from binary pairwise comparisons using the same procedure. Given a voter~$i$ who answered~$T_i$ comparisons between alternatives~$a^{(t)}$ and~$b^{(t)}$, we fit a per-voter logistic regression
$
P\bigl(\text{voter } i \text{ chooses } a^{(t)}\bigr) = \sigma\!\bigl(\beta_i^\top (x_{a^{(t)}} - x_{b^{(t)}})\bigr),
$
where $x_a \in \mathbb{R}^d$ is the feature vector of alternative~$a$, and~$\sigma$ is the logistic function. This corresponds to a Thurstone--Mosteller random utility model with linear utilities $u_i(a) = \beta_i^\top x_a$. Feature differences are standardized per voter (zero mean, unit variance) before fitting, and we use~$\ell_2$ regularization with default strength~$C=1$ (scikit-learn's \texttt{LogisticRegression}). The learned coefficient vector~$\beta_i \in \mathbb{R}^d$ is the voter's scoring vector.

The five datasets differ in the alternatives
that are being compared, and the feature set used.  

\paragraph{\textsc{MMachine} (Moral Machine)~\citep{awad2018moral,noothigattu2018voting}.}
$n = 137$ voters, $d = 24$ features.  Users compared autonomous-vehicle
moral dilemmas in which a self-driving car must choose which of two
groups of pedestrians (or passengers) to prioritize when an accident is
unavoidable.  Each alternative is encoded as a vector of demographic
and contextual attributes of the corresponding group. Following \citet{DBLP:conf/aies/KimKAADTR18}, we fit one scoring vector per country by pooling all responses from that country; we only include countries for which at least $1000$ responses have been recorded.  
The $24$ features are: $4$ scenario-level attributes (\textsf{Intervention},
\textsf{Barrier}, \textsf{CrossingSignal}, \textsf{NumberOfCharacters})
and $20$ pedestrian-group composition attributes
(\textsf{Man}, \textsf{Woman}, \textsf{Pregnant}, \textsf{Stroller},
\textsf{OldMan}, \textsf{OldWoman}, \textsf{Boy}, \textsf{Girl},
\textsf{Homeless}, \textsf{LargeWoman}, \textsf{LargeMan}, \textsf{Criminal},
\textsf{MaleExecutive}, \textsf{FemaleExecutive},
\textsf{FemaleAthlete}, \textsf{MaleAthlete},
\textsf{FemaleDoctor}, \textsf{MaleDoctor}, \textsf{Dog}, \textsf{Cat}).
The 2D variant \textsc{MMachine-2D} uses the feature pair
\textsf{LargeWoman} and \textsf{MaleExecutive}.

\paragraph{\textsc{Kidney} (Kidney)~\citep{DBLP:conf/aaai/KeswaniCNCHBS26}.}
$n = 404$ voters, $d = 8$ features.  Voters compared hypothetical
kidney-allocation candidates described by clinical and behavioral
attributes, choosing which of two patients should receive an available
organ.
The $8$ features are: \textsf{dep} (number of dependents),
\textsf{life} (life expectancy with the transplant), \textsf{years\_waiting}
(years on the waiting list), \textsf{obesity} (BMI status), \textsf{alco}
(prior alcoholism), \textsf{work\_hours} (weekly work hours), \textsf{crim}
(prior criminal record), and \textsf{reject\_chance} (estimated
probability of organ rejection).  The 2D variant \textsc{Kidney-2D}
uses the feature pair \textsf{alco} and \textsf{crim}.

\paragraph{\textsc{FoodRescue} (412 Food Rescue)~\citep{lee2019webuildai}.}
$n = 19$ voters, $d = 7$ features.  Voters compared candidate recipient
organizations for food donations, deciding which of two organizations
should be prioritized for an incoming rescue.
The $7$ features are: \textsf{size} (organization size),
\textsf{access} (alternative-food access score),
\textsf{income} (median income of the served population),
\textsf{poverty} (poverty rate of the served population),
\textsf{last\_donation} (time since last received donation),
\textsf{total} (total donations received historically), and
\textsf{dist} (delivery distance).  The 2D variant
\textsc{FoodRescue-2D} uses the feature pair \textsf{last\_donation}
and \textsf{total}.

\paragraph{\textsc{KidneyStudy1} \cite{DBLP:conf/aies/BoerstlerKCBCHS24}.}
$n = 17$ voters, $d = 4$ features.  An earlier pilot study on
kidney-allocation preferences with a smaller voter population and a
narrower feature set than the main \textsc{Kidney} dataset.  The
$4$ features are: \textsf{alco}, \textsf{dep}, \textsf{life}, and
\textsf{crim} (definitions as in \textsc{Kidney}).  The 2D variant
\textsc{KidneyStudy1-2D} uses the feature pair \textsf{dep} and
\textsf{life}.

\begin{figure}[t]
  \centering
  \resizebox{0.95\textwidth}{!}{\input{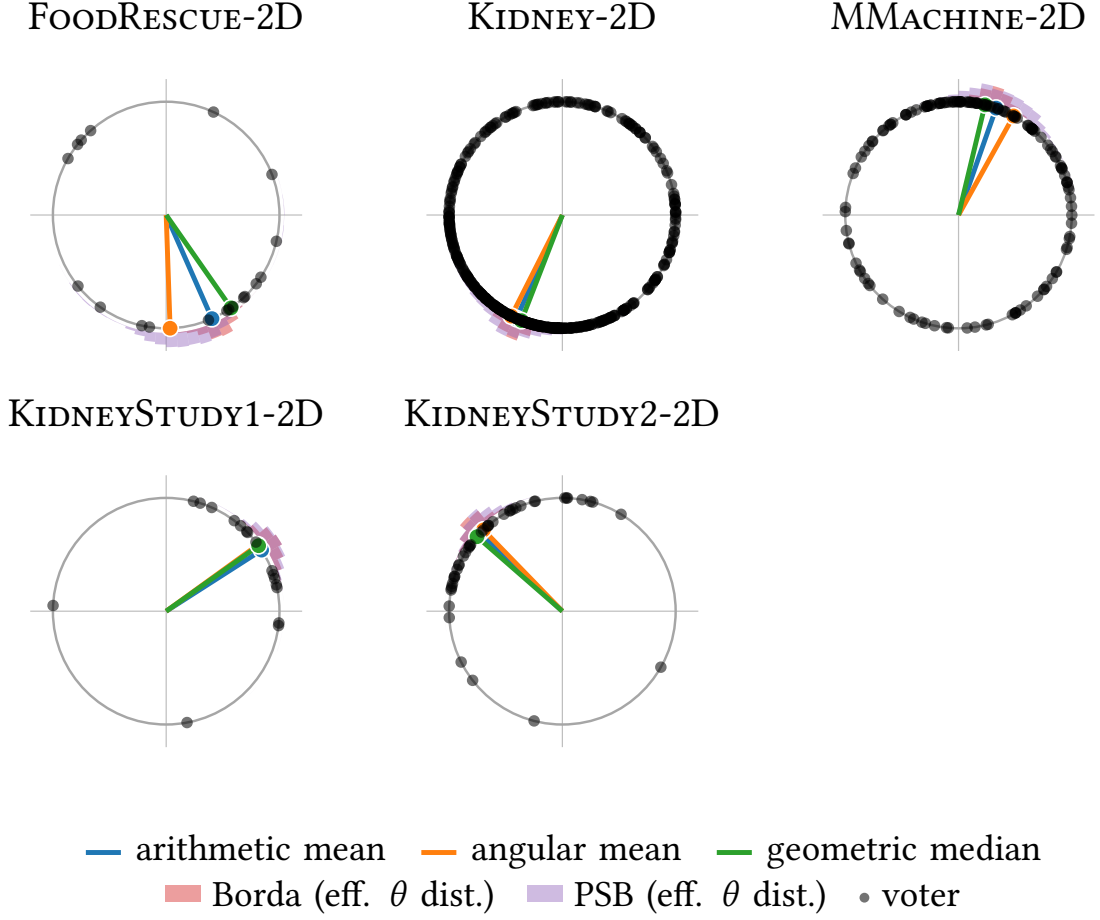}}
  \caption{Voter and mechanism positions on $S^1$ for the 2D variants
    of all five datasets.  Black dots: voters.  Radial lines:
    arithmetic mean (blue), angular mean (orange), geometric median
    (green).  Shades outside the unit circle: histograms of
    the effective direction $\theta_{\mathrm{eff}}(X)$
    recovered from the rankings produced by Borda (red) and PSB
    (purple) over $R = 2000$ batches per dataset.}
  \label{fig:exp9-subset-viz}
\end{figure}

\paragraph{\textsc{KidneyStudy2} \cite{DBLP:conf/aies/BoerstlerKCBCHS24}.}
$n = 43$ voters, $d = 5$ features.  A follow-up kidney-allocation
preference study with a partially overlapping but distinct feature
set.  The $5$ features
are: \textsf{elderlyDep} (number of elderly dependents),
\textsf{lifeYearsGained} (expected life-years gained from the
transplant), \textsf{obesity}, \textsf{weeklyWorkhours}, and
\textsf{yearsWaiting}.  The 2D variant \textsc{KidneyStudy2-2D}
uses the feature pair \textsf{obesity} and \textsf{weeklyWorkhours}.

\paragraph{Dataset Properties}

\begin{table}[t]
\centering
\caption{Heterogeneity of voter types' scoring vectors and pairwise angular distances between aggregation rules, on the two kidney studies. All values are in degrees.}
\label{tab:heterogeneity-kidneystudies}
\resizebox{0.7\linewidth}{!}{\begin{tabular}{l|cc|cc}
 & \multicolumn{2}{c|}{\textsc{KidneyStudy1}} & \multicolumn{2}{c}{\textsc{KidneyStudy2}} \\
 & Orig. & 2D & Orig. & 2D \\[0.25em]
\hline\hline
\multicolumn{5}{l}{\textbf{Heterogeneity of $\theta^{(i)}$}} \\
max pairwise angle between $\theta^{(i)}$ & 85.26 & 176.91 & 88.00 & 179.86 \\
avg of pairwise angular distances between $\theta^{(i)}$ & 48.72 & 53.29 & 39.87 & 49.58 \\
stdev of pairwise angular distances between $\theta^{(i)}$ & 20.07 & 46.13 & 17.73 & 40.51 \\
\hline
\multicolumn{5}{l}{\textbf{Rule coincidence}} \\
$d_\measuredangle(\theta_{\mathrm{ang}}, \theta_{\mathrm{arith}})$ & 0.33 & 3.23 & 0.66 & 2.30 \\
$d_\measuredangle(\theta_{\mathrm{ang}}, \theta_{\mathrm{med}})$ & 3.38 & 0.55 & 6.41 & 4.73 \\
$d_\measuredangle(\theta_{\mathrm{arith}}, \theta_{\mathrm{med}})$ & 3.06 & 2.68 & 5.77 & 2.43 \\
\end{tabular}}
\end{table}

\Cref{fig:exp9-subset-viz} visualizes the voter and mechanism
positions on $S^1$ for the 2D variants of all five datasets.
Visualizing Borda and PSB on $S^1$ requires care.  Both rules produce
an arbitrary ranking per batch, not a single direction, and the ranking they
produce on a given batch need not be induced by any scoring
vector.  To place them on the same circle as the fixed linear rules, we
recover for each batch $X$ an \emph{effective direction}
$\theta_{\mathrm{eff}}(X) \in S^1$, defined as
\[
  \theta_{\mathrm{eff}}(X)
  \;\in\;
  \argmin_{\theta \in S^1}\;
  d_{KT}\!\bigl(\succ^{X}_\theta,\; \succ^X_f\bigr),
\]
where $\succ^X_f$ is the ranking produced by the rule
$f \in \{\text{Borda}, \text{PSB}\}$ on $X$ and
$d_{KT}$ is Kendall-tau distance. 

\subsubsection{Results on Kidney Study 1 and 2}\label{app:kid}

We replicate the main-body results
(\Cref{tab:overview-minimal-body}, \Cref{fig:subsample}, \Cref{fig:convergence})
on the 2D variants of \textsc{KidneyStudy1} and \textsc{KidneyStudy2}.
\Cref{tab:overview-minimal-kid} is the analogue of
\Cref{tab:overview-minimal-body}, \Cref{fig:convergence-appendix}
the analogue of \Cref{fig:convergence}, and \Cref{fig:subsample-appendix} the
analogue of \Cref{fig:subsample}.

\input{Figures_final/overview_table_minimal_appendix}

\definecolor{exp8color0}{rgb}{0.1216,0.4667,0.7059}
\definecolor{exp8color1}{rgb}{1.0000,0.4980,0.0549}
\definecolor{exp8color2}{rgb}{0.1725,0.6275,0.1725}
\definecolor{exp8color3}{rgb}{0.8392,0.1529,0.1569}
\definecolor{exp8color4}{rgb}{0.5804,0.4039,0.7412}

\begin{figure*}[h]
  \centering
  \begin{tikzpicture}
    \node[font=\footnotesize, inner sep=0pt] {%
      \tikz[baseline=-0.5ex]\draw[exp8color0, line width=1pt] (0,0) -- (0.35,0);~arithmetic mean\quad
      \tikz[baseline=-0.5ex]\draw[exp8color1, line width=1pt] (0,0) -- (0.35,0);~angular mean\quad
      \tikz[baseline=-0.5ex]\draw[exp8color2, line width=1pt] (0,0) -- (0.35,0);~geometric median\quad
      \tikz[baseline=-0.5ex]\draw[exp8color3, line width=1pt] (0,0) -- (0.35,0);~Borda\quad
      \tikz[baseline=-0.5ex]\draw[exp8color4, line width=1pt] (0,0) -- (0.35,0);~PSB
    };
  \end{tikzpicture}

  \vspace{4pt}

  \begin{minipage}[f]{0.49\textwidth}
    \centering
    \resizebox{\linewidth}{!}{\input{Figures_final/exp10_subset_appendix}}
    \caption{$\pind$ (solid) and $\pcol$ (dashed) on the full electorates, as a function of batch size $m$. Counterpart of \Cref{fig:convergence} for the 2D variants
      of the two kidney study datasets.}
    \label{fig:convergence-appendix}
  \end{minipage}\hfill
  \begin{minipage}[t]{0.49\textwidth}
    \centering
    \resizebox{\linewidth}{!}{\input{Figures_final/exp8_multi_subset_appendix}}
   \caption{Counterpart of \Cref{fig:subsample} for the two kidney
  study datasets: $\pind$ on heterogeneity-filtered voter subsamples
  (std of pairwise angular distances $\geq 65^\circ$), as a function
  of subsample size $n$, with shaded bands showing the inter-quartile
  range across 100 subsamples per $n$.  The maximum $n$ is smaller
  than in \Cref{fig:subsample} because the kidney-study electorates
  are less diverse: at large $n$, no subsample of voters spans a wide
  enough arc on $S^1$ to satisfy the filter.}
    \label{fig:subsample-appendix}
  \end{minipage}
\end{figure*}

\begin{figure*}[h]
  \centering
  \begin{tikzpicture}
    \node[font=\footnotesize, inner sep=0pt] {%
      \tikz[baseline=-0.5ex]\draw[exp8color0, line width=1pt] (0,0) -- (0.35,0);~arithmetic mean\quad
      \tikz[baseline=-0.5ex]\draw[exp8color1, line width=1pt] (0,0) -- (0.35,0);~angular mean\quad
      \tikz[baseline=-0.5ex]\draw[exp8color2, line width=1pt] (0,0) -- (0.35,0);~geometric median\quad
      \tikz[baseline=-0.5ex]\draw[exp8color3, line width=1pt] (0,0) -- (0.35,0);~Borda\quad
      \tikz[baseline=-0.5ex]\draw[exp8color4, line width=1pt] (0,0) -- (0.35,0);~PSB
    };
  \end{tikzpicture}

  \vspace{4pt}

  \resizebox{\textwidth}{!}{\input{Figures_final/exp10_full_all}}
  \caption{$\pind$ (solid) and $\pcol$ (dashed) on the full feature
    sets of all five datasets, as a function of batch size $m$.  As
    in \Cref{fig:convergence}, the per-batch gap shrinks rapidly as
    $m$ grows.  The five rules are nearly indistinguishable: the
    full-feature electorates are not heterogeneous enough to separate
    them, consistent with the homogeneity statistics in
    \Cref{tab:heterogeneity-main-t}.}
  \label{fig:convergence-full}
\end{figure*}
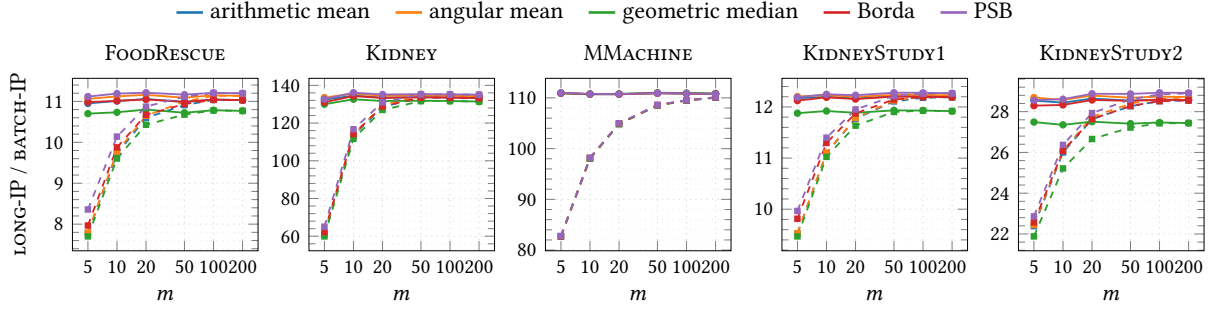

\subsubsection{Results on Full Datasets} \label{app:original-datasets}

\Cref{fig:convergence-full} confirms this in the same convergence-in-$m$ plot
used in \Cref{fig:convergence}.  We omit a figure parallel to \Cref{fig:subsample} concerning heterogeneity-filtered
subsamples: at the original feature dimensions, voters'
scoring vectors are too clustered for the $65^\circ$ pairwise-distance
filter  to admit any subsample, so the
filtered analysis cannot be run on the original feature sets.

\subsection{Ablations}\label{app:ablations}
Our empirical results presented in the main body use uniform voter
weights $\alpha^{(i)} = 1/n$ and a uniform distribution on $\mathbb{S}^{d-1}$.  We now proceed to relax each in turn.

\subsubsection{Non-Uniform Weights}\label{app:weights}

In \Cref{sec:empirical} we used uniform voter weights.  Here we draw
weights from a $\mathrm{Dirichlet}(1)$ distribution---that is,
uniformly at random from the simplex of weight vectors---once per
dataset.  
\Cref{tab:overview-minimal-randw} reports $\pind$ and $\pcol$ for
each rule on each of the five datasets and their 2D variants under
random weights, in the same format as \Cref{tab:overview-minimal-body}
in the main body.

\input{Figures_final/overview_table_minimal_randw}

\Cref{fig:convergence-full-randw} repeats the convergence-in-$m$
analysis of \Cref{fig:convergence} on the full feature sets of
all five datasets under random weights.

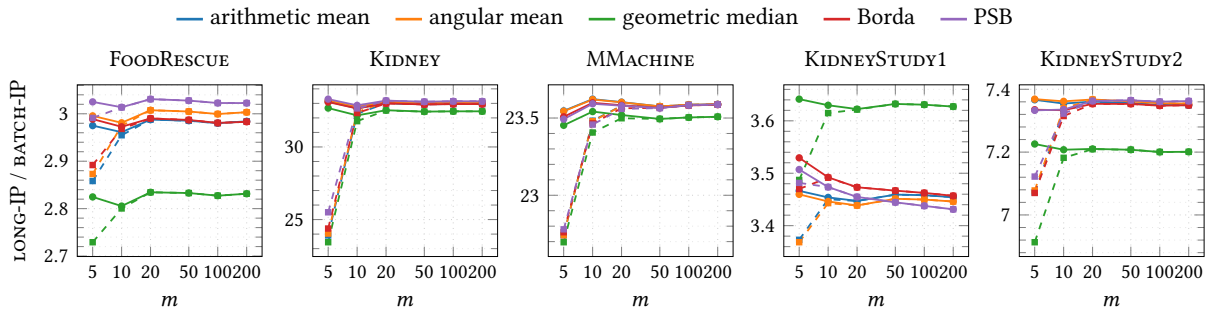
\begin{figure*}[h]
  \centering
  \begin{tikzpicture}
    \node[font=\footnotesize, inner sep=0pt] {%
      \tikz[baseline=-0.5ex]\draw[exp8color0, line width=1pt] (0,0) -- (0.35,0);~arithmetic mean\quad
      \tikz[baseline=-0.5ex]\draw[exp8color1, line width=1pt] (0,0) -- (0.35,0);~angular mean\quad
      \tikz[baseline=-0.5ex]\draw[exp8color2, line width=1pt] (0,0) -- (0.35,0);~geometric median\quad
      \tikz[baseline=-0.5ex]\draw[exp8color3, line width=1pt] (0,0) -- (0.35,0);~Borda\quad
      \tikz[baseline=-0.5ex]\draw[exp8color4, line width=1pt] (0,0) -- (0.35,0);~PSB
    };
  \end{tikzpicture}

  \vspace{4pt}

  \resizebox{\textwidth}{!}{\input{Figures_final/exp10_full_all_randw}}
  \caption{Random-weights variant of \Cref{fig:convergence}.
    $\pind$ (solid) and $\pcol$ (dashed) on the full feature sets of
    all five datasets, as a function of batch size $m$, with voter
    weights drawn once per dataset from a $\mathrm{Dirichlet}(1)$
    distribution.  The qualitative picture is unchanged from the
    uniform-weights case: the per-batch gap closes rapidly as $m$
    grows and the five rules remain nearly indistinguishable on the
    full-feature electorates.}
  \label{fig:convergence-full-randw}
\end{figure*}

The qualitative conclusions of \Cref{sec:empirical} carry over.  In
\Cref{tab:overview-minimal-randw}, the relative ordering of rules
within each dataset and the relative ordering of datasets are
unchanged from \Cref{tab:overview-minimal-body}: on the full feature
sets, all rules satisfy IP comfortably and are nearly
indistinguishable; on the 2D variants, the arithmetic mean and the
geometric median fall below the entitlement on \textsc{Kidney-2D}
and \textsc{MMachine-2D}, while the angular mean and PSB satisfy it.
Absolute IP values do shift -- because the entitlement denominator
$\alpha^{(i)} \binom{m}{2}$ changes with the weight draw.  In \Cref{fig:convergence-full-randw}, the
per-batch gap $\pind - \pcol$ closes at the same rate as in
\Cref{fig:convergence-full}.

\subsubsection{Non-Uniform Item Distribution} \label{app:acg}
The proportionality guarantee for the angular mean
(\Cref{thm:thang_prop}) holds under spherically symmetric batch
distributions $\mathcal{D}_\circ$.  In this section we probe what
happens when this assumption is dropped, by drawing items from an
\emph{angular central Gaussian} (ACG) distribution
\citep{tyler1987statistical} on $S^{d-1}$ parameterized by a
concentration axis $v \in S^{d-1}$ and an anisotropy parameter
$\lambda \in (0, 1]$:
\[
  z \sim \mathcal{N}\bigl(0,\; \lambda \cdot I + (1 - \lambda)\, v v^\top \bigr),
  \qquad x = z / \|z\|.
\]
The covariance has eigenvalue $1$ along $v$ and $\lambda$ along every
direction orthogonal to $v$; $\lambda = 1$ recovers the isotropic
Gaussian (and hence the uniform distribution on $S^{d-1}$), while
$\lambda \to 0$ collapses items onto $\{+v, -v\}$.  We probe two
settings: a synthetic two-voter setup and the real-world 2D
electorates.

\paragraph{Synthetic two-voter ACG sweep.}
\Cref{fig:exp7-synth} extends \Cref{fig:exp2-synth} to
non-spherical sampling.  We fix $\alpha^{(1)} = 0.7$ and sweep five
angles $\phi \in \{45^\circ, 90^\circ, 135^\circ, 150^\circ,
175^\circ\}$ between the two voters' scoring vectors.  Each panel
column corresponds to one of five concentration axes $v$ placed at
geodesic positions $t \in \{0, 0.25, 0.5, 0.75, 1\}$ between
$\theta^{(1)}$ and $\theta^{(2)}$, and each row to one of the five
aggregation rules.  Within each panel, $\textsf{IP}_i$ for voter type
$1$ (solid) and voter type $2$ (dashed) is plotted as a function of
$\lambda$.

\begin{figure*}[h]
  \centering
  \resizebox{\textwidth}{!}{\input{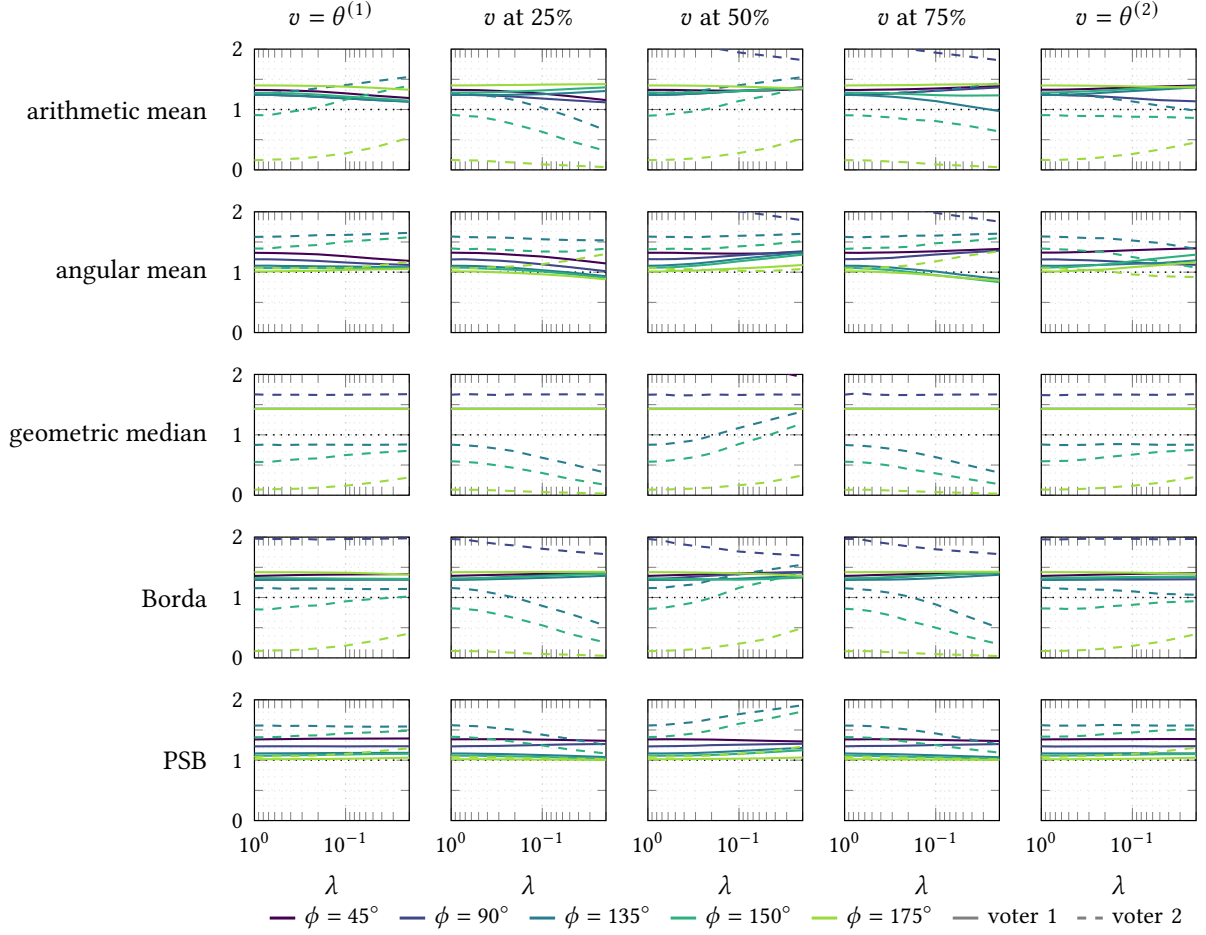}}
  \caption{Per-voter $\textsf{IP}_i$ on the two-voter profile with
    items drawn from an angular central Gaussian (ACG) batch
    distribution.  Anisotropy $\lambda \in [0.02, 1]$ on the $x$-axis
    (log scale, decreasing left-to-right); concentration axis $v$
    placed at geodesic positions $t \in \{0, 0.25, 0.5, 0.75, 1\}$
    between $\theta^{(1)}$ and $\theta^{(2)}$ (columns); rows are
    mechanisms.  Solid lines: voter~1 ($\alpha^{(1)} = 0.7$); dashed:
    voter~2.  Colors: $\phi \in \{45^\circ, 90^\circ, 135^\circ,
    150^\circ, 175^\circ\}$.  $\lambda = 1$ recovers spherical
    symmetry; $\lambda \to 0$ collapses items onto $\{+v, -v\}$.}
  \label{fig:exp7-synth}
\end{figure*}

\paragraph{Real-world ACG sweep.}
We replicate the same sweep on the 2D variants of all five real-world
datasets.  Because each electorate now contains many voters in
specific positions on $S^1$ (rather than a stylized two-voter
placement), the natural reference axes for $v$ are no longer the
voters themselves.  We use four fixed compass directions
$v \in \{0^\circ, 45^\circ, 90^\circ, 135^\circ\}$ on $S^1$, which
cover the half-circle without redundancy under the $v \sim -v$
symmetry of the ACG. We show the results in \Cref{fig:exp-acg-real}.

\begin{figure*}[h]
  \centering
  \resizebox{\textwidth}{!}{\input{Figures_final/exp_acg_real}}
  \caption{Real-world counterpart of \Cref{fig:exp7-synth}.  Long-run
    individual proportionality $\pind$ on the 2D variants of all five
    real-world datasets, as a function of the ACG anisotropy
    parameter $\lambda$ (log scale, decreasing left-to-right) and
    four concentration axes $v \in \{0^\circ, 45^\circ, 90^\circ,
    135^\circ\}$ on $S^1$.  Mechanism = color; axis = linestyle.
    Voter weights are uniform $\alpha^{(i)} = 1/n$.  At $\lambda = 1$
    axis direction has no effect and values match the corresponding
    rows of \Cref{tab:overview-minimal-body}.  The horizontal dotted line
    marks the proportional entitlement $\pind = 1$.  Note the log
    $y$-axis: PSB lives one to two orders of magnitude above the
    fixed linear rules on \textsc{Kidney-2D} and
    \textsc{MMachine-2D}.}
  \label{fig:exp-acg-real}
\end{figure*}
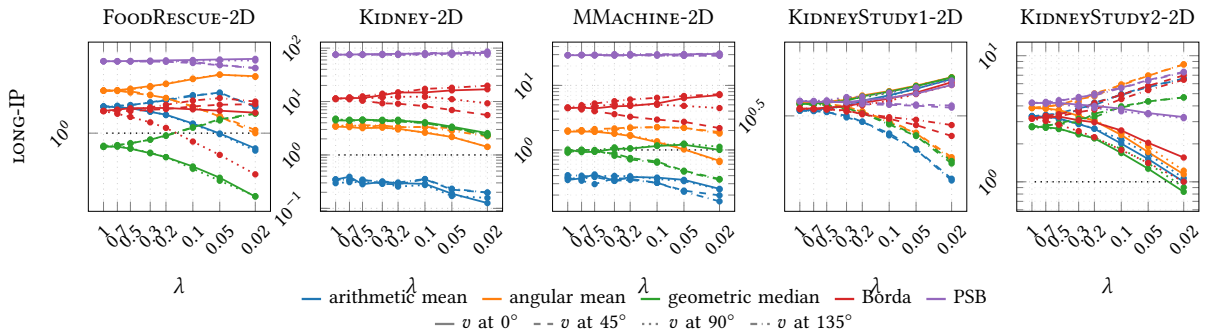

 We take
this as additional empirical evidence that the angular mean's
proportionality is robust to substantial departures from spherical
symmetry on the real-world electorates studied in
\Cref{sec:empirical}.

\end{document}

%% file: Figures/angle_agreement.tex
        \begin{tikzpicture}[scale=1.2]
        
        % Define coordinates
        \coordinate (O) at (0,0);
        \coordinate (theta1) at (60:1); % \theta^{(1)}
        \coordinate (theta2) at (20:1); % \theta^{(2)}

        % Calculate angles for boundaries
        \pgfmathanglebetweenpoints{\pgfpointanchor{O}{center}}{\pgfpointanchor{theta1}{center}}
        \let\angleone\pgfmathresult
        \pgfmathanglebetweenpoints{\pgfpointanchor{O}{center}}{\pgfpointanchor{theta2}{center}}
        \let\angletwo\pgfmathresult

        % Fill agreement regions (double wedge)
        \fill[violet!40, opacity=0.5] (O) 
          -- ({90+\angleone}:1) arc ({90+\angleone}:{270+\angletwo}:1) -- cycle;
        \fill[orange!40, opacity=0.5] (O) 
          -- ({270+\angleone}:1) arc ({270+\angleone}:{450+\angletwo}:1) -- cycle;

        % Draw the unit circle
        \draw[black] (O) circle (1);
        \fill (0,0) circle (0.02);

        % Draw the axes (lightly dashed)
        %\draw[->, dashed, black!40] (-1.2, 0) -- (1.2, 0) node[right] {$x$};
        %\draw[->, dashed, black!40] (0, -1.2) -- (0, 1.2) node[above] {$y$};

        % Draw orthogonal boundaries (hyperplanes)
        \draw[red!60!black, dashed, thick] ({90+\angleone}:1.15) -- ({270+\angleone}:1.15);
        % \draw[blue!80!black, dashed, thick] ({90+\angletwo}:1.15) -- ({270+\angletwo}:1.15);
        \draw[teal!80!black, dashed, thick] ({90+\angletwo}:1.15) -- ({270+\angletwo}:1.15);

        % Angle arc between the vectors
        \pic [draw=black, fill=gray!30, angle eccentricity=1.3, angle radius=0.4cm, "\color{black}${\scriptstyle\dang(\theta,\psi)}$" {shift={(0.3cm, -0.35cm)}}, opacity=1] {angle = theta2--O--theta1};
        
        % Draw the vectors
        \draw[->, thick, red!60!black] (O) -- (theta1) node[above=.2, right] {$\theta$};
        % \draw[->, thick, blue!80!black]  (O) -- (theta2) node[right] {$\psi$};
        \draw[->, thick, teal!80!black]  (O) -- (theta2) node[right] {$\psi$};

        % Labels for the regions
        %\node at (-0.6, -0.6) {$+$ agreement};
        %\node at (0.6, 0.6) {$-$ agreement};

    \end{tikzpicture}

%% file: Figures/proportionality_proof.tex
\begin{tikzpicture}[scale=1.0, >=stealth]

        % ==========================================

        % LEFT DIAGRAM: The Global Sphere

        % ==========================================

        \begin{scope}[xshift=-5.5cm, yshift=0cm]

            % Sphere outline (Radius 1.75)

            \draw[thick] (0,0) circle (1.75);

            % ---------------------------------------------------------

            % POINTS ON THE BACK OF THE SPHERE

            % ---------------------------------------------------------

            \filldraw[gray!50] (0.84, 0.77) circle (1.5pt);

            \filldraw[gray!50] (-1.33, -0.07) circle (1.5pt);

            % Equator (Great Circle) 

            \draw[thick] (-1.75,0) arc (180:360:1.75 and 0.525);

            \draw[dashed, gray] (1.75,0) arc (0:180:1.75 and 0.525);

            % G_i label sitting safely under the arc

            \node at (-1.0, -0.65) {$G_i$};

            % ---------------------------------------------------------

            % PRIMARY COORDINATES

            % ---------------------------------------------------------

            \coordinate (Tstar) at (-1.75, 0);            

            \coordinate (MinusTstar) at (1.75, 0);        

            \coordinate (Ti) at (0, -0.525);              

            \coordinate (Tprime) at (0.7, -0.48);        

            \coordinate (Tj) at (-0.91, 0.84);

            % ==========================================

            % 3D TANGENT PLANE AT THETA_ANG

            % ==========================================

            \filldraw[fill=pink!30, draw=pink!80!black, thick, opacity=0.6, join=round] 

                (-2.15, 1.2) -- (-1.35, 0.8) -- (-1.35, -1.2) -- (-2.15, -0.8) -- cycle;

            \draw[pink!80!black, dashed, opacity=0.5] (-2.15, 0.2) -- (-1.35, -0.2);

            \draw[pink!80!black, dashed, opacity=0.5] (-1.75, 1.0) -- (-1.75, -1.0);

            % ---------------------------------------------------------

            % GEODESIC ARCS & HIGHLIGHTS

            % ---------------------------------------------------------

            \draw[line width=8pt, yellow!50, cap=round] (Tj) to[bend left=8] (Tprime);

            \draw[thick] (Tj) to[bend left=8] (Tprime);

            \draw[thick] (Tj) to[bend left=10] (MinusTstar);

            % ---------------------------------------------------------

            % GAMMA_J ANGLE

            % ---------------------------------------------------------

            \draw[thick] ([shift=(168:0.35)]MinusTstar) arc (168:225:0.35);

            \node at ([shift=(196:0.55)]MinusTstar) {$\gamma_j$};

            % ---------------------------------------------------------

            % POINTS ON THE FRONT OF THE SPHERE

            % ---------------------------------------------------------

            \filldraw[gray] (0.21, 0.56) circle (1.5pt);

            \filldraw[gray] (-0.14, 1.26) circle (1.5pt);

            \filldraw[gray] (0.56, -1.05) circle (1.5pt);

            % Label shifted significantly left (left=12pt) to clear the tangent box

            \filldraw[magenta] (Tstar) circle (2pt) node[left=12pt] {$\theta_{ang}$};

            \filldraw[magenta] (MinusTstar) circle (2pt) node[right=2pt] {$-\theta_{ang}$};

            \filldraw[green!60!black] (Ti) circle (2pt) node[below=2pt] {$\theta^{(i)}$};

            \filldraw[cyan] (Tprime) circle (2pt) node[below right=-0.1pt] {$\theta'$};

            \filldraw[violet] (Tj) circle (2pt) node[above=1pt] {$\theta^{(j)}$}; 

        \end{scope}

        % ==========================================

        % MIDDLE DIAGRAM: Euclidean Triangle

        % ==========================================

        % Centered horizontally at x=0 via xshift=-2.0, Centered vertically via yshift=-0.9

        \begin{scope}[xshift=-2.0cm, yshift=-0.9cm]

            \coordinate (E_Tj) at (0, 1.8);         % Shifted fully left

            \coordinate (E_Tprime) at (1.2, 0);     % Shifted right relative to Tj, creating the ~120 deg angle

            \coordinate (E_Tstar) at (4.0, 0);

            % Yellow Highlight 

            \draw[line width=8pt, yellow!50, cap=round] (E_Tj) to (E_Tprime);

            % Triangle Edges

            \draw[thick] (E_Tstar) -- node[below, font=\small] {$\pi - d_{\measuredangle}(\theta_{ang}, \theta')$}(E_Tprime);

            \draw[thick] (E_Tstar) -- node[sloped, above, font=\small] {$\pi - d_{\measuredangle}(\theta_{ang}, \theta^{(j)})$} (E_Tj);

            \draw[thick] (E_Tj) -- node[sloped, below, font=\small] {$> d_{\measuredangle}(\theta', \theta^{(j)})$} (E_Tprime);

            % Acute Interior Angle Gamma_j (Calculates cleanly inside the new geometry)

            \pic [draw, thick, "$\gamma_j$", angle radius=0.55cm, angle eccentricity=1.5] {angle = E_Tj--E_Tstar--E_Tprime};

            % Points and Labels

            \filldraw[magenta] (E_Tstar) circle (2pt) node[below right =2pt] {$-\theta_{ang}$};

            \filldraw[cyan] (E_Tprime) circle (2pt) node[below=2pt] {$\theta'$};

            \filldraw[violet] (E_Tj) circle (2pt) node[above=2pt] {$\theta^{(j)}$}; 

        \end{scope}

        % ==========================================

        % RIGHT DIAGRAM: 1st Order Conditions

        % ==========================================

        % Centered horizontally at x=+5.5 via xshift=5.25, Centered vertically via yshift=-0.4

        \begin{scope}[xshift=5.5cm, yshift=-0.4cm]
            % Modified: replaced \draw with \filldraw and added matching styles
            \filldraw[fill=pink!30, draw=pink!80!black, thick, opacity=0.6] (-1.3, -0.7) rectangle (1.8, 1.5);
            
            \coordinate (O) at (0, 0);
            
            \draw[->, very thick, green!60!black] (O) -- (1.5, 0.05);
            \draw[->, very thick, violet] (O) -- (0.8, 1.2);
            
            % Gray Balancing Vectors 
            \draw[->, thick, gray] (O) -- (-0.8, 0.7);   
            \draw[->, thick, gray] (O) -- (-0.5, 0.65);  
            \draw[->, thick, gray] (O) -- (-0.8, 0.4);  
            \draw[->, thick, gray] (O) -- (-0.25, 1.0);  
            \draw[->, thick, gray] (O) -- (-0.75, -0.4); 
            \draw[->, thick, gray] (O) -- (-1.0, -0.2); 
            \draw[->, thick, gray] (O) -- (-0.4, -0.5); 
            
            % Label moved to the right so it completely clears the down/left gray arrows
            \filldraw[magenta] (O) circle (2pt) node[below right=2pt] {$\theta_{ang}$};
        \end{scope}

\end{tikzpicture}

%% file: Figures_final/double-right.tex
% Auto-generated by exp10_propind_propcol_vs_m.py --pgfplots.
% Required preamble in the host document:
%   \usepackage{tikz}
%   \usepackage{pgfplots}
%   \usepgfplotslibrary{groupplots}
%   \pgfplotsset{compat=1.18}
%
\definecolor{exp10color0}{rgb}{0.1216,0.4667,0.7059}
\definecolor{exp10color1}{rgb}{1.0000,0.4980,0.0549}
\definecolor{exp10color2}{rgb}{0.1725,0.6275,0.1725}
\definecolor{exp10color3}{rgb}{0.8392,0.1529,0.1569}
\definecolor{exp10color4}{rgb}{0.5804,0.4039,0.7412}

\begin{tikzpicture}
  \begin{groupplot}[
    group style={
      group size=3 by 1,
      horizontal sep=22pt,
    },
    width=4.4cm,
    height=4.2cm,
    xmode=log,
    log basis x=10,
    xtick={5,10,20,50,100,200},
    xticklabels={5,10,20,50,100,200},
    xlabel={$m$},
    title style={font=\small,yshift=-2pt},
    label style={font=\small},
    tick label style={font=\footnotesize},
    minor y tick num=4,
    grid=both,
    minor grid style={dotted, gray!20},
    major grid style={dotted, gray!40},
  ]
    \nextgroupplot[title={\textsc{MMachine}-2D},
    ylabel={$\pind$ / $\pcol$},
    ]
      \addplot[color=exp10color0, line width=0.8pt, mark=*, mark size=1.2pt, mark options={solid}]
        coordinates {(5,0.3699) (10,0.3851) (20,0.3594) (50,0.3714) (100,0.3696) (200,0.3680)};
      \addplot[color=exp10color0, line width=0.8pt, dashed, mark=square*, mark size=1.0pt, mark options={solid}]
        coordinates {(5,0.0617) (10,0.2299) (20,0.3425) (50,0.3714) (100,0.3696) (200,0.3680)};
      \addplot[color=exp10color1, line width=0.8pt, mark=*, mark size=1.2pt, mark options={solid}]
        coordinates {(5,1.9386) (10,1.9332) (20,1.9494) (50,1.9527) (100,1.9578) (200,1.9516)};
      \addplot[color=exp10color1, line width=0.8pt, dashed, mark=square*, mark size=1.0pt, mark options={solid}]
        coordinates {(5,0.2466) (10,0.9620) (20,1.5672) (50,1.9153) (100,1.9569) (200,1.9516)};
      \addplot[color=exp10color2, line width=0.8pt, mark=*, mark size=1.2pt, mark options={solid}]
        coordinates {(5,0.9111) (10,0.9072) (20,0.9341) (50,0.9527) (100,0.9530) (200,0.9490)};
      \addplot[color=exp10color2, line width=0.8pt, dashed, mark=square*, mark size=1.0pt, mark options={solid}]
        coordinates {(5,0.1302) (10,0.4947) (20,0.8270) (50,0.9525) (100,0.9530) (200,0.9490)};
      \addplot[color=exp10color3, line width=0.8pt, mark=*, mark size=1.2pt, mark options={solid}]
        coordinates {(5,5.2745) (10,4.5986) (20,3.7354) (50,2.8575) (100,2.4255) (200,2.1808)};
      \addplot[color=exp10color3, line width=0.8pt, dashed, mark=square*, mark size=1.0pt, mark options={solid}]
        coordinates {(5,1.4317) (10,2.6456) (20,3.0028) (50,2.6635) (100,2.3123) (200,2.0958)};
      \addplot[color=exp10color4, line width=0.8pt, mark=*, mark size=1.2pt, mark options={solid}]
        coordinates {(5,31.0510) (10,29.4794) (20,29.8801) (50,30.1690) (100,30.2163) (200,30.1529)};
      \addplot[color=exp10color4, line width=0.8pt, dashed, mark=square*, mark size=1.0pt, mark options={solid}]
        coordinates {(5,21.2898) (10,25.7530) (20,28.3478) (50,29.5156) (100,29.8127) (200,29.9377)};
    \nextgroupplot[title={\textsc{FoodRescue}-2D}]
      \addplot[color=exp10color0, line width=0.8pt, mark=*, mark size=1.2pt, mark options={solid}]
        coordinates {(5,1.9171) (10,1.9234) (20,1.8941) (50,1.9009) (100,1.9012) (200,1.9007)};
      \addplot[color=exp10color0, line width=0.8pt, dashed, mark=square*, mark size=1.0pt, mark options={solid}]
        coordinates {(5,1.6007) (10,1.9021) (20,1.8940) (50,1.9009) (100,1.9012) (200,1.9007)};
      \addplot[color=exp10color1, line width=0.8pt, mark=*, mark size=1.2pt, mark options={solid}]
        coordinates {(5,2.7626) (10,2.7951) (20,2.8175) (50,2.8166) (100,2.8161) (200,2.8119)};
      \addplot[color=exp10color1, line width=0.8pt, dashed, mark=square*, mark size=1.0pt, mark options={solid}]
        coordinates {(5,1.9798) (10,2.6323) (20,2.7925) (50,2.8164) (100,2.8161) (200,2.8119)};
      \addplot[color=exp10color2, line width=0.8pt, mark=*, mark size=1.2pt, mark options={solid}]
        coordinates {(5,0.7258) (10,0.7433) (20,0.7362) (50,0.7301) (100,0.7321) (200,0.7322)};
      \addplot[color=exp10color2, line width=0.8pt, dashed, mark=square*, mark size=1.0pt, mark options={solid}]
        coordinates {(5,0.6697) (10,0.7425) (20,0.7362) (50,0.7301) (100,0.7321) (200,0.7322)};
      \addplot[color=exp10color3, line width=0.8pt, mark=*, mark size=1.2pt, mark options={solid}]
        coordinates {(5,1.5950) (10,1.7187) (20,1.7718) (50,1.7934) (100,1.8038) (200,1.8044)};
      \addplot[color=exp10color3, line width=0.8pt, dashed, mark=square*, mark size=1.0pt, mark options={solid}]
        coordinates {(5,1.4887) (10,1.7155) (20,1.7718) (50,1.7934) (100,1.8038) (200,1.8044)};
      \addplot[color=exp10color4, line width=0.8pt, mark=*, mark size=1.2pt, mark options={solid}]
        coordinates {(5,5.5983) (10,5.7084) (20,5.7547) (50,5.7864) (100,5.7678) (200,5.6956)};
      \addplot[color=exp10color4, line width=0.8pt, dashed, mark=square*, mark size=1.0pt, mark options={solid}]
        coordinates {(5,4.7481) (10,5.4275) (20,5.6279) (50,5.7073) (100,5.6895) (200,5.6610)};
    \nextgroupplot[title={\textsc{Kidney}-2D}]
      \addplot[color=exp10color0, line width=0.8pt, mark=*, mark size=1.2pt, mark options={solid}]
        coordinates {(5,0.3030) (10,0.3681) (20,0.3338) (50,0.3366) (100,0.3330) (200,0.3282)};
      \addplot[color=exp10color0, line width=0.8pt, dashed, mark=square*, mark size=1.0pt, mark options={solid}]
        coordinates {(5,0.0000) (10,0.1167) (20,0.2381) (50,0.3357) (100,0.3330) (200,0.3282)};
      \addplot[color=exp10color1, line width=0.8pt, mark=*, mark size=1.2pt, mark options={solid}]
        coordinates {(5,2.9290) (10,3.2140) (20,3.3447) (50,3.3926) (100,3.3883) (200,3.3811)};
      \addplot[color=exp10color1, line width=0.8pt, dashed, mark=square*, mark size=1.0pt, mark options={solid}]
        coordinates {(5,0.6464) (10,2.0784) (20,3.1831) (50,3.3924) (100,3.3883) (200,3.3811)};
      \addplot[color=exp10color2, line width=0.8pt, mark=*, mark size=1.2pt, mark options={solid}]
        coordinates {(5,4.7066) (10,4.7403) (20,4.4695) (50,4.5560) (100,4.5227) (200,4.5323)};
      \addplot[color=exp10color2, line width=0.8pt, dashed, mark=square*, mark size=1.0pt, mark options={solid}]
        coordinates {(5,1.0706) (10,2.6664) (20,3.9433) (50,4.5494) (100,4.5227) (200,4.5323)};
      \addplot[color=exp10color3, line width=0.8pt, mark=*, mark size=1.2pt, mark options={solid}]
        coordinates {(5,12.7260) (10,11.1728) (20,9.0347) (50,6.3524) (100,4.9091) (200,3.8523)};
      \addplot[color=exp10color3, line width=0.8pt, dashed, mark=square*, mark size=1.0pt, mark options={solid}]
        coordinates {(5,2.6260) (10,6.6660) (20,7.4240) (50,5.9360) (100,4.7534) (200,3.7666)};
      \addplot[color=exp10color4, line width=0.8pt, mark=*, mark size=1.2pt, mark options={solid}]
        coordinates {(5,80.8808) (10,75.3999) (20,75.6160) (50,75.5129) (100,75.3925) (200,75.2350)};
      \addplot[color=exp10color4, line width=0.8pt, dashed, mark=square*, mark size=1.0pt, mark options={solid}]
        coordinates {(5,52.0150) (10,65.5647) (20,71.8578) (50,74.1327) (100,74.6188) (200,74.7697)};
  \end{groupplot}
\end{tikzpicture}

%% file: Figures_final/double-left.tex
% Auto-generated by exp8_multi.py --pgfplots.  Do not edit by hand.
% Required preamble in the host document:
%   \usepackage{tikz}
%   \usepackage{pgfplots}
%   \usepgfplotslibrary{groupplots}
%   \pgfplotsset{compat=1.18}
%
\definecolor{exp8color0}{rgb}{0.1216,0.4667,0.7059}
\definecolor{exp8color1}{rgb}{1.0000,0.4980,0.0549}
\definecolor{exp8color2}{rgb}{0.1725,0.6275,0.1725}
\definecolor{exp8color3}{rgb}{0.8392,0.1529,0.1569}
\definecolor{exp8color4}{rgb}{0.5804,0.4039,0.7412}

\begin{tikzpicture}
  \begin{groupplot}[
    group style={
      group size=3 by 1,
      horizontal sep=22pt,
    },
    width=4.4cm,
    height=4.2cm,
    xlabel={$n$},
    ymin=0,
    title style={font=\small,yshift=-2pt},
    label style={font=\small},
    tick label style={font=\footnotesize},
    minor y tick num=4,
    grid=both,
    minor grid style={dotted, gray!20},
    major grid style={dotted, gray!40},
  ]
    \nextgroupplot[title={\textsc{MMachine-2D}},
    ymax=4.931,
    ylabel={$\pind$},
    ]
      \draw[black, dotted, line width=0.6pt] (axis cs:3,1) -- (axis cs:14,1);
      \addplot[draw=none, fill=exp8color3, fill opacity=0.18, forget plot]
        coordinates {(3,0.6472) (4,0.9941) (5,0.8829) (6,1.2957) (7,1.0436) (8,1.1600) (9,1.3177) (10,1.1021) (11,1.1664) (12,0.9337) (13,1.0688) (14,1.0660) (14,0.5612) (13,0.5253) (12,0.5056) (11,0.4919) (10,0.5732) (9,0.4224) (8,0.5122) (7,0.4687) (6,0.4666) (5,0.4318) (4,0.3897) (3,0.3582)}
        --cycle;
      \addplot[color=exp8color3, line width=0.8pt, mark=*, mark size=1.2pt, mark options={solid}]
        coordinates {(3,0.5015) (4,0.8472) (5,0.7276) (6,0.9897) (7,0.8664) (8,0.9977) (9,0.9993) (10,1.0009) (11,0.9824) (12,0.8121) (13,0.9321) (14,1.0637)};
      \addplot[draw=none, fill=exp8color4, fill opacity=0.18, forget plot]
        coordinates {(3,1.2470) (4,1.5621) (5,1.9822) (6,2.4002) (7,2.7660) (8,2.9961) (9,3.1799) (10,3.4933) (11,3.7783) (12,3.8602) (13,4.4434) (14,4.5660) (14,3.5114) (13,3.2063) (12,2.9858) (11,2.9196) (10,2.5707) (9,2.3383) (8,2.1402) (7,1.8924) (6,1.5618) (5,1.3418) (4,1.2428) (3,1.1327)}
        --cycle;
      \addplot[color=exp8color4, line width=0.8pt, mark=*, mark size=1.2pt, mark options={solid}]
        coordinates {(3,1.1892) (4,1.4586) (5,1.6855) (6,2.0318) (7,2.2826) (8,2.5706) (9,2.8053) (10,3.0532) (11,3.3610) (12,3.4553) (13,3.8439) (14,4.0854)};
      \addplot[draw=none, fill=exp8color1, fill opacity=0.18, forget plot]
        coordinates {(3,1.2465) (4,1.5691) (5,1.9378) (6,2.2940) (7,2.4243) (8,2.5846) (9,2.9100) (10,2.9724) (11,3.2307) (12,3.1212) (13,3.4901) (14,3.7246) (14,2.2696) (13,2.2855) (12,2.1667) (11,2.0788) (10,1.8507) (9,1.9510) (8,1.8013) (7,1.6099) (6,1.4632) (5,1.2831) (4,1.1460) (3,1.0893)}
        --cycle;
      \addplot[color=exp8color1, line width=0.8pt, mark=*, mark size=1.2pt, mark options={solid}]
        coordinates {(3,1.1771) (4,1.4119) (5,1.6065) (6,1.8637) (7,2.0512) (8,2.2061) (9,2.4503) (10,2.4703) (11,2.6768) (12,2.6800) (13,2.9031) (14,3.0704)};
      \addplot[draw=none, fill=exp8color0, fill opacity=0.18, forget plot]
        coordinates {(3,0.6926) (4,1.0705) (5,0.9693) (6,1.2958) (7,0.9959) (8,1.1651) (9,1.3021) (10,1.0274) (11,1.0703) (12,0.7979) (13,0.9653) (14,0.8550) (14,0.1808) (13,0.2009) (12,0.2234) (11,0.1868) (10,0.3634) (9,0.2589) (8,0.1841) (7,0.2241) (6,0.3198) (5,0.3349) (4,0.2897) (3,0.2740)}
        --cycle;
      \addplot[color=exp8color0, line width=0.8pt, mark=*, mark size=1.2pt, mark options={solid}]
        coordinates {(3,0.5036) (4,0.8042) (5,0.6987) (6,0.8754) (7,0.7537) (8,0.8228) (9,0.9104) (10,0.8583) (11,0.8087) (12,0.6434) (13,0.7068) (14,0.7348)};
      \addplot[draw=none, fill=exp8color2, fill opacity=0.18, forget plot]
        coordinates {(3,0.4076) (4,0.5508) (5,0.5183) (6,0.5892) (7,0.5724) (8,0.6180) (9,0.6415) (10,0.6567) (11,0.6465) (12,0.5835) (13,0.6207) (14,0.5874) (14,0.1146) (13,0.1471) (12,0.1168) (11,0.0935) (10,0.1685) (9,0.1372) (8,0.1226) (7,0.1285) (6,0.1392) (5,0.1364) (4,0.1584) (3,0.1679)}
        --cycle;
      \addplot[color=exp8color2, line width=0.8pt, mark=*, mark size=1.2pt, mark options={solid}]
        coordinates {(3,0.2932) (4,0.4134) (5,0.3686) (6,0.3913) (7,0.3951) (8,0.4562) (9,0.4278) (10,0.4272) (11,0.4438) (12,0.3721) (13,0.4534) (14,0.3945)};
    \nextgroupplot[title={\textsc{FoodRescue-2D}},
    ymax=5.162,
    ]
      \draw[black, dotted, line width=0.6pt] (axis cs:3,1) -- (axis cs:15,1);
      \addplot[draw=none, fill=exp8color3, fill opacity=0.18, forget plot]
        coordinates {(3,0.5788) (4,1.6521) (5,0.7411) (6,1.5591) (7,1.3798) (8,1.6006) (9,1.4166) (10,1.9971) (11,1.4946) (12,1.5663) (13,1.4893) (14,1.5380) (15,1.4465) (15,0.6004) (14,0.7023) (13,0.6692) (12,0.6823) (11,0.6039) (10,0.6114) (9,0.6153) (8,0.6282) (7,0.4684) (6,0.4582) (5,0.3754) (4,0.3427) (3,0.2594)}
        --cycle;
      \addplot[color=exp8color3, line width=0.8pt, mark=*, mark size=1.2pt, mark options={solid}]
        coordinates {(3,0.4074) (4,0.8633) (5,0.5749) (6,1.0549) (7,1.1267) (8,1.1843) (9,1.1309) (10,1.3575) (11,1.0981) (12,1.2441) (13,1.1252) (14,1.0670) (15,0.9987)};
      \addplot[draw=none, fill=exp8color4, fill opacity=0.18, forget plot]
        coordinates {(3,1.2127) (4,1.9060) (5,2.0589) (6,2.4854) (7,3.0693) (8,3.2912) (9,3.4732) (10,3.8841) (11,4.1941) (12,4.4876) (13,4.6505) (14,4.6620) (15,4.7799) (15,4.4675) (14,4.4294) (13,4.1885) (12,3.5967) (11,3.3399) (10,3.0782) (9,2.5574) (8,2.4234) (7,2.1808) (6,1.9405) (5,1.2932) (4,1.2213) (3,1.0900)}
        --cycle;
      \addplot[color=exp8color4, line width=0.8pt, mark=*, mark size=1.2pt, mark options={solid}]
        coordinates {(3,1.1502) (4,1.4634) (5,1.6173) (6,2.1366) (7,2.5270) (8,2.8046) (9,3.1142) (10,3.4662) (11,3.6533) (12,4.0474) (13,4.3433) (14,4.5181) (15,4.6124)};
      \addplot[draw=none, fill=exp8color1, fill opacity=0.18, forget plot]
        coordinates {(3,1.2149) (4,1.7997) (5,1.9088) (6,2.4114) (7,2.8240) (8,3.0080) (9,3.2714) (10,3.6344) (11,3.6235) (12,4.0578) (13,4.1781) (14,4.2973) (15,4.4233) (15,2.1930) (14,3.3609) (13,3.1112) (12,2.9858) (11,2.7334) (10,2.5047) (9,2.3644) (8,2.2370) (7,2.0143) (6,1.7757) (5,1.2489) (4,1.1549) (3,1.0576)}
        --cycle;
      \addplot[color=exp8color1, line width=0.8pt, mark=*, mark size=1.2pt, mark options={solid}]
        coordinates {(3,1.1365) (4,1.4116) (5,1.5570) (6,2.0545) (7,2.3749) (8,2.5827) (9,2.8097) (10,3.0252) (11,3.1067) (12,3.3668) (13,3.4813) (14,3.5622) (15,3.4798)};
      \addplot[draw=none, fill=exp8color0, fill opacity=0.18, forget plot]
        coordinates {(3,0.6030) (4,1.5084) (5,0.7283) (6,1.5505) (7,1.5104) (8,1.5191) (9,1.5354) (10,1.5907) (11,1.5722) (12,1.8041) (13,1.6794) (14,1.6866) (15,1.6775) (15,0.1325) (14,0.2036) (13,0.2012) (12,0.4834) (11,0.2115) (10,0.2891) (9,0.3693) (8,0.3723) (7,0.2654) (6,0.3488) (5,0.2378) (4,0.3043) (3,0.1642)}
        --cycle;
      \addplot[color=exp8color0, line width=0.8pt, mark=*, mark size=1.2pt, mark options={solid}]
        coordinates {(3,0.3926) (4,0.7943) (5,0.5076) (6,0.9850) (7,1.0058) (8,1.0566) (9,1.0124) (10,1.1202) (11,0.9955) (12,1.2155) (13,1.0197) (14,0.9507) (15,0.8429)};
      \addplot[draw=none, fill=exp8color2, fill opacity=0.18, forget plot]
        coordinates {(3,0.3672) (4,0.4216) (5,0.4178) (6,0.4228) (7,0.5455) (8,0.4767) (9,0.4192) (10,0.4553) (11,0.4213) (12,0.5076) (13,0.3860) (14,0.3306) (15,0.4063) (15,0.0782) (14,0.1369) (13,0.1249) (12,0.1379) (11,0.0959) (10,0.1209) (9,0.1258) (8,0.1563) (7,0.1387) (6,0.1004) (5,0.0894) (4,0.1363) (3,0.0911)}
        --cycle;
      \addplot[color=exp8color2, line width=0.8pt, mark=*, mark size=1.2pt, mark options={solid}]
        coordinates {(3,0.2271) (4,0.3597) (5,0.2719) (6,0.2865) (7,0.4533) (8,0.3463) (9,0.3032) (10,0.3766) (11,0.2807) (12,0.3530) (13,0.2944) (14,0.2757) (15,0.2479)};
    \nextgroupplot[title={\textsc{Kidney-2D}},
    ymax=5.069,
    ]
      \draw[black, dotted, line width=0.6pt] (axis cs:3,1) -- (axis cs:15,1);
      \addplot[draw=none, fill=exp8color3, fill opacity=0.18, forget plot]
        coordinates {(3,0.6306) (4,1.0021) (5,0.8493) (6,1.0379) (7,1.2107) (8,1.2741) (9,1.4167) (10,1.0508) (11,1.1550) (13,1.1108) (14,1.0875) (15,1.2030) (15,0.5117) (14,0.5793) (13,0.4448) (11,0.5050) (10,0.4336) (9,0.5299) (8,0.4626) (7,0.4344) (6,0.4368) (5,0.3868) (4,0.4355) (3,0.3128)}
        --cycle;
      \addplot[color=exp8color3, line width=0.8pt, mark=*, mark size=1.2pt, mark options={solid}]
        coordinates {(3,0.4852) (4,0.8303) (5,0.6485) (6,0.9018) (7,0.9353) (8,0.9423) (9,1.1208) (10,0.9549) (11,0.9199) (13,0.8791) (14,0.9959) (15,0.9301)};
      \addplot[draw=none, fill=exp8color4, fill opacity=0.18, forget plot]
        coordinates {(3,1.2405) (4,1.5602) (5,1.9780) (6,2.2832) (7,2.6719) (8,2.8965) (9,3.5616) (10,3.3668) (11,3.5935) (13,4.2277) (14,4.5334) (15,4.6937) (15,3.6234) (14,3.3773) (13,3.1493) (11,2.6445) (10,2.4741) (9,2.4206) (8,2.1448) (7,1.8685) (6,1.5234) (5,1.2926) (4,1.2605) (3,1.1109)}
        --cycle;
      \addplot[color=exp8color4, line width=0.8pt, mark=*, mark size=1.2pt, mark options={solid}]
        coordinates {(3,1.1825) (4,1.4534) (5,1.6014) (6,1.9843) (7,2.2661) (8,2.5568) (9,2.9774) (10,2.9783) (11,3.1554) (13,3.7121) (14,4.0618) (15,4.2348)};
      \addplot[draw=none, fill=exp8color1, fill opacity=0.18, forget plot]
        coordinates {(3,1.2360) (4,1.5663) (5,1.8056) (6,2.1609) (7,2.4122) (8,2.7042) (9,2.9754) (10,2.9210) (11,2.9351) (13,3.3919) (14,3.5581) (15,3.6708) (15,2.3305) (14,2.1567) (13,2.2427) (11,1.7552) (10,2.0031) (9,2.1148) (8,1.8815) (7,1.5239) (6,1.4872) (5,1.2005) (4,1.1752) (3,1.0902)}
        --cycle;
      \addplot[color=exp8color1, line width=0.8pt, mark=*, mark size=1.2pt, mark options={solid}]
        coordinates {(3,1.1761) (4,1.4124) (5,1.5208) (6,1.8266) (7,2.0595) (8,2.2810) (9,2.5692) (10,2.4745) (11,2.4163) (13,2.7690) (14,2.8815) (15,2.9529)};
      \addplot[draw=none, fill=exp8color0, fill opacity=0.18, forget plot]
        coordinates {(3,0.6713) (4,1.0836) (5,0.8729) (6,1.0834) (7,1.2289) (8,1.3035) (9,1.3114) (10,1.0634) (11,1.0779) (13,0.8300) (14,0.9919) (15,0.8413) (15,0.1752) (14,0.2896) (13,0.1808) (11,0.3394) (10,0.2070) (9,0.2529) (8,0.2507) (7,0.2166) (6,0.3365) (5,0.2742) (4,0.3604) (3,0.2663)}
        --cycle;
      \addplot[color=exp8color0, line width=0.8pt, mark=*, mark size=1.2pt, mark options={solid}]
        coordinates {(3,0.4945) (4,0.8055) (5,0.6010) (6,0.8036) (7,0.8444) (8,0.8195) (9,0.9318) (10,0.7653) (11,0.8025) (13,0.5958) (14,0.7712) (15,0.5982)};
      \addplot[draw=none, fill=exp8color2, fill opacity=0.18, forget plot]
        coordinates {(3,0.3941) (4,0.6467) (5,0.5236) (6,0.5175) (7,0.5565) (8,0.5969) (9,0.5747) (10,0.6153) (11,0.6560) (13,0.5609) (14,0.6065) (15,0.6865) (15,0.1365) (14,0.1430) (13,0.1553) (11,0.1506) (10,0.1253) (9,0.1511) (8,0.1308) (7,0.1707) (6,0.1227) (5,0.1585) (4,0.1472) (3,0.1436)}
        --cycle;
      \addplot[color=exp8color2, line width=0.8pt, mark=*, mark size=1.2pt, mark options={solid}]
        coordinates {(3,0.2867) (4,0.4177) (5,0.3591) (6,0.4113) (7,0.4584) (8,0.4188) (9,0.4708) (10,0.4283) (11,0.4511) (13,0.4086) (14,0.4298) (15,0.4461)};
  \end{groupplot}
\end{tikzpicture}

%% file: Figures/batch_prop_UB_diagram.tex
\begin{tikzpicture}[
    >=Stealth,
    dot/.style={circle, fill=black, inner sep=1pt},
    small_dot/.style={circle, fill=black, inner sep=0.8pt},
    font=\sffamily
]

    % --- Parameters ---
    \def\R{3.5}              % Radius of the main circle
    \def\angF{55}            % Angle for theta_f
    \def\angV{117.5}         % Angle for v (bisecting 180 and 55)
    \def\angFperp{145}       % Angle orthogonal to theta_f (\angF + 90)
    
    % --- Colors ---
    \colorlet{colT1}{purple!80!black}
    \colorlet{colT2}{green!40!black}
    \colorlet{colTf}{cyan!80!blue}
    \colorlet{colV}{violet!80!magenta}
    
    % --- Main Axes & Circle ---
    \draw[gray!70, thin] (-4.5, 0) -- (4.5, 0);
    \draw[gray!70, thin] (0, -4) -- (0, 4);
    
    % Shaded Region (Intersection of halfspaces)
    \fill[colV!15] (0,0) -- (0,\R) arc (90:\angFperp:\R) -- cycle;
    
    % Main Circle
    \draw[thick] (0,0) circle (\R);
    
    % Orthogonal Dotted Lines
    \draw[colT1, dotted, thick] (0, -4) -- (0, 4); 
    \draw[colTf, dotted, thick] (\angFperp:\R+0.5) -- (\angFperp-180:\R+0.5);

    % --- Main Vectors ---
    \draw[->, very thick, colT1] (0,0) -- (180:\R) node[above left] {\Large $\theta^{(1)}$};
    \draw[->, very thick, colT2] (0,0) -- (0:\R) node[above right] {\Large $\theta^{(2)}$};
    \draw[->, very thick, colTf] (0,0) -- (\angF:\R) node[above] {\Large $\theta_f$};
    \draw[->, thick, colV] (0,0) -- (\angV:\R) node[above left] {\Large $v$};

    % ==========================================
    % --- Left Inset: Main Ball B_R(y) ---
    % ==========================================
    \coordinate (Y1) at (-1.4, -1.2);
    \fill[gray!15, draw=gray!40] (Y1) circle (1.2);
    
    % Shifted label to the upper right quadrant of the ball
    \node[font=\scriptsize, text=gray!80!black] at ($(Y1)+(45:0.75)$) {$\overline{B}_R(y)$};
    
    % Support line for the points
    \draw[colV!80] ($(Y1) - (\angV:1.15)$) -- ($(Y1) + (\angV:1.15)$);
    
    % Points and Labels
    \node[small_dot, label={right:\small $y$}] at (Y1) {};
    
    \coordinate (x1s_1) at ($(Y1) + (\angV:0.5)$);
    \coordinate (x2s_1) at ($(Y1) - (\angV:0.3)$);
    \coordinate (x3s_1) at ($(Y1) - (\angV:0.8)$);
    
    \node[small_dot, label={left:\small $x_1^*$}] at (x1s_1) {};
    \node[small_dot, label={left:\small $x_2^*$}] at (x2s_1) {};
    \node[small_dot, label={left:\small $x_3^*$}] at (x3s_1) {};

    % ==========================================
    % --- Right Inset: Magnified Detail ---
    % ==========================================
    \coordinate (Y2) at (7.5, 0);
    \fill[gray!10, draw=gray!40] (Y2) circle (2);
    \node[right, font=\small, text=gray!80!black] at ($(Y2)+(0.4, 0.8)$) {$\overline{B}_R(y)$};
    
    % Support line
    \draw[colV!80] ($(Y2) - (\angV:1.9)$) -- ($(Y2) + (\angV:1.9)$);
    
    % Center y
    % \node[dot, label={right:\small $y$}] at (Y2) {};
    
    % Asterisk Points
    \coordinate (x1s_2) at ($(Y2) + (\angV:0.8)$);
    \coordinate (x2s_2) at ($(Y2) - (\angV:0.5)$);
    \coordinate (x3s_2) at ($(Y2) - (\angV:1.4)$);
    
    % Epsilon Neighborhoods (Dashed Circles) - SHRUNK to 0.18
    \draw[dashed, fill=gray!30, draw=gray!70, thick] (x1s_2) circle (0.18);
    \draw[dashed, fill=gray!30, draw=gray!70, thick] (x2s_2) circle (0.18);
    \draw[dashed, fill=gray!30, draw=gray!70, thick] (x3s_2) circle (0.18);
    
    % Redraw asterisk points on top of gray fill
    \node[dot] at (x1s_2) {};
    \node[dot] at (x2s_2) {};
    \node[dot] at (x3s_2) {};
    
    % Sampled points: offset shrunk to 0.18 to match smaller epsilon balls
    \coordinate (x2_samp) at ($(x2s_2) + (\angV-90:0.18)$);
    \coordinate (x3_samp) at ($(x3s_2) + (190:0.18)$);
    
    \node[dot, label={right: $x_2$}] at (x2_samp) {};
    \node[dot, label={left: $x_3$}] at (x3_samp) {};
    
    % Orange difference vector connecting samples
    \draw[->, very thick, orange] (x3_samp) -- (x2_samp);

    % ==========================================
    % --- Origin Difference Vector ---
    % ==========================================
    % Calculated using the exact vector difference between x2_samp and x3_samp.
    % Because the epsilon balls are smaller, the difference vector naturally 
    % leans more vertical, causing it to fall cleanly inside the purple halfspace.
    \draw[->, very thick, orange] (0,0) -- ($ (x2_samp) - (x3_samp) $);

\end{tikzpicture}

%% file: Figures/arith_ang_helper.tex
\begin{tikzpicture}[>={Stealth}]
    \begin{scope}[shift={(-2, 0)}]
    % Parameters
    \def\R{2}        % Radius of the circle
    \def\di{55}      % Angle d_i in degrees from the vertical y-axis

    % Coordinates
    \coordinate (O) at (0,0);
    \coordinate (Top) at (0,\R);
    \coordinate (P) at ({\R*sin(\di)}, {\R*cos(\di)});
    \coordinate (Yproj) at (0, {\R*cos(\di)});
    \coordinate (Xproj) at ({\R*sin(\di)}, \R);

    % Horizontal gray dashed line
    \draw[gray, dashed, thick] (-2.5, \R) -- (2.5, \R);

    % Main Circle
    \draw[very thick, color=black!90] (O) circle (\R);

    % Central dot
    \filldraw[black] (O) circle (1.5pt);

    %Top Dot
    \filldraw[black] (Top) circle (1.5pt);

    % Vertical dashed line
    \draw[thick, dashed] (O) -- (Top);

    % Line to point on the circle (theta_i)
    \draw[thick, dashed] (O) -- (P);
    
    % Point dot and theta_i label
    \filldraw[black] (P) circle (1.5pt);
    \node[above right=1pt] at (P) {$\scriptstyle\thetai$};
    
    % Top label (theta_avg)
    \node[above] at (Top) {$\scriptstyle\thang$};

    % Purple Angle d_i (with hatching pattern)
    \fill[fill=purple!30] 
        (O) -- (0, 0.8) arc (90:{90-\di}:0.8) -- cycle;
    \draw[purple!80, thick] (0, 0.8) arc (90:{90-\di}:0.8);
    \node[purple!80, below right] at (0.1, 0.1) {$\scriptstyle d_i$};

    % Red Vector (cos d_i)
    \draw[->, red!80!orange, very thick] (O) -- (Yproj);
    \node[red!80!orange, left=3pt] at (0, {\R*cos(\di)/2}) {$\scriptstyle\cos(d_i)$};

    % Orange dashed line projecting to the point
    \draw[green!60!black, dashed, thick] (Yproj) -- (P);

    % Green Vector (w_i = v_i * sin d_i)
    \draw[->, green!60!black, very thick] (Top) -- (Xproj);
    \node[green!60!black, above=3pt] at ({0.5 + \R*sin(\di)/2}, \R) {$\scriptstyle w_i = v_i \sin(d_i)$};

    % Green dashed line projecting down to the point
    \draw[black, dashed, thick] (Xproj) -- (P);
    \end{scope}

    \begin{scope}[shift={(5, 0)}]
    % Parameters
    \def\R{2}        % Radius of the circle
    \def\ang{45}     % Angle of the red line
    \def\r{1.5}      % Length of the red line

    % Coordinates
    \coordinate (O) at (0,0);
    \coordinate (Top) at (0,\R);
    \coordinate (P) at ({\r*cos(\ang)}, {\r*sin(\ang)});
    \coordinate (Atip) at (0, {\r*sin(\ang)});

    %Top Dot
    \filldraw[black] (Top) circle (1.5pt);

    % Circle
    \draw[very thick, color=black!90] (O) circle (\R);

    % Vertical Line
    \draw[thick, color=black!90] (O) -- (Top);
    \node[above, color=black!90] at (Top) {$\scriptstyle\thang$};

    % Purple Angle Fill
    % Fills the sector from the y-axis (90 degrees) to the red line (\ang degrees)
    \fill[purple!30] (O) -- (0, 0.8) arc (90:\ang:0.8) -- cycle;
    \draw[purple!80, thick] (0, 0.8) arc (90:{90-\ang}:0.8);
    
    % Purple Label for the distance/angle
    \node[purple!80, below right] at (0, 0.1) {$\scriptstyle\scriptstyle\dang(\thang, \thar)$};

    % Center dot
    \filldraw[black] (O) circle (1.5pt);

    % Green Vector A
    \draw[green!60!black, thick] (O) -- (Atip);
    \node[left=2pt, green!60!black] at (0, {\r*sin(\ang)/2}) {$\scriptstyle A$};

    % Red Line
    \draw[red!80!orange, very thick] (O) -- (P);
    \filldraw[red!80!orange] (P) circle (1.5pt);
    \node[below right=.3pt, red!80!orange] at (P) {$\scriptstyle\thar$};

    % Light Blue dashed line ||B||
    \draw[cyan, thick] (Atip) -- (P);
    \node[above=2pt, cyan] at ({\r*cos(\ang)/2}, {\r*sin(\ang)}) {$\scriptstyle\|B\|$};
    \end{scope}

\end{tikzpicture}

%% file: Figures_final/exp8_multi_subset_propcol.tex
% Auto-generated by exp8_multi.py.  Do not edit by hand.
% Required preamble in the host document:
%   \usepackage{tikz}
%   \usepackage{pgfplots}
%   \usepgfplotslibrary{groupplots}
%   \pgfplotsset{compat=1.18}
%
\definecolor{exp8color0}{rgb}{0.1216,0.4667,0.7059}
\definecolor{exp8color1}{rgb}{1.0000,0.4980,0.0549}
\definecolor{exp8color2}{rgb}{0.1725,0.6275,0.1725}
\definecolor{exp8color3}{rgb}{0.8392,0.1529,0.1569}
\definecolor{exp8color4}{rgb}{0.5804,0.4039,0.7412}

\begin{tikzpicture}
  \begin{groupplot}[
    group style={
      group size=3 by 1,
      horizontal sep=22pt,
    },
    width=4.4cm,
    height=4.2cm,
    xlabel={$n$},
    ymin=0,
    title style={font=\small,yshift=-2pt},
    label style={font=\small},
    tick label style={font=\footnotesize},
    minor y tick num=4,
    grid=both,
    minor grid style={dotted, gray!20},
    major grid style={dotted, gray!40},
  ]
    \nextgroupplot[title={\textsc{MMachine}-2D},
    ymax=5.608,
    ylabel={$\pcol$},
    ]
      \draw[black, dotted, line width=0.6pt] (axis cs:3,1) -- (axis cs:14,1);
      \addplot[draw=none, fill=exp8color3, fill opacity=0.18, forget plot]
        coordinates {(3,0.6268) (4,1.2371) (5,0.9970) (6,1.6955) (7,1.5496) (8,1.6233) (9,1.6131) (10,1.9037) (11,1.5459) (12,1.4542) (13,1.3533) (14,1.5964) (14,0.5557) (13,0.4943) (12,0.6517) (11,0.5274) (10,0.5566) (9,0.5685) (8,0.4436) (7,0.5386) (6,0.4285) (5,0.4844) (4,0.3828) (3,0.3111)}
        --cycle;
      \addplot[color=exp8color3, line width=0.8pt, mark=*, mark size=1.2pt, mark options={solid}]
        coordinates {(3,0.4784) (4,0.8534) (5,0.7784) (6,1.0753) (7,1.0816) (8,1.1201) (9,1.1644) (10,1.3241) (11,1.2224) (12,1.3435) (13,1.1741) (14,1.3011)};
      \addplot[draw=none, fill=exp8color4, fill opacity=0.18, forget plot]
        coordinates {(3,1.2365) (4,1.6892) (5,2.0352) (6,2.5748) (7,2.8963) (8,3.1259) (9,3.5816) (10,3.8103) (11,4.1269) (12,4.4929) (13,4.5965) (14,5.1930) (14,3.5495) (13,3.3484) (12,3.2468) (11,3.0049) (10,2.8025) (9,2.4379) (8,2.1332) (7,2.0970) (6,1.4574) (5,1.3852) (4,1.2291) (3,1.1146)}
        --cycle;
      \addplot[color=exp8color4, line width=0.8pt, mark=*, mark size=1.2pt, mark options={solid}]
        coordinates {(3,1.1797) (4,1.4725) (5,1.7227) (6,2.0486) (7,2.4984) (8,2.6211) (9,2.9880) (10,3.3188) (11,3.5791) (12,3.8275) (13,3.9436) (14,4.3796)};
      \addplot[draw=none, fill=exp8color1, fill opacity=0.18, forget plot]
        coordinates {(3,1.2306) (4,1.5873) (5,1.9122) (6,2.2386) (7,2.6089) (8,2.7288) (9,3.0614) (10,3.3420) (11,3.4026) (12,3.6727) (13,3.5416) (14,4.1588) (14,2.5201) (13,2.1284) (12,1.9579) (11,2.1042) (10,2.1653) (9,1.9043) (8,1.8283) (7,1.7728) (6,1.3469) (5,1.3106) (4,1.1621) (3,1.0693)}
        --cycle;
      \addplot[color=exp8color1, line width=0.8pt, mark=*, mark size=1.2pt, mark options={solid}]
        coordinates {(3,1.1611) (4,1.3890) (5,1.6273) (6,1.8469) (7,2.2118) (8,2.2546) (9,2.4893) (10,2.6757) (11,2.7282) (12,2.9208) (13,2.7361) (14,3.2293)};
      \addplot[draw=none, fill=exp8color0, fill opacity=0.18, forget plot]
        coordinates {(3,0.6667) (4,1.1074) (5,1.0855) (6,1.6517) (7,1.5570) (8,1.5093) (9,1.5613) (10,1.5711) (11,1.0773) (12,1.3960) (13,1.0853) (14,1.0677) (14,0.2358) (13,0.2690) (12,0.3295) (11,0.2122) (10,0.3313) (9,0.3097) (8,0.2219) (7,0.3131) (6,0.3029) (5,0.3884) (4,0.2643) (3,0.2170)}
        --cycle;
      \addplot[color=exp8color0, line width=0.8pt, mark=*, mark size=1.2pt, mark options={solid}]
        coordinates {(3,0.4781) (4,0.7758) (5,0.7619) (6,0.9714) (7,0.9888) (8,0.9745) (9,1.0632) (10,1.1197) (11,0.9045) (12,1.1555) (13,0.8762) (14,0.9221)};
      \addplot[draw=none, fill=exp8color2, fill opacity=0.18, forget plot]
        coordinates {(3,0.3832) (4,0.5381) (5,0.6061) (6,0.6911) (7,0.6630) (8,0.6099) (9,0.6229) (10,0.5149) (11,0.5735) (12,0.7805) (13,0.6794) (14,0.6125) (14,0.1461) (13,0.1561) (12,0.1444) (11,0.1464) (10,0.1594) (9,0.1744) (8,0.1688) (7,0.1322) (6,0.1756) (5,0.1932) (4,0.1294) (3,0.1161)}
        --cycle;
      \addplot[color=exp8color2, line width=0.8pt, mark=*, mark size=1.2pt, mark options={solid}]
        coordinates {(3,0.2783) (4,0.3831) (5,0.4134) (6,0.4936) (7,0.4402) (8,0.4590) (9,0.4656) (10,0.4340) (11,0.4318) (12,0.5205) (13,0.4468) (14,0.5037)};
    \nextgroupplot[title={\textsc{FoodRescue}-2D},
    ymax=5.022,
    ]
      \draw[black, dotted, line width=0.6pt] (axis cs:3,1) -- (axis cs:15,1);
      \addplot[draw=none, fill=exp8color3, fill opacity=0.18, forget plot]
        coordinates {(3,0.5788) (4,1.6082) (5,0.7311) (6,1.5591) (7,1.3794) (8,1.5968) (9,1.4155) (10,1.9971) (11,1.4858) (12,1.5661) (13,1.4520) (14,1.5379) (15,1.4465) (15,0.5847) (14,0.6804) (13,0.5934) (12,0.6630) (11,0.5093) (10,0.5457) (9,0.5882) (8,0.6244) (7,0.4521) (6,0.4450) (5,0.3688) (4,0.3427) (3,0.2594)}
        --cycle;
      \addplot[color=exp8color3, line width=0.8pt, mark=*, mark size=1.2pt, mark options={solid}]
        coordinates {(3,0.4074) (4,0.8256) (5,0.5692) (6,1.0370) (7,1.1011) (8,1.1704) (9,1.1063) (10,1.3222) (11,1.0732) (12,1.2210) (13,1.0848) (14,1.0456) (15,0.9675)};
      \addplot[draw=none, fill=exp8color4, fill opacity=0.18, forget plot]
        coordinates {(3,1.2127) (4,1.8965) (5,2.0194) (6,2.4848) (7,3.0427) (8,3.2854) (9,3.4516) (10,3.7462) (11,4.1710) (12,4.4281) (13,4.5497) (14,4.6083) (15,4.6501) (15,4.4580) (14,4.3962) (13,4.1827) (12,3.5531) (11,3.2900) (10,3.0731) (9,2.5410) (8,2.4087) (7,2.1474) (6,1.9387) (5,1.2932) (4,1.2213) (3,1.0900)}
        --cycle;
      \addplot[color=exp8color4, line width=0.8pt, mark=*, mark size=1.2pt, mark options={solid}]
        coordinates {(3,1.1502) (4,1.4556) (5,1.6118) (6,2.1276) (7,2.5051) (8,2.7887) (9,3.0875) (10,3.4253) (11,3.6150) (12,3.9968) (13,4.2780) (14,4.4661) (15,4.5352)};
      \addplot[draw=none, fill=exp8color1, fill opacity=0.18, forget plot]
        coordinates {(3,1.2038) (4,1.6650) (5,1.8562) (6,2.3323) (7,2.6389) (8,2.9364) (9,3.1850) (10,3.4473) (11,3.5811) (12,3.9281) (13,4.0897) (14,4.2200) (15,4.3600) (15,2.0542) (14,3.3231) (13,3.1038) (12,2.9567) (11,2.6514) (10,2.4692) (9,2.3589) (8,2.2122) (7,2.0075) (6,1.7586) (5,1.2477) (4,1.1528) (3,1.0501)}
        --cycle;
      \addplot[color=exp8color1, line width=0.8pt, mark=*, mark size=1.2pt, mark options={solid}]
        coordinates {(3,1.1281) (4,1.3636) (5,1.5348) (6,2.0064) (7,2.2911) (8,2.5155) (9,2.7372) (10,2.9317) (11,3.0189) (12,3.2650) (13,3.3610) (14,3.4470) (15,3.3913)};
      \addplot[draw=none, fill=exp8color0, fill opacity=0.18, forget plot]
        coordinates {(3,0.6029) (4,1.4816) (5,0.7283) (6,1.5502) (7,1.5100) (8,1.5176) (9,1.5354) (10,1.5794) (11,1.5722) (12,1.7789) (13,1.6793) (14,1.6863) (15,1.6775) (15,0.1325) (14,0.1431) (13,0.1115) (12,0.4834) (11,0.2087) (10,0.2666) (9,0.3406) (8,0.3723) (7,0.2478) (6,0.3488) (5,0.2378) (4,0.3043) (3,0.1642)}
        --cycle;
      \addplot[color=exp8color0, line width=0.8pt, mark=*, mark size=1.2pt, mark options={solid}]
        coordinates {(3,0.3925) (4,0.7614) (5,0.5069) (6,0.9718) (7,0.9865) (8,1.0480) (9,1.0057) (10,1.1088) (11,0.9856) (12,1.2023) (13,1.0013) (14,0.9438) (15,0.8416)};
      \addplot[draw=none, fill=exp8color2, fill opacity=0.18, forget plot]
        coordinates {(3,0.3672) (4,0.4216) (5,0.4079) (6,0.4228) (7,0.5455) (8,0.4767) (9,0.4192) (10,0.4553) (11,0.4213) (12,0.5076) (13,0.3860) (14,0.3306) (15,0.4063) (15,0.0782) (14,0.0962) (13,0.0952) (12,0.0959) (11,0.0752) (10,0.0846) (9,0.0739) (8,0.1385) (7,0.1027) (6,0.0950) (5,0.0825) (4,0.1334) (3,0.0911)}
        --cycle;
      \addplot[color=exp8color2, line width=0.8pt, mark=*, mark size=1.2pt, mark options={solid}]
        coordinates {(3,0.2271) (4,0.3555) (5,0.2678) (6,0.2833) (7,0.4438) (8,0.3413) (9,0.2912) (10,0.3674) (11,0.2710) (12,0.3434) (13,0.2742) (14,0.2685) (15,0.2357)};
    \nextgroupplot[title={\textsc{Kidney}-2D},
    ymax=4.899,
    ]
      \draw[black, dotted, line width=0.6pt] (axis cs:3,1) -- (axis cs:15,1);
      \addplot[draw=none, fill=exp8color3, fill opacity=0.18, forget plot]
        coordinates {(3,0.6306) (4,1.0021) (5,0.8476) (6,1.0379) (7,1.1901) (8,1.2399) (9,1.4073) (10,1.0508) (11,1.0786) (13,0.9697) (14,1.0683) (15,1.0849) (15,0.4639) (14,0.5055) (13,0.4118) (11,0.4758) (10,0.4116) (9,0.4686) (8,0.4347) (7,0.4314) (6,0.4211) (5,0.3681) (4,0.4354) (3,0.3128)}
        --cycle;
      \addplot[color=exp8color3, line width=0.8pt, mark=*, mark size=1.2pt, mark options={solid}]
        coordinates {(3,0.4852) (4,0.8156) (5,0.6431) (6,0.8753) (7,0.9159) (8,0.9144) (9,1.0769) (10,0.9163) (11,0.8653) (13,0.8057) (14,0.9142) (15,0.8398)};
      \addplot[draw=none, fill=exp8color4, fill opacity=0.18, forget plot]
        coordinates {(3,1.2405) (4,1.5602) (5,1.9664) (6,2.2794) (7,2.6134) (8,2.8963) (9,3.5168) (10,3.3069) (11,3.5091) (13,3.9993) (14,4.3858) (15,4.5357) (15,3.5041) (14,3.3078) (13,3.0418) (11,2.5532) (10,2.4384) (9,2.4178) (8,2.1259) (7,1.8606) (6,1.5234) (5,1.2926) (4,1.2605) (3,1.1109)}
        --cycle;
      \addplot[color=exp8color4, line width=0.8pt, mark=*, mark size=1.2pt, mark options={solid}]
        coordinates {(3,1.1825) (4,1.4500) (5,1.5979) (6,1.9682) (7,2.2490) (8,2.5345) (9,2.9381) (10,2.9384) (11,3.0864) (13,3.6180) (14,3.9616) (15,4.1155)};
      \addplot[draw=none, fill=exp8color1, fill opacity=0.18, forget plot]
        coordinates {(3,1.2278) (4,1.5583) (5,1.7653) (6,2.0963) (7,2.3797) (8,2.6701) (9,2.8824) (10,2.8938) (11,2.9043) (13,3.3263) (14,3.5198) (15,3.6500) (15,2.2810) (14,2.1550) (13,2.2394) (11,1.7435) (10,2.0009) (9,2.1079) (8,1.8790) (7,1.5223) (6,1.4846) (5,1.1992) (4,1.1737) (3,1.0822)}
        --cycle;
      \addplot[color=exp8color1, line width=0.8pt, mark=*, mark size=1.2pt, mark options={solid}]
        coordinates {(3,1.1678) (4,1.3858) (5,1.5012) (6,1.7916) (7,2.0132) (8,2.2440) (9,2.5062) (10,2.4159) (11,2.3748) (13,2.7182) (14,2.8283) (15,2.8849)};
      \addplot[draw=none, fill=exp8color0, fill opacity=0.18, forget plot]
        coordinates {(3,0.6713) (4,1.0831) (5,0.8729) (6,1.0834) (7,1.2289) (8,1.3033) (9,1.3114) (10,1.0407) (11,1.0360) (13,0.7731) (14,0.8995) (15,0.7373) (15,0.1291) (14,0.2128) (13,0.1455) (11,0.2673) (10,0.1985) (9,0.2509) (8,0.2150) (7,0.2163) (6,0.3364) (5,0.2742) (4,0.3604) (3,0.2663)}
        --cycle;
      \addplot[color=exp8color0, line width=0.8pt, mark=*, mark size=1.2pt, mark options={solid}]
        coordinates {(3,0.4944) (4,0.7916) (5,0.5997) (6,0.7869) (7,0.8323) (8,0.8070) (9,0.9132) (10,0.7399) (11,0.7580) (13,0.5626) (14,0.7117) (15,0.5539)};
      \addplot[draw=none, fill=exp8color2, fill opacity=0.18, forget plot]
        coordinates {(3,0.3940) (4,0.6467) (5,0.5236) (6,0.5175) (7,0.5552) (8,0.5768) (9,0.5635) (10,0.5544) (11,0.6106) (13,0.4970) (14,0.5145) (15,0.5794) (15,0.0948) (14,0.1062) (13,0.1309) (11,0.1413) (10,0.1159) (9,0.1392) (8,0.1176) (7,0.1665) (6,0.1192) (5,0.1453) (4,0.1371) (3,0.1436)}
        --cycle;
      \addplot[color=exp8color2, line width=0.8pt, mark=*, mark size=1.2pt, mark options={solid}]
        coordinates {(3,0.2867) (4,0.4172) (5,0.3560) (6,0.4007) (7,0.4522) (8,0.3992) (9,0.4524) (10,0.3924) (11,0.4168) (13,0.3494) (14,0.3756) (15,0.3835)};
  \end{groupplot}
\end{tikzpicture}

%% file: Figures_final/overview_table_minimal_appendix.tex
\begin{table*}[h]
\centering
\footnotesize
\setlength{\tabcolsep}{4pt}
\setlength{\extrarowheight}{2pt}
\caption{Long-run IP ($\pind$) and per-batch IP ($\pcol$) for each rule on \textsc{KidneyStudy1} and \textsc{KidneyStudy2} and its 2D variant.  Cell shading is proportional to $\log_{10}(\text{value})$: darker means more disproportional.}
\label{tab:overview-minimal-kid}
\begin{tabular}{l@{\hspace{10pt}}>{\centering\arraybackslash}p{0.55cm}@{\hspace{10pt}}>{\centering\arraybackslash}p{0.55cm}@{\hspace{10pt}}>{\centering\arraybackslash}p{0.55cm}@{\hspace{10pt}}>{\centering\arraybackslash}p{0.55cm}@{\hspace{10pt}}>{\centering\arraybackslash}p{0.55cm}@{\hspace{10pt}}>{\centering\arraybackslash}p{0.55cm}@{\hspace{10pt}}>{\centering\arraybackslash}p{0.55cm}@{\hspace{10pt}}>{\centering\arraybackslash}p{0.55cm}@{\hspace{10pt}}>{\centering\arraybackslash}p{0.55cm}@{\hspace{10pt}}>{\centering\arraybackslash}p{0.55cm}}
\toprule
 & \multicolumn{2}{c}{arithmetic} & \multicolumn{2}{c}{angular} & \multicolumn{2}{c}{geometric} & \multicolumn{2}{c}{Borda} & \multicolumn{2}{c}{PSB} \\
\cmidrule(lr){2-3} \cmidrule(lr){4-5} \cmidrule(lr){6-7} \cmidrule(lr){8-9} \cmidrule(lr){10-11}
dataset & \textsf{long} & \textsf{batch} & \textsf{long} & \textsf{batch} & \textsf{long} & \textsf{batch} & \textsf{long} & \textsf{batch} & \textsf{long} & \textsf{batch} \\
\midrule
\textsc{KidneyStudy1} & \cellcolor{black!22}12 & \cellcolor{black!22}11 & \cellcolor{black!22}12 & \cellcolor{black!22}11 & \cellcolor{black!22}12 & \cellcolor{black!22}11 & \cellcolor{black!22}12 & \cellcolor{black!22}11 & \cellcolor{black!22}12 & \cellcolor{black!22}11 \\
\textsc{KidneyStudy2} & \cellcolor{black!26}29 & \cellcolor{black!26}26 & \cellcolor{black!26}29 & \cellcolor{black!26}26 & \cellcolor{black!26}27 & \cellcolor{black!25}25 & \cellcolor{black!26}28 & \cellcolor{black!26}26 & \cellcolor{black!26}29 & \cellcolor{black!26}26 \\
\midrule
\textsc{KidneyStudy1}-2D & \cellcolor{black!16}3.4 & \cellcolor{black!16}3.3 & \cellcolor{black!17}3.7 & \cellcolor{black!16}3.6 & \cellcolor{black!17}3.6 & \cellcolor{black!16}3.5 & \cellcolor{black!16}3.5 & \cellcolor{black!16}3.4 & \cellcolor{black!17}3.8 & \cellcolor{black!17}3.8 \\
\textsc{KidneyStudy2}-2D & \cellcolor{black!16}3.3 & \cellcolor{black!16}3.3 & \cellcolor{black!17}3.9 & \cellcolor{black!17}3.8 & \cellcolor{black!15}2.8 & \cellcolor{black!15}2.8 & \cellcolor{black!16}3.2 & \cellcolor{black!16}3.2 & \cellcolor{black!17}4.3 & \cellcolor{black!17}4.3 \\
\bottomrule
\end{tabular}
\end{table*}

%% file: Figures_final/exp10_subset_appendix.tex
% Auto-generated by exp10_propind_propcol_vs_m.py.
% Required preamble in the host document:
%   \usepackage{tikz}
%   \usepackage{pgfplots}
%   \usepgfplotslibrary{groupplots}
%   \pgfplotsset{compat=1.18}
%
\definecolor{exp10color0}{rgb}{0.1216,0.4667,0.7059}
\definecolor{exp10color1}{rgb}{1.0000,0.4980,0.0549}
\definecolor{exp10color2}{rgb}{0.1725,0.6275,0.1725}
\definecolor{exp10color3}{rgb}{0.8392,0.1529,0.1569}
\definecolor{exp10color4}{rgb}{0.5804,0.4039,0.7412}

\begin{tikzpicture}
  \begin{groupplot}[
    group style={
      group size=2 by 1,
      horizontal sep=22pt,
    },
    width=4.4cm,
    height=4.2cm,
    xmode=log,
    log basis x=10,
    xtick={5,10,20,50,100,200},
    xticklabels={5,10,20,50,100,200},
    xlabel={$m$},
    title style={font=\small,yshift=-2pt},
    label style={font=\small},
    tick label style={font=\footnotesize},
    minor y tick num=4,
    grid=both,
    minor grid style={dotted, gray!20},
    major grid style={dotted, gray!40},
  ]
    \nextgroupplot[title={\textsc{KidneyStudy1}-2D},
    ylabel={$\pind$ / $\pcol$},
    ]
      \addplot[color=exp10color0, line width=0.8pt, mark=*, mark size=1.2pt, mark options={solid}]
        coordinates {(5,3.3533) (10,3.3856) (20,3.3685) (50,3.3434) (100,3.3547) (200,3.3534)};
      \addplot[color=exp10color0, line width=0.8pt, dashed, mark=square*, mark size=1.0pt, mark options={solid}]
        coordinates {(5,2.8611) (10,3.3292) (20,3.3664) (50,3.3434) (100,3.3547) (200,3.3534)};
      \addplot[color=exp10color1, line width=0.8pt, mark=*, mark size=1.2pt, mark options={solid}]
        coordinates {(5,3.6848) (10,3.6937) (20,3.6685) (50,3.6483) (100,3.6600) (200,3.6585)};
      \addplot[color=exp10color1, line width=0.8pt, dashed, mark=square*, mark size=1.0pt, mark options={solid}]
        coordinates {(5,3.0183) (10,3.5968) (20,3.6615) (50,3.6483) (100,3.6600) (200,3.6585)};
      \addplot[color=exp10color2, line width=0.8pt, mark=*, mark size=1.2pt, mark options={solid}]
        coordinates {(5,3.6261) (10,3.6384) (20,3.6158) (50,3.5972) (100,3.6081) (200,3.6067)};
      \addplot[color=exp10color2, line width=0.8pt, dashed, mark=square*, mark size=1.0pt, mark options={solid}]
        coordinates {(5,2.9996) (10,3.5496) (20,3.6100) (50,3.5972) (100,3.6081) (200,3.6067)};
      \addplot[color=exp10color3, line width=0.8pt, mark=*, mark size=1.2pt, mark options={solid}]
        coordinates {(5,3.4450) (10,3.4485) (20,3.4208) (50,3.4021) (100,3.4083) (200,3.4100)};
      \addplot[color=exp10color3, line width=0.8pt, dashed, mark=square*, mark size=1.0pt, mark options={solid}]
        coordinates {(5,3.0566) (10,3.4161) (20,3.4204) (50,3.4021) (100,3.4083) (200,3.4100)};
      \addplot[color=exp10color4, line width=0.8pt, mark=*, mark size=1.2pt, mark options={solid}]
        coordinates {(5,3.7757) (10,3.7462) (20,3.6959) (50,3.6687) (100,3.6541) (200,3.6341)};
      \addplot[color=exp10color4, line width=0.8pt, dashed, mark=square*, mark size=1.0pt, mark options={solid}]
        coordinates {(5,3.4808) (10,3.7285) (20,3.6958) (50,3.6687) (100,3.6541) (200,3.6341)};
    \nextgroupplot[title={\textsc{KidneyStudy2}-2D}]
      \addplot[color=exp10color0, line width=0.8pt, mark=*, mark size=1.2pt, mark options={solid}]
        coordinates {(5,3.4465) (10,3.2546) (20,3.3418) (50,3.3305) (100,3.3301) (200,3.3300)};
      \addplot[color=exp10color0, line width=0.8pt, dashed, mark=square*, mark size=1.0pt, mark options={solid}]
        coordinates {(5,3.1175) (10,3.2484) (20,3.3418) (50,3.3305) (100,3.3301) (200,3.3300)};
      \addplot[color=exp10color1, line width=0.8pt, mark=*, mark size=1.2pt, mark options={solid}]
        coordinates {(5,3.9839) (10,3.8012) (20,3.8992) (50,3.8795) (100,3.8802) (200,3.8793)};
      \addplot[color=exp10color1, line width=0.8pt, dashed, mark=square*, mark size=1.0pt, mark options={solid}]
        coordinates {(5,3.5303) (10,3.7826) (20,3.8992) (50,3.8795) (100,3.8802) (200,3.8793)};
      \addplot[color=exp10color2, line width=0.8pt, mark=*, mark size=1.2pt, mark options={solid}]
        coordinates {(5,2.8767) (10,2.7033) (20,2.7512) (50,2.7513) (100,2.7496) (200,2.7488)};
      \addplot[color=exp10color2, line width=0.8pt, dashed, mark=square*, mark size=1.0pt, mark options={solid}]
        coordinates {(5,2.6466) (10,2.7014) (20,2.7512) (50,2.7513) (100,2.7496) (200,2.7488)};
      \addplot[color=exp10color3, line width=0.8pt, mark=*, mark size=1.2pt, mark options={solid}]
        coordinates {(5,3.2809) (10,3.1294) (20,3.2336) (50,3.2060) (100,3.2061) (200,3.2064)};
      \addplot[color=exp10color3, line width=0.8pt, dashed, mark=square*, mark size=1.0pt, mark options={solid}]
        coordinates {(5,3.2443) (10,3.1294) (20,3.2336) (50,3.2060) (100,3.2061) (200,3.2064)};
      \addplot[color=exp10color4, line width=0.8pt, mark=*, mark size=1.2pt, mark options={solid}]
        coordinates {(5,4.1818) (10,4.1944) (20,4.3257) (50,4.3414) (100,4.3855) (200,4.4125)};
      \addplot[color=exp10color4, line width=0.8pt, dashed, mark=square*, mark size=1.0pt, mark options={solid}]
        coordinates {(5,4.1624) (10,4.1944) (20,4.3257) (50,4.3414) (100,4.3855) (200,4.4125)};
  \end{groupplot}
\end{tikzpicture}

%% file: Figures_final/exp8_multi_subset_appendix.tex
% Auto-generated by exp8_multi.py.  Do not edit by hand.
% Required preamble in the host document:
%   \usepackage{tikz}
%   \usepackage{pgfplots}
%   \usepgfplotslibrary{groupplots}
%   \pgfplotsset{compat=1.18}
%
\definecolor{exp8color0}{rgb}{0.1216,0.4667,0.7059}
\definecolor{exp8color1}{rgb}{1.0000,0.4980,0.0549}
\definecolor{exp8color2}{rgb}{0.1725,0.6275,0.1725}
\definecolor{exp8color3}{rgb}{0.8392,0.1529,0.1569}
\definecolor{exp8color4}{rgb}{0.5804,0.4039,0.7412}

\begin{tikzpicture}
  \begin{groupplot}[
    group style={
      group size=2 by 1,
      horizontal sep=22pt,
    },
    width=4.4cm,
    height=4.2cm,
    xlabel={$n$},
    ymin=0,
    title style={font=\small,yshift=-2pt},
    label style={font=\small},
    tick label style={font=\footnotesize},
    minor y tick num=4,
    grid=both,
    minor grid style={dotted, gray!20},
    major grid style={dotted, gray!40},
  ]
    \nextgroupplot[title={\textsc{KidneyStudy1}-2D},
    ymax=1.782,
    ylabel={$\pind$},
    ]
      \draw[black, dotted, line width=0.6pt] (axis cs:3,1) -- (axis cs:7,1);
      \addplot[draw=none, fill=exp8color3, fill opacity=0.18, forget plot]
        coordinates {(3,0.5996) (4,0.8530) (5,0.8264) (6,0.8023) (7,0.8642) (7,0.6753) (6,0.6299) (5,0.5702) (4,0.4514) (3,0.3486)}
        --cycle;
      \addplot[color=exp8color3, line width=0.8pt, mark=*, mark size=1.2pt, mark options={solid}]
        coordinates {(3,0.4814) (4,0.6447) (5,0.7029) (6,0.7268) (7,0.7820)};
      \addplot[draw=none, fill=exp8color4, fill opacity=0.18, forget plot]
        coordinates {(3,1.2247) (4,1.4626) (5,1.5250) (6,1.5447) (7,1.5987) (7,1.4762) (6,1.4268) (5,1.3680) (4,1.2323) (3,1.1273)}
        --cycle;
      \addplot[color=exp8color4, line width=0.8pt, mark=*, mark size=1.2pt, mark options={solid}]
        coordinates {(3,1.1807) (4,1.3478) (5,1.4421) (6,1.4968) (7,1.5479)};
      \addplot[draw=none, fill=exp8color1, fill opacity=0.18, forget plot]
        coordinates {(3,1.2266) (4,1.4551) (5,1.5554) (6,1.5776) (7,1.6501) (7,1.4832) (6,1.4232) (5,1.3378) (4,1.1984) (3,1.0900)}
        --cycle;
      \addplot[color=exp8color1, line width=0.8pt, mark=*, mark size=1.2pt, mark options={solid}]
        coordinates {(3,1.1693) (4,1.3311) (5,1.4411) (6,1.5083) (7,1.5768)};
      \addplot[draw=none, fill=exp8color0, fill opacity=0.18, forget plot]
        coordinates {(3,0.6410) (4,0.8945) (5,0.9051) (6,0.8629) (7,0.9113) (7,0.6770) (6,0.6351) (5,0.5547) (4,0.4106) (3,0.2737)}
        --cycle;
      \addplot[color=exp8color0, line width=0.8pt, mark=*, mark size=1.2pt, mark options={solid}]
        coordinates {(3,0.4814) (4,0.6497) (5,0.7311) (6,0.7601) (7,0.8060)};
      \addplot[draw=none, fill=exp8color2, fill opacity=0.18, forget plot]
        coordinates {(3,0.3516) (4,0.4853) (5,0.6037) (6,0.7140) (7,0.8075) (7,0.6220) (6,0.5462) (5,0.4419) (4,0.3636) (3,0.1640)}
        --cycle;
      \addplot[color=exp8color2, line width=0.8pt, mark=*, mark size=1.2pt, mark options={solid}]
        coordinates {(3,0.2730) (4,0.4075) (5,0.5244) (6,0.6134) (7,0.7210)};
    \nextgroupplot[title={\textsc{KidneyStudy2}-2D},
    ymax=1.520,
    ]
      \draw[black, dotted, line width=0.6pt] (axis cs:3,1) -- (axis cs:8,1);
      \addplot[draw=none, fill=exp8color3, fill opacity=0.18, forget plot]
        coordinates {(3,0.5137) (4,0.5959) (5,0.6169) (6,0.5585) (7,0.4551) (8,0.4041) (8,0.2568) (7,0.2487) (6,0.3004) (5,0.2663) (4,0.2948) (3,0.2779)}
        --cycle;
      \addplot[color=exp8color3, line width=0.8pt, mark=*, mark size=1.2pt, mark options={solid}]
        coordinates {(3,0.4053) (4,0.4564) (5,0.4617) (6,0.4372) (7,0.3613) (8,0.3446)};
      \addplot[draw=none, fill=exp8color4, fill opacity=0.18, forget plot]
        coordinates {(3,1.1847) (4,1.3173) (5,1.3935) (6,1.4073) (7,1.3784) (8,1.3776) (8,1.2794) (7,1.2710) (6,1.2835) (5,1.2336) (4,1.1878) (3,1.0963)}
        --cycle;
      \addplot[color=exp8color4, line width=0.8pt, mark=*, mark size=1.2pt, mark options={solid}]
        coordinates {(3,1.1502) (4,1.2618) (5,1.3272) (6,1.3572) (7,1.3360) (8,1.3431)};
      \addplot[draw=none, fill=exp8color1, fill opacity=0.18, forget plot]
        coordinates {(3,1.1801) (4,1.3124) (5,1.3873) (6,1.3796) (7,1.2958) (8,1.3506) (8,1.0946) (7,1.0829) (6,1.1159) (5,1.0692) (4,1.0808) (3,1.0631)}
        --cycle;
      \addplot[color=exp8color1, line width=0.8pt, mark=*, mark size=1.2pt, mark options={solid}]
        coordinates {(3,1.1354) (4,1.2022) (5,1.2346) (6,1.2619) (7,1.2052) (8,1.2394)};
      \addplot[draw=none, fill=exp8color0, fill opacity=0.18, forget plot]
        coordinates {(3,0.5362) (4,0.6107) (5,0.6500) (6,0.5543) (7,0.4047) (8,0.4245) (8,0.1246) (7,0.1090) (6,0.1796) (5,0.1219) (4,0.1634) (3,0.1845)}
        --cycle;
      \addplot[color=exp8color0, line width=0.8pt, mark=*, mark size=1.2pt, mark options={solid}]
        coordinates {(3,0.3874) (4,0.4005) (5,0.3891) (6,0.3823) (7,0.2761) (8,0.2905)};
      \addplot[draw=none, fill=exp8color2, fill opacity=0.18, forget plot]
        coordinates {(3,0.3084) (4,0.4213) (5,0.4835) (6,0.4277) (7,0.3800) (8,0.2582) (8,0.0759) (7,0.0905) (6,0.1433) (5,0.1036) (4,0.1252) (3,0.0958)}
        --cycle;
      \addplot[color=exp8color2, line width=0.8pt, mark=*, mark size=1.2pt, mark options={solid}]
        coordinates {(3,0.2251) (4,0.2900) (5,0.3175) (6,0.2864) (7,0.2448) (8,0.2363)};
  \end{groupplot}
\end{tikzpicture}

%% file: Figures_final/exp10_full_all.tex
% Auto-generated by exp10_propind_propcol_vs_m.py.
% Required preamble in the host document:
%   \usepackage{tikz}
%   \usepackage{pgfplots}
%   \usepgfplotslibrary{groupplots}
%   \pgfplotsset{compat=1.18}
%
\definecolor{exp10color0}{rgb}{0.1216,0.4667,0.7059}
\definecolor{exp10color1}{rgb}{1.0000,0.4980,0.0549}
\definecolor{exp10color2}{rgb}{0.1725,0.6275,0.1725}
\definecolor{exp10color3}{rgb}{0.8392,0.1529,0.1569}
\definecolor{exp10color4}{rgb}{0.5804,0.4039,0.7412}

\begin{tikzpicture}
  \begin{groupplot}[
    group style={
      group size=5 by 1,
      horizontal sep=22pt,
    },
    width=4.4cm,
    height=4.2cm,
    xmode=log,
    log basis x=10,
    xtick={5,10,20,50,100,200},
    xticklabels={5,10,20,50,100,200},
    xlabel={$m$},
    title style={font=\small,yshift=-2pt},
    label style={font=\small},
    tick label style={font=\footnotesize},
    minor y tick num=4,
    grid=both,
    minor grid style={dotted, gray!20},
    major grid style={dotted, gray!40},
  ]
    \nextgroupplot[title={\textsc{FoodRescue}},
    ylabel={$\pind$ / $\pcol$},
    ]
      \addplot[color=exp10color0, line width=0.8pt, mark=*, mark size=1.2pt, mark options={solid}]
        coordinates {(5,10.9507) (10,11.0097) (20,11.0590) (50,10.9871) (100,11.0473) (200,11.0323)};
      \addplot[color=exp10color0, line width=0.8pt, dashed, mark=square*, mark size=1.0pt, mark options={solid}]
        coordinates {(5,7.7795) (10,9.7246) (20,10.5970) (50,10.9097) (100,11.0371) (200,11.0319)};
      \addplot[color=exp10color1, line width=0.8pt, mark=*, mark size=1.2pt, mark options={solid}]
        coordinates {(5,11.0590) (10,11.1234) (20,11.1607) (50,11.0899) (100,11.1501) (200,11.1354)};
      \addplot[color=exp10color1, line width=0.8pt, dashed, mark=square*, mark size=1.0pt, mark options={solid}]
        coordinates {(5,7.7938) (10,9.7723) (20,10.6490) (50,10.9949) (100,11.1351) (200,11.1348)};
      \addplot[color=exp10color2, line width=0.8pt, mark=*, mark size=1.2pt, mark options={solid}]
        coordinates {(5,10.7008) (10,10.7339) (20,10.8007) (50,10.7227) (100,10.7837) (200,10.7686)};
      \addplot[color=exp10color2, line width=0.8pt, dashed, mark=square*, mark size=1.0pt, mark options={solid}]
        coordinates {(5,7.7007) (10,9.6045) (20,10.4297) (50,10.6763) (100,10.7797) (200,10.7682)};
      \addplot[color=exp10color3, line width=0.8pt, mark=*, mark size=1.2pt, mark options={solid}]
        coordinates {(5,10.9830) (10,11.0130) (20,11.0477) (50,10.9875) (100,11.0436) (200,11.0311)};
      \addplot[color=exp10color3, line width=0.8pt, dashed, mark=square*, mark size=1.0pt, mark options={solid}]
        coordinates {(5,7.9667) (10,9.8764) (20,10.6725) (50,10.9349) (100,11.0382) (200,11.0308)};
      \addplot[color=exp10color4, line width=0.8pt, mark=*, mark size=1.2pt, mark options={solid}]
        coordinates {(5,11.1179) (10,11.1916) (20,11.2126) (50,11.1651) (100,11.2098) (200,11.1983)};
      \addplot[color=exp10color4, line width=0.8pt, dashed, mark=square*, mark size=1.0pt, mark options={solid}]
        coordinates {(5,8.3581) (10,10.1395) (20,10.8725) (50,11.1163) (100,11.2051) (200,11.1982)};
    \nextgroupplot[title={\textsc{Kidney}}]
      \addplot[color=exp10color0, line width=0.8pt, mark=*, mark size=1.2pt, mark options={solid}]
        coordinates {(5,132.3504) (10,134.5724) (20,133.4678) (50,133.7418) (100,133.5355) (200,133.3083)};
      \addplot[color=exp10color0, line width=0.8pt, dashed, mark=square*, mark size=1.0pt, mark options={solid}]
        coordinates {(5,60.8626) (10,113.4252) (20,128.5953) (50,133.5571) (100,133.5307) (200,133.3083)};
      \addplot[color=exp10color1, line width=0.8pt, mark=*, mark size=1.2pt, mark options={solid}]
        coordinates {(5,133.5220) (10,135.6318) (20,134.4459) (50,134.6735) (100,134.5169) (200,134.2837)};
      \addplot[color=exp10color1, line width=0.8pt, dashed, mark=square*, mark size=1.0pt, mark options={solid}]
        coordinates {(5,61.2868) (10,113.9953) (20,129.3980) (50,134.4660) (100,134.5119) (200,134.2837)};
      \addplot[color=exp10color2, line width=0.8pt, mark=*, mark size=1.2pt, mark options={solid}]
        coordinates {(5,129.9870) (10,132.6108) (20,131.7029) (50,131.8318) (100,131.6983) (200,131.4048)};
      \addplot[color=exp10color2, line width=0.8pt, dashed, mark=square*, mark size=1.0pt, mark options={solid}]
        coordinates {(5,59.9132) (10,111.6925) (20,126.9591) (50,131.6540) (100,131.6937) (200,131.4048)};
      \addplot[color=exp10color3, line width=0.8pt, mark=*, mark size=1.2pt, mark options={solid}]
        coordinates {(5,131.0576) (10,134.2986) (20,133.2902) (50,133.6658) (100,133.5012) (200,133.2742)};
      \addplot[color=exp10color3, line width=0.8pt, dashed, mark=square*, mark size=1.0pt, mark options={solid}]
        coordinates {(5,62.2362) (10,114.1569) (20,128.8303) (50,133.5070) (100,133.4970) (200,133.2742)};
      \addplot[color=exp10color4, line width=0.8pt, mark=*, mark size=1.2pt, mark options={solid}]
        coordinates {(5,132.6938) (10,136.0941) (20,135.1476) (50,135.3275) (100,135.2209) (200,135.0563)};
      \addplot[color=exp10color4, line width=0.8pt, dashed, mark=square*, mark size=1.0pt, mark options={solid}]
        coordinates {(5,65.0440) (10,116.7066) (20,131.0363) (50,135.2043) (100,135.2186) (200,135.0563)};
    \nextgroupplot[title={\textsc{MMachine}}]
      \addplot[color=exp10color0, line width=0.8pt, mark=*, mark size=1.2pt, mark options={solid}]
        coordinates {(5,110.9289) (10,110.7036) (20,110.6992) (50,110.9174) (100,110.8490) (200,110.7772)};
      \addplot[color=exp10color0, line width=0.8pt, dashed, mark=square*, mark size=1.0pt, mark options={solid}]
        coordinates {(5,82.7343) (10,98.0905) (20,104.8749) (50,108.5506) (100,109.5070) (200,110.0835)};
      \addplot[color=exp10color1, line width=0.8pt, mark=*, mark size=1.2pt, mark options={solid}]
        coordinates {(5,110.9289) (10,110.7280) (20,110.7075) (50,110.9268) (100,110.8588) (200,110.7866)};
      \addplot[color=exp10color1, line width=0.8pt, dashed, mark=square*, mark size=1.0pt, mark options={solid}]
        coordinates {(5,82.7206) (10,98.0890) (20,104.8771) (50,108.5559) (100,109.5155) (200,110.0922)};
      \addplot[color=exp10color2, line width=0.8pt, mark=*, mark size=1.2pt, mark options={solid}]
        coordinates {(5,110.9152) (10,110.7599) (20,110.7887) (50,111.0178) (100,110.9509) (200,110.8757)};
      \addplot[color=exp10color2, line width=0.8pt, dashed, mark=square*, mark size=1.0pt, mark options={solid}]
        coordinates {(5,82.7001) (10,98.0250) (20,104.8010) (50,108.4941) (100,109.4375) (200,110.0225)};
      \addplot[color=exp10color3, line width=0.8pt, mark=*, mark size=1.2pt, mark options={solid}]
        coordinates {(5,110.9426) (10,110.7478) (20,110.7440) (50,110.9222) (100,110.8531) (200,110.7795)};
      \addplot[color=exp10color3, line width=0.8pt, dashed, mark=square*, mark size=1.0pt, mark options={solid}]
        coordinates {(5,82.7001) (10,98.1559) (20,104.9323) (50,108.5802) (100,109.5285) (200,110.0951)};
      \addplot[color=exp10color4, line width=0.8pt, mark=*, mark size=1.2pt, mark options={solid}]
        coordinates {(5,110.9974) (10,110.7812) (20,110.7468) (50,110.9413) (100,110.8617) (200,110.7852)};
      \addplot[color=exp10color4, line width=0.8pt, dashed, mark=square*, mark size=1.0pt, mark options={solid}]
        coordinates {(5,82.7891) (10,98.2077) (20,104.9730) (50,108.6191) (100,109.5504) (200,110.1177)};
    \nextgroupplot[title={\textsc{KidneyStudy1}}]
      \addplot[color=exp10color0, line width=0.8pt, mark=*, mark size=1.2pt, mark options={solid}]
        coordinates {(5,12.1712) (10,12.1981) (20,12.1653) (50,12.2073) (100,12.2051) (200,12.1938)};
      \addplot[color=exp10color0, line width=0.8pt, dashed, mark=square*, mark size=1.0pt, mark options={solid}]
        coordinates {(5,9.5226) (10,11.0987) (20,11.7708) (50,12.1080) (100,12.1796) (200,12.1897)};
      \addplot[color=exp10color1, line width=0.8pt, mark=*, mark size=1.2pt, mark options={solid}]
        coordinates {(5,12.2034) (10,12.2304) (20,12.1991) (50,12.2381) (100,12.2364) (200,12.2253)};
      \addplot[color=exp10color1, line width=0.8pt, dashed, mark=square*, mark size=1.0pt, mark options={solid}]
        coordinates {(5,9.5285) (10,11.1067) (20,11.7846) (50,12.1271) (100,12.2049) (200,12.2194)};
      \addplot[color=exp10color2, line width=0.8pt, mark=*, mark size=1.2pt, mark options={solid}]
        coordinates {(5,11.8779) (10,11.9204) (20,11.8855) (50,11.9330) (100,11.9297) (200,11.9180)};
      \addplot[color=exp10color2, line width=0.8pt, dashed, mark=square*, mark size=1.0pt, mark options={solid}]
        coordinates {(5,9.4681) (10,11.0237) (20,11.6347) (50,11.9014) (100,11.9274) (200,11.9180)};
      \addplot[color=exp10color3, line width=0.8pt, mark=*, mark size=1.2pt, mark options={solid}]
        coordinates {(5,12.1235) (10,12.1884) (20,12.1539) (50,12.2025) (100,12.2010) (200,12.1929)};
      \addplot[color=exp10color3, line width=0.8pt, dashed, mark=square*, mark size=1.0pt, mark options={solid}]
        coordinates {(5,9.8133) (10,11.2965) (20,11.8746) (50,12.1469) (100,12.1907) (200,12.1921)};
      \addplot[color=exp10color4, line width=0.8pt, mark=*, mark size=1.2pt, mark options={solid}]
        coordinates {(5,12.1890) (10,12.2479) (20,12.2298) (50,12.2786) (100,12.2767) (200,12.2682)};
      \addplot[color=exp10color4, line width=0.8pt, dashed, mark=square*, mark size=1.0pt, mark options={solid}]
        coordinates {(5,9.9654) (10,11.3985) (20,11.9507) (50,12.2143) (100,12.2618) (200,12.2660)};
    \nextgroupplot[title={\textsc{KidneyStudy2}}]
      \addplot[color=exp10color0, line width=0.8pt, mark=*, mark size=1.2pt, mark options={solid}]
        coordinates {(5,28.5434) (10,28.4407) (20,28.6421) (50,28.5402) (100,28.5950) (200,28.5715)};
      \addplot[color=exp10color0, line width=0.8pt, dashed, mark=square*, mark size=1.0pt, mark options={solid}]
        coordinates {(5,22.3815) (10,26.0007) (20,27.6291) (50,28.2677) (100,28.5163) (200,28.5523)};
      \addplot[color=exp10color1, line width=0.8pt, mark=*, mark size=1.2pt, mark options={solid}]
        coordinates {(5,28.6917) (10,28.5778) (20,28.7891) (50,28.6829) (100,28.7396) (200,28.7149)};
      \addplot[color=exp10color1, line width=0.8pt, dashed, mark=square*, mark size=1.0pt, mark options={solid}]
        coordinates {(5,22.4396) (10,26.0781) (20,27.7462) (50,28.3928) (100,28.6527) (200,28.6927)};
      \addplot[color=exp10color2, line width=0.8pt, mark=*, mark size=1.2pt, mark options={solid}]
        coordinates {(5,27.4877) (10,27.3528) (20,27.5080) (50,27.4106) (100,27.4684) (200,27.4420)};
      \addplot[color=exp10color2, line width=0.8pt, dashed, mark=square*, mark size=1.0pt, mark options={solid}]
        coordinates {(5,21.8806) (10,25.2147) (20,26.6586) (50,27.2196) (100,27.4268) (200,27.4365)};
      \addplot[color=exp10color3, line width=0.8pt, mark=*, mark size=1.2pt, mark options={solid}]
        coordinates {(5,28.2940) (10,28.3298) (20,28.5729) (50,28.5305) (100,28.5869) (200,28.5685)};
      \addplot[color=exp10color3, line width=0.8pt, dashed, mark=square*, mark size=1.0pt, mark options={solid}]
        coordinates {(5,22.5729) (10,26.0728) (20,27.6385) (50,28.2789) (100,28.5158) (200,28.5520)};
      \addplot[color=exp10color4, line width=0.8pt, mark=*, mark size=1.2pt, mark options={solid}]
        coordinates {(5,28.5606) (10,28.6112) (20,28.8780) (50,28.8702) (100,28.9295) (200,28.9205)};
      \addplot[color=exp10color4, line width=0.8pt, dashed, mark=square*, mark size=1.0pt, mark options={solid}]
        coordinates {(5,22.8631) (10,26.3700) (20,27.9302) (50,28.6006) (100,28.8483) (200,28.8989)};
  \end{groupplot}
\end{tikzpicture}

%% file: Figures_final/overview_table_minimal_randw.tex
\begin{table*}[h]
\centering
\footnotesize
\setlength{\tabcolsep}{4pt}
\setlength{\extrarowheight}{2pt}
\caption{Long-run IP ($\pind$) and per-batch IP ($\pcol$) for each rule on each dataset and its 2D variant.  Cell shading is proportional to $\log_{10}(\text{value})$: darker means more disproportional.  Voter weights: Dirichlet(1) weights.}
\label{tab:overview-minimal-randw}
\begin{tabular}{l@{\hspace{10pt}}>{\centering\arraybackslash}p{0.55cm}@{\hspace{10pt}}>{\centering\arraybackslash}p{0.55cm}@{\hspace{10pt}}>{\centering\arraybackslash}p{0.55cm}@{\hspace{10pt}}>{\centering\arraybackslash}p{0.55cm}@{\hspace{10pt}}>{\centering\arraybackslash}p{0.55cm}@{\hspace{10pt}}>{\centering\arraybackslash}p{0.55cm}@{\hspace{10pt}}>{\centering\arraybackslash}p{0.55cm}@{\hspace{10pt}}>{\centering\arraybackslash}p{0.55cm}@{\hspace{10pt}}>{\centering\arraybackslash}p{0.55cm}@{\hspace{10pt}}>{\centering\arraybackslash}p{0.55cm}}
\toprule
 & \multicolumn{2}{c}{arithmetic} & \multicolumn{2}{c}{angular} & \multicolumn{2}{c}{geometric} & \multicolumn{2}{c}{Borda} & \multicolumn{2}{c}{PSB} \\
\cmidrule(lr){2-3} \cmidrule(lr){4-5} \cmidrule(lr){6-7} \cmidrule(lr){8-9} \cmidrule(lr){10-11}
dataset & \textsf{long} & \textsf{batch} & \textsf{long} & \textsf{batch} & \textsf{long} & \textsf{batch} & \textsf{long} & \textsf{batch} & \textsf{long} & \textsf{batch} \\
\midrule
\textsc{FoodRescue} & \cellcolor{black!16}3.0 & \cellcolor{black!16}3.0 & \cellcolor{black!16}3.0 & \cellcolor{black!16}3.0 & \cellcolor{black!15}2.8 & \cellcolor{black!15}2.8 & \cellcolor{black!16}3.0 & \cellcolor{black!16}3.0 & \cellcolor{black!16}3.0 & \cellcolor{black!16}3.0 \\
\textsc{Kidney} & \cellcolor{black!28}40 & \cellcolor{black!28}40 & \cellcolor{black!28}40 & \cellcolor{black!28}40 & \cellcolor{black!28}40 & \cellcolor{black!28}40 & \cellcolor{black!28}40 & \cellcolor{black!28}40 & \cellcolor{black!28}40 & \cellcolor{black!28}40 \\
\textsc{MMachine} & \cellcolor{black!25}22 & \cellcolor{black!25}22 & \cellcolor{black!25}22 & \cellcolor{black!25}22 & \cellcolor{black!25}22 & \cellcolor{black!25}22 & \cellcolor{black!25}22 & \cellcolor{black!25}22 & \cellcolor{black!25}22 & \cellcolor{black!25}22 \\
\textsc{KidneyStudy1} & \cellcolor{black!16}3.5 & \cellcolor{black!16}3.5 & \cellcolor{black!16}3.5 & \cellcolor{black!16}3.5 & \cellcolor{black!16}3.4 & \cellcolor{black!16}3.4 & \cellcolor{black!16}3.5 & \cellcolor{black!16}3.5 & \cellcolor{black!16}3.5 & \cellcolor{black!16}3.5 \\
\textsc{KidneyStudy2} & \cellcolor{black!20}7.1 & \cellcolor{black!20}7.1 & \cellcolor{black!20}7.1 & \cellcolor{black!20}7.1 & \cellcolor{black!20}7.3 & \cellcolor{black!20}7.3 & \cellcolor{black!20}7.2 & \cellcolor{black!20}7.2 & \cellcolor{black!20}7.2 & \cellcolor{black!20}7.2 \\
\midrule
\textsc{FoodRescue}-2D & \cellcolor{black!16}3.3 & \cellcolor{black!15}2.9 & \cellcolor{black!17}3.9 & \cellcolor{black!17}3.7 & \cellcolor{black!15}2.6 & \cellcolor{black!14}2.2 & \cellcolor{black!16}3.3 & \cellcolor{black!15}2.9 & \cellcolor{black!17}3.8 & \cellcolor{black!17}3.8 \\
\textsc{Kidney}-2D & \cellcolor{black!8}0.5 & 0.02 & \cellcolor{black!16}3.4 & \cellcolor{black!9}0.8 & \cellcolor{black!17}4.1 & \cellcolor{black!12}1.4 & \cellcolor{black!18}4.9 & \cellcolor{black!15}2.8 & \cellcolor{black!23}14 & \cellcolor{black!23}14 \\
\textsc{MMachine}-2D & \cellcolor{black!15}2.4 & \cellcolor{black!9}0.7 & \cellcolor{black!12}1.5 & \cellcolor{black!11}1.1 & \cellcolor{black!4}0.3 & \cellcolor{black!2}0.2 & \cellcolor{black!15}2.6 & \cellcolor{black!14}1.9 & \cellcolor{black!21}10 & \cellcolor{black!21}9.6 \\
\textsc{KidneyStudy1}-2D & \cellcolor{black!17}3.8 & \cellcolor{black!17}3.7 & \cellcolor{black!17}3.8 & \cellcolor{black!17}3.8 & \cellcolor{black!17}3.8 & \cellcolor{black!17}3.8 & \cellcolor{black!17}3.8 & \cellcolor{black!17}3.7 & \cellcolor{black!17}3.8 & \cellcolor{black!17}3.8 \\
\textsc{KidneyStudy2}-2D & \cellcolor{black!19}6.8 & \cellcolor{black!19}6.2 & \cellcolor{black!19}6.7 & \cellcolor{black!19}6.3 & \cellcolor{black!19}6.8 & \cellcolor{black!19}6.2 & \cellcolor{black!19}6.8 & \cellcolor{black!19}6.5 & \cellcolor{black!19}6.8 & \cellcolor{black!19}6.6 \\
\bottomrule
\end{tabular}
\end{table*}

%% file: Figures_final/exp10_full_all_randw.tex
% Auto-generated by exp10_propind_propcol_vs_m.py.
% Required preamble in the host document:
%   \usepackage{tikz}
%   \usepackage{pgfplots}
%   \usepgfplotslibrary{groupplots}
%   \pgfplotsset{compat=1.18}
%
\definecolor{exp10color0}{rgb}{0.1216,0.4667,0.7059}
\definecolor{exp10color1}{rgb}{1.0000,0.4980,0.0549}
\definecolor{exp10color2}{rgb}{0.1725,0.6275,0.1725}
\definecolor{exp10color3}{rgb}{0.8392,0.1529,0.1569}
\definecolor{exp10color4}{rgb}{0.5804,0.4039,0.7412}

\begin{tikzpicture}
  \begin{groupplot}[
    group style={
      group size=5 by 1,
      horizontal sep=22pt,
    },
    width=4.4cm,
    height=4.2cm,
    xmode=log,
    log basis x=10,
    xtick={5,10,20,50,100,200},
    xticklabels={5,10,20,50,100,200},
    xlabel={$m$},
    title style={font=\small,yshift=-2pt},
    label style={font=\small},
    tick label style={font=\footnotesize},
    minor y tick num=4,
    grid=both,
    minor grid style={dotted, gray!20},
    major grid style={dotted, gray!40},
  ]
    \nextgroupplot[title={\textsc{FoodRescue}},
    ylabel={$\pind$ / $\pcol$},
    ]
      \addplot[color=exp10color0, line width=0.8pt, mark=*, mark size=1.2pt, mark options={solid}]
        coordinates {(5,2.9751) (10,2.9612) (20,2.9877) (50,2.9854) (100,2.9802) (200,2.9841)};
      \addplot[color=exp10color0, line width=0.8pt, dashed, mark=square*, mark size=1.0pt, mark options={solid}]
        coordinates {(5,2.8582) (10,2.9547) (20,2.9877) (50,2.9854) (100,2.9802) (200,2.9841)};
      \addplot[color=exp10color1, line width=0.8pt, mark=*, mark size=1.2pt, mark options={solid}]
        coordinates {(5,2.9956) (10,2.9810) (20,3.0075) (50,3.0048) (100,2.9995) (200,3.0033)};
      \addplot[color=exp10color1, line width=0.8pt, dashed, mark=square*, mark size=1.0pt, mark options={solid}]
        coordinates {(5,2.8731) (10,2.9748) (20,3.0075) (50,3.0048) (100,2.9995) (200,3.0033)};
      \addplot[color=exp10color2, line width=0.8pt, mark=*, mark size=1.2pt, mark options={solid}]
        coordinates {(5,2.8247) (10,2.8054) (20,2.8345) (50,2.8329) (100,2.8273) (200,2.8314)};
      \addplot[color=exp10color2, line width=0.8pt, dashed, mark=square*, mark size=1.0pt, mark options={solid}]
        coordinates {(5,2.7292) (10,2.8004) (20,2.8345) (50,2.8329) (100,2.8273) (200,2.8314)};
      \addplot[color=exp10color3, line width=0.8pt, mark=*, mark size=1.2pt, mark options={solid}]
        coordinates {(5,2.9888) (10,2.9726) (20,2.9902) (50,2.9870) (100,2.9807) (200,2.9831)};
      \addplot[color=exp10color3, line width=0.8pt, dashed, mark=square*, mark size=1.0pt, mark options={solid}]
        coordinates {(5,2.8919) (10,2.9682) (20,2.9902) (50,2.9870) (100,2.9807) (200,2.9831)};
      \addplot[color=exp10color4, line width=0.8pt, mark=*, mark size=1.2pt, mark options={solid}]
        coordinates {(5,3.0251) (10,3.0138) (20,3.0309) (50,3.0278) (100,3.0227) (200,3.0226)};
      \addplot[color=exp10color4, line width=0.8pt, dashed, mark=square*, mark size=1.0pt, mark options={solid}]
        coordinates {(5,2.9913) (10,3.0134) (20,3.0309) (50,3.0278) (100,3.0227) (200,3.0226)};
    \nextgroupplot[title={\textsc{Kidney}}]
      \addplot[color=exp10color0, line width=0.8pt, mark=*, mark size=1.2pt, mark options={solid}]
        coordinates {(5,33.1245) (10,32.6573) (20,33.0211) (50,32.9360) (100,32.9574) (200,32.9557)};
      \addplot[color=exp10color0, line width=0.8pt, dashed, mark=square*, mark size=1.0pt, mark options={solid}]
        coordinates {(5,23.9067) (10,32.2887) (20,33.0093) (50,32.9360) (100,32.9574) (200,32.9557)};
      \addplot[color=exp10color1, line width=0.8pt, mark=*, mark size=1.2pt, mark options={solid}]
        coordinates {(5,33.1978) (10,32.7656) (20,33.1210) (50,33.0229) (100,33.0498) (200,33.0488)};
      \addplot[color=exp10color1, line width=0.8pt, dashed, mark=square*, mark size=1.0pt, mark options={solid}]
        coordinates {(5,24.0373) (10,32.4036) (20,33.1100) (50,33.0229) (100,33.0498) (200,33.0488)};
      \addplot[color=exp10color2, line width=0.8pt, mark=*, mark size=1.2pt, mark options={solid}]
        coordinates {(5,32.6775) (10,32.1727) (20,32.5282) (50,32.4337) (100,32.4532) (200,32.4534)};
      \addplot[color=exp10color2, line width=0.8pt, dashed, mark=square*, mark size=1.0pt, mark options={solid}]
        coordinates {(5,23.4413) (10,31.7759) (20,32.5133) (50,32.4337) (100,32.4532) (200,32.4534)};
      \addplot[color=exp10color3, line width=0.8pt, mark=*, mark size=1.2pt, mark options={solid}]
        coordinates {(5,33.0914) (10,32.6489) (20,33.0062) (50,32.9268) (100,32.9536) (200,32.9512)};
      \addplot[color=exp10color3, line width=0.8pt, dashed, mark=square*, mark size=1.0pt, mark options={solid}]
        coordinates {(5,24.3825) (10,32.3377) (20,32.9942) (50,32.9268) (100,32.9536) (200,32.9512)};
      \addplot[color=exp10color4, line width=0.8pt, mark=*, mark size=1.2pt, mark options={solid}]
        coordinates {(5,33.2924) (10,32.8633) (20,33.1901) (50,33.1132) (100,33.1471) (200,33.1405)};
      \addplot[color=exp10color4, line width=0.8pt, dashed, mark=square*, mark size=1.0pt, mark options={solid}]
        coordinates {(5,25.5023) (10,32.6543) (20,33.1831) (50,33.1132) (100,33.1471) (200,33.1405)};
    \nextgroupplot[title={\textsc{MMachine}}]
      \addplot[color=exp10color0, line width=0.8pt, mark=*, mark size=1.2pt, mark options={solid}]
        coordinates {(5,23.5462) (10,23.6223) (20,23.6029) (50,23.5751) (100,23.5855) (200,23.5889)};
      \addplot[color=exp10color0, line width=0.8pt, dashed, mark=square*, mark size=1.0pt, mark options={solid}]
        coordinates {(5,22.7408) (10,23.4825) (20,23.5791) (50,23.5749) (100,23.5855) (200,23.5889)};
      \addplot[color=exp10color1, line width=0.8pt, mark=*, mark size=1.2pt, mark options={solid}]
        coordinates {(5,23.5437) (10,23.6198) (20,23.6018) (50,23.5753) (100,23.5855) (200,23.5890)};
      \addplot[color=exp10color1, line width=0.8pt, dashed, mark=square*, mark size=1.0pt, mark options={solid}]
        coordinates {(5,22.7409) (10,23.4794) (20,23.5778) (50,23.5751) (100,23.5855) (200,23.5890)};
      \addplot[color=exp10color2, line width=0.8pt, mark=*, mark size=1.2pt, mark options={solid}]
        coordinates {(5,23.4529) (10,23.5417) (20,23.5192) (50,23.4942) (100,23.5038) (200,23.5081)};
      \addplot[color=exp10color2, line width=0.8pt, dashed, mark=square*, mark size=1.0pt, mark options={solid}]
        coordinates {(5,22.6961) (10,23.4063) (20,23.4985) (50,23.4941) (100,23.5038) (200,23.5081)};
      \addplot[color=exp10color3, line width=0.8pt, mark=*, mark size=1.2pt, mark options={solid}]
        coordinates {(5,23.5079) (10,23.5985) (20,23.5828) (50,23.5662) (100,23.5814) (200,23.5864)};
      \addplot[color=exp10color3, line width=0.8pt, dashed, mark=square*, mark size=1.0pt, mark options={solid}]
        coordinates {(5,22.7609) (10,23.4651) (20,23.5630) (50,23.5661) (100,23.5814) (200,23.5864)};
      \addplot[color=exp10color4, line width=0.8pt, mark=*, mark size=1.2pt, mark options={solid}]
        coordinates {(5,23.4926) (10,23.5922) (20,23.5797) (50,23.5647) (100,23.5831) (200,23.5875)};
      \addplot[color=exp10color4, line width=0.8pt, dashed, mark=square*, mark size=1.0pt, mark options={solid}]
        coordinates {(5,22.7795) (10,23.4575) (20,23.5591) (50,23.5645) (100,23.5831) (200,23.5875)};
    \nextgroupplot[title={\textsc{KidneyStudy1}}]
      \addplot[color=exp10color0, line width=0.8pt, mark=*, mark size=1.2pt, mark options={solid}]
        coordinates {(5,3.4664) (10,3.4539) (20,3.4472) (50,3.4592) (100,3.4581) (200,3.4543)};
      \addplot[color=exp10color0, line width=0.8pt, dashed, mark=square*, mark size=1.0pt, mark options={solid}]
        coordinates {(5,3.3730) (10,3.4504) (20,3.4470) (50,3.4592) (100,3.4581) (200,3.4543)};
      \addplot[color=exp10color1, line width=0.8pt, mark=*, mark size=1.2pt, mark options={solid}]
        coordinates {(5,3.4597) (10,3.4461) (20,3.4386) (50,3.4512) (100,3.4498) (200,3.4460)};
      \addplot[color=exp10color1, line width=0.8pt, dashed, mark=square*, mark size=1.0pt, mark options={solid}]
        coordinates {(5,3.3684) (10,3.4430) (20,3.4384) (50,3.4512) (100,3.4498) (200,3.4460)};
      \addplot[color=exp10color2, line width=0.8pt, mark=*, mark size=1.2pt, mark options={solid}]
        coordinates {(5,3.6414) (10,3.6297) (20,3.6224) (50,3.6322) (100,3.6308) (200,3.6272)};
      \addplot[color=exp10color2, line width=0.8pt, dashed, mark=square*, mark size=1.0pt, mark options={solid}]
        coordinates {(5,3.4872) (10,3.6148) (20,3.6215) (50,3.6322) (100,3.6308) (200,3.6272)};
      \addplot[color=exp10color3, line width=0.8pt, mark=*, mark size=1.2pt, mark options={solid}]
        coordinates {(5,3.5294) (10,3.4922) (20,3.4732) (50,3.4668) (100,3.4625) (200,3.4569)};
      \addplot[color=exp10color3, line width=0.8pt, dashed, mark=square*, mark size=1.0pt, mark options={solid}]
        coordinates {(5,3.4700) (10,3.4915) (20,3.4731) (50,3.4668) (100,3.4625) (200,3.4569)};
      \addplot[color=exp10color4, line width=0.8pt, mark=*, mark size=1.2pt, mark options={solid}]
        coordinates {(5,3.5071) (10,3.4736) (20,3.4548) (50,3.4445) (100,3.4377) (200,3.4312)};
      \addplot[color=exp10color4, line width=0.8pt, dashed, mark=square*, mark size=1.0pt, mark options={solid}]
        coordinates {(5,3.4816) (10,3.4736) (20,3.4548) (50,3.4445) (100,3.4377) (200,3.4312)};
    \nextgroupplot[title={\textsc{KidneyStudy2}}]
      \addplot[color=exp10color0, line width=0.8pt, mark=*, mark size=1.2pt, mark options={solid}]
        coordinates {(5,7.3665) (10,7.3545) (20,7.3608) (50,7.3558) (100,7.3503) (200,7.3493)};
      \addplot[color=exp10color0, line width=0.8pt, dashed, mark=square*, mark size=1.0pt, mark options={solid}]
        coordinates {(5,7.0716) (10,7.3318) (20,7.3604) (50,7.3558) (100,7.3503) (200,7.3493)};
      \addplot[color=exp10color1, line width=0.8pt, mark=*, mark size=1.2pt, mark options={solid}]
        coordinates {(5,7.3686) (10,7.3616) (20,7.3675) (50,7.3613) (100,7.3557) (200,7.3545)};
      \addplot[color=exp10color1, line width=0.8pt, dashed, mark=square*, mark size=1.0pt, mark options={solid}]
        coordinates {(5,7.0779) (10,7.3397) (20,7.3672) (50,7.3613) (100,7.3557) (200,7.3545)};
      \addplot[color=exp10color2, line width=0.8pt, mark=*, mark size=1.2pt, mark options={solid}]
        coordinates {(5,7.2258) (10,7.2075) (20,7.2102) (50,7.2075) (100,7.2004) (200,7.2009)};
      \addplot[color=exp10color2, line width=0.8pt, dashed, mark=square*, mark size=1.0pt, mark options={solid}]
        coordinates {(5,6.9132) (10,7.1819) (20,7.2093) (50,7.2075) (100,7.2004) (200,7.2009)};
      \addplot[color=exp10color3, line width=0.8pt, mark=*, mark size=1.2pt, mark options={solid}]
        coordinates {(5,7.3349) (10,7.3334) (20,7.3536) (50,7.3534) (100,7.3477) (200,7.3493)};
      \addplot[color=exp10color3, line width=0.8pt, dashed, mark=square*, mark size=1.0pt, mark options={solid}]
        coordinates {(5,7.0700) (10,7.3148) (20,7.3534) (50,7.3534) (100,7.3477) (200,7.3493)};
      \addplot[color=exp10color4, line width=0.8pt, mark=*, mark size=1.2pt, mark options={solid}]
        coordinates {(5,7.3328) (10,7.3354) (20,7.3628) (50,7.3649) (100,7.3605) (200,7.3625)};
      \addplot[color=exp10color4, line width=0.8pt, dashed, mark=square*, mark size=1.0pt, mark options={solid}]
        coordinates {(5,7.1221) (10,7.3224) (20,7.3627) (50,7.3649) (100,7.3605) (200,7.3625)};
  \end{groupplot}
\end{tikzpicture}

%% file: Figures_final/exp_acg_real.tex
% Auto-generated by exp_acg_real.py.  Do not edit by hand.
% Required preamble in the host document:
%   \usepackage{tikz}
%   \usepackage{pgfplots}
%   \usepgfplotslibrary{groupplots}
%   \pgfplotsset{compat=1.18}
%
\definecolor{acgcolor0}{rgb}{0.1216,0.4667,0.7059}
\definecolor{acgcolor1}{rgb}{1.0000,0.4980,0.0549}
\definecolor{acgcolor2}{rgb}{0.1725,0.6275,0.1725}
\definecolor{acgcolor3}{rgb}{0.8392,0.1529,0.1569}
\definecolor{acgcolor4}{rgb}{0.5804,0.4039,0.7412}

\begin{tikzpicture}
  \begin{groupplot}[
    group style={
      group size=5 by 1,
      horizontal sep=22pt,
    },
    width=4.4cm,
    height=4.2cm,
    xmode=log,
    x dir=reverse,
    log basis x=10,
    xtick={1,0.7,0.5,0.3,0.2,0.1,0.05,0.02},
    xticklabels={1,0.7,0.5,0.3,0.2,0.1,0.05,0.02},
    xlabel={$\lambda$},
    ymode=log,
    log basis y=10,
    title style={font=\small,yshift=-2pt},
    label style={font=\small},
    tick label style={font=\footnotesize, rotate=45},
    grid=both,
    minor grid style={dotted, gray!20},
    major grid style={dotted, gray!40},
  ]
    \nextgroupplot[title={\textsc{FoodRescue}-2D}, ymin=0.1509, ymax=9.045, ylabel={$\pind$}]
      \draw[black, dotted, line width=0.6pt] (axis cs:0.02,1) -- (axis cs:1,1);
      \addplot[color=acgcolor0, solid, line width=0.8pt, mark=*, mark size=1.0pt, mark options={solid}]
        coordinates {(0.02,0.6912) (0.05,0.9914) (0.1,1.2620) (0.2,1.5626) (0.3,1.6739) (0.5,1.8322) (0.7,1.8669) (1,1.9139)};
      \addplot[color=acgcolor0, dashed, line width=0.8pt, mark=*, mark size=1.0pt, mark options={solid}]
        coordinates {(0.02,1.8740) (0.05,2.6725) (0.1,2.4622) (0.2,2.2114) (0.3,2.0949) (0.5,1.9678) (0.7,1.9342) (1,1.8975)};
      \addplot[color=acgcolor0, dotted, line width=0.8pt, mark=*, mark size=1.0pt, mark options={solid}]
        coordinates {(0.02,0.6555) (0.05,0.9755) (0.1,1.2401) (0.2,1.5487) (0.3,1.6935) (0.5,1.8324) (0.7,1.8472) (1,1.8850)};
      \addplot[color=acgcolor0, dashdotted, line width=0.8pt, mark=*, mark size=1.0pt, mark options={solid}]
        coordinates {(0.02,1.9133) (0.05,2.6144) (0.1,2.5067) (0.2,2.2454) (0.3,2.1054) (0.5,1.9695) (0.7,1.9315) (1,1.9002)};
      \addplot[color=acgcolor1, solid, line width=0.8pt, mark=*, mark size=1.0pt, mark options={solid}]
        coordinates {(0.02,3.9647) (0.05,4.1067) (0.1,3.6987) (0.2,3.2815) (0.3,3.1251) (0.5,2.9311) (0.7,2.8166) (1,2.7995)};
      \addplot[color=acgcolor1, dashed, line width=0.8pt, mark=*, mark size=1.0pt, mark options={solid}]
        coordinates {(0.02,0.9992) (0.05,1.4955) (0.1,1.9103) (0.2,2.3383) (0.3,2.5357) (0.5,2.7514) (0.7,2.7704) (1,2.7892)};
      \addplot[color=acgcolor1, dotted, line width=0.8pt, mark=*, mark size=1.0pt, mark options={solid}]
        coordinates {(0.02,3.9056) (0.05,4.1308) (0.1,3.7411) (0.2,3.3121) (0.3,3.0446) (0.5,2.8924) (0.7,2.8329) (1,2.7850)};
      \addplot[color=acgcolor1, dashdotted, line width=0.8pt, mark=*, mark size=1.0pt, mark options={solid}]
        coordinates {(0.02,1.0779) (0.05,1.5151) (0.1,1.8947) (0.2,2.3416) (0.3,2.5038) (0.5,2.6682) (0.7,2.8023) (1,2.8103)};
      \addplot[color=acgcolor2, solid, line width=0.8pt, mark=*, mark size=1.0pt, mark options={solid}]
        coordinates {(0.02,0.2155) (0.05,0.3403) (0.1,0.4467) (0.2,0.5599) (0.3,0.6156) (0.5,0.6874) (0.7,0.7212) (1,0.7224)};
      \addplot[color=acgcolor2, dashed, line width=0.8pt, mark=*, mark size=1.0pt, mark options={solid}]
        coordinates {(0.02,1.6101) (0.05,1.3984) (0.1,1.1514) (0.2,0.9468) (0.3,0.8539) (0.5,0.7830) (0.7,0.7440) (1,0.7165)};
      \addplot[color=acgcolor2, dotted, line width=0.8pt, mark=*, mark size=1.0pt, mark options={solid}]
        coordinates {(0.02,0.2179) (0.05,0.3162) (0.1,0.4226) (0.2,0.5605) (0.3,0.6238) (0.5,0.6905) (0.7,0.7313) (1,0.7167)};
      \addplot[color=acgcolor2, dashdotted, line width=0.8pt, mark=*, mark size=1.0pt, mark options={solid}]
        coordinates {(0.02,1.6036) (0.05,1.3602) (0.1,1.1413) (0.2,0.9373) (0.3,0.8561) (0.5,0.7807) (0.7,0.7494) (1,0.7372)};
      \addplot[color=acgcolor3, solid, line width=0.8pt, mark=*, mark size=1.0pt, mark options={solid}]
        coordinates {(0.02,1.6414) (0.05,1.7216) (0.1,1.7797) (0.2,1.8396) (0.3,1.8373) (0.5,1.8272) (0.7,1.7659) (1,1.7146)};
      \addplot[color=acgcolor3, dashed, line width=0.8pt, mark=*, mark size=1.0pt, mark options={solid}]
        coordinates {(0.02,2.0146) (0.05,1.9897) (0.1,1.9099) (0.2,1.7885) (0.3,1.7613) (0.5,1.7047) (0.7,1.7212) (1,1.7073)};
      \addplot[color=acgcolor3, dotted, line width=0.8pt, mark=*, mark size=1.0pt, mark options={solid}]
        coordinates {(0.02,0.3694) (0.05,0.5907) (0.1,0.8159) (0.2,1.1351) (0.3,1.2990) (0.5,1.4809) (0.7,1.5968) (1,1.7022)};
      \addplot[color=acgcolor3, dashdotted, line width=0.8pt, mark=*, mark size=1.0pt, mark options={solid}]
        coordinates {(0.02,2.1493) (0.05,2.3353) (0.1,2.1977) (0.2,1.9853) (0.3,1.8681) (0.5,1.7782) (0.7,1.7474) (1,1.7087)};
      \addplot[color=acgcolor4, solid, line width=0.8pt, mark=*, mark size=1.0pt, mark options={solid}]
        coordinates {(0.02,6.0300) (0.05,5.9434) (0.1,5.8693) (0.2,5.8186) (0.3,5.7762) (0.5,5.7315) (0.7,5.7027) (1,5.6947)};
      \addplot[color=acgcolor4, dashed, line width=0.8pt, mark=*, mark size=1.0pt, mark options={solid}]
        coordinates {(0.02,4.8169) (0.05,5.1729) (0.1,5.4547) (0.2,5.6208) (0.3,5.6567) (0.5,5.6464) (0.7,5.6783) (1,5.6907)};
      \addplot[color=acgcolor4, dotted, line width=0.8pt, mark=*, mark size=1.0pt, mark options={solid}]
        coordinates {(0.02,5.8087) (0.05,5.7927) (0.1,5.7517) (0.2,5.7266) (0.3,5.7106) (0.5,5.7146) (0.7,5.6690) (1,5.7021)};
      \addplot[color=acgcolor4, dashdotted, line width=0.8pt, mark=*, mark size=1.0pt, mark options={solid}]
        coordinates {(0.02,4.8484) (0.05,5.2195) (0.1,5.5890) (0.2,5.6823) (0.3,5.6973) (0.5,5.6905) (0.7,5.6920) (1,5.6985)};
    \nextgroupplot[title={\textsc{Kidney}-2D}, ymin=0.0880, ymax=129.455]
      \draw[black, dotted, line width=0.6pt] (axis cs:0.02,1) -- (axis cs:1,1);
      \addplot[color=acgcolor0, solid, line width=0.8pt, mark=*, mark size=1.0pt, mark options={solid}]
        coordinates {(0.02,0.1257) (0.05,0.1885) (0.1,0.2873) (0.2,0.2963) (0.3,0.3142) (0.5,0.2828) (0.7,0.3905) (1,0.3456)};
      \addplot[color=acgcolor0, dashed, line width=0.8pt, mark=*, mark size=1.0pt, mark options={solid}]
        coordinates {(0.02,0.1975) (0.05,0.2244) (0.1,0.3501) (0.2,0.2828) (0.3,0.2873) (0.5,0.3456) (0.7,0.3322) (1,0.3501)};
      \addplot[color=acgcolor0, dotted, line width=0.8pt, mark=*, mark size=1.0pt, mark options={solid}]
        coordinates {(0.02,0.1571) (0.05,0.1706) (0.1,0.2738) (0.2,0.2559) (0.3,0.3008) (0.5,0.3097) (0.7,0.3097) (1,0.2963)};
      \addplot[color=acgcolor0, dashdotted, line width=0.8pt, mark=*, mark size=1.0pt, mark options={solid}]
        coordinates {(0.02,0.1975) (0.05,0.2334) (0.1,0.3546) (0.2,0.3052) (0.3,0.3187) (0.5,0.3367) (0.7,0.3456) (1,0.3412)};
      \addplot[color=acgcolor1, solid, line width=0.8pt, mark=*, mark size=1.0pt, mark options={solid}]
        coordinates {(0.02,1.4185) (0.05,2.1681) (0.1,2.6125) (0.2,3.0973) (0.3,3.2320) (0.5,3.1736) (0.7,3.3128) (1,3.3981)};
      \addplot[color=acgcolor1, dashed, line width=0.8pt, mark=*, mark size=1.0pt, mark options={solid}]
        coordinates {(0.02,2.2938) (0.05,3.1826) (0.1,3.3263) (0.2,3.3083) (0.3,3.5732) (0.5,3.3083) (0.7,3.7033) (1,3.4968)};
      \addplot[color=acgcolor1, dotted, line width=0.8pt, mark=*, mark size=1.0pt, mark options={solid}]
        coordinates {(0.02,1.4140) (0.05,2.2040) (0.1,2.6484) (0.2,2.9941) (0.3,3.3263) (0.5,3.3487) (0.7,3.6764) (1,3.4026)};
      \addplot[color=acgcolor1, dashdotted, line width=0.8pt, mark=*, mark size=1.0pt, mark options={solid}]
        coordinates {(0.02,2.2310) (0.05,2.8325) (0.1,3.3308) (0.2,3.3038) (0.3,3.5597) (0.5,3.6001) (0.7,3.4430) (1,3.5148)};
      \addplot[color=acgcolor2, solid, line width=0.8pt, mark=*, mark size=1.0pt, mark options={solid}]
        coordinates {(0.02,2.5587) (0.05,3.3846) (0.1,4.0939) (0.2,4.3901) (0.3,4.4530) (0.5,4.5697) (0.7,4.4036) (1,4.4350)};
      \addplot[color=acgcolor2, dashed, line width=0.8pt, mark=*, mark size=1.0pt, mark options={solid}]
        coordinates {(0.02,2.3073) (0.05,3.2230) (0.1,3.9592) (0.2,4.2375) (0.3,4.3452) (0.5,4.4171) (0.7,4.5607) (1,4.7492)};
      \addplot[color=acgcolor2, dotted, line width=0.8pt, mark=*, mark size=1.0pt, mark options={solid}]
        coordinates {(0.02,2.4420) (0.05,3.2275) (0.1,3.9547) (0.2,4.3363) (0.3,4.5562) (0.5,4.4799) (0.7,4.3632) (1,4.7492)};
      \addplot[color=acgcolor2, dashdotted, line width=0.8pt, mark=*, mark size=1.0pt, mark options={solid}]
        coordinates {(0.02,2.3297) (0.05,3.2948) (0.1,3.9457) (0.2,4.5652) (0.3,4.4440) (0.5,4.5876) (0.7,4.5024) (1,4.5876)};
      \addplot[color=acgcolor3, solid, line width=0.8pt, mark=*, mark size=1.0pt, mark options={solid}]
        coordinates {(0.02,16.9456) (0.05,16.1645) (0.1,14.8088) (0.2,14.7550) (0.3,13.1524) (0.5,11.5948) (0.7,11.6576) (1,11.3120)};
      \addplot[color=acgcolor3, dashed, line width=0.8pt, mark=*, mark size=1.0pt, mark options={solid}]
        coordinates {(0.02,5.5797) (0.05,6.9219) (0.1,8.2012) (0.2,9.0945) (0.3,9.2740) (0.5,10.3379) (0.7,11.0247) (1,11.4287)};
      \addplot[color=acgcolor3, dotted, line width=0.8pt, mark=*, mark size=1.0pt, mark options={solid}]
        coordinates {(0.02,9.2516) (0.05,11.0606) (0.1,12.3130) (0.2,12.2636) (0.3,11.8731) (0.5,11.4467) (0.7,11.4556) (1,11.3210)};
      \addplot[color=acgcolor3, dashdotted, line width=0.8pt, mark=*, mark size=1.0pt, mark options={solid}]
        coordinates {(0.02,19.7960) (0.05,18.1216) (0.1,17.2418) (0.2,14.7236) (0.3,13.9964) (0.5,12.5016) (0.7,11.9898) (1,11.0516)};
      \addplot[color=acgcolor4, solid, line width=0.8pt, mark=*, mark size=1.0pt, mark options={solid}]
        coordinates {(0.02,81.0155) (0.05,80.3242) (0.1,79.1795) (0.2,78.0124) (0.3,77.6264) (0.5,76.6433) (0.7,75.5660) (1,75.2472)};
      \addplot[color=acgcolor4, dashed, line width=0.8pt, mark=*, mark size=1.0pt, mark options={solid}]
        coordinates {(0.02,81.5362) (0.05,78.2817) (0.1,76.4413) (0.2,74.9465) (0.3,75.0677) (0.5,74.9061) (0.7,75.1305) (1,75.6916)};
      \addplot[color=acgcolor4, dotted, line width=0.8pt, mark=*, mark size=1.0pt, mark options={solid}]
        coordinates {(0.02,75.2203) (0.05,74.6772) (0.1,74.4213) (0.2,74.8028) (0.3,75.1575) (0.5,75.3056) (0.7,75.4223) (1,75.4672)};
      \addplot[color=acgcolor4, dashdotted, line width=0.8pt, mark=*, mark size=1.0pt, mark options={solid}]
        coordinates {(0.02,86.3034) (0.05,82.6180) (0.1,80.5307) (0.2,78.3042) (0.3,77.5949) (0.5,75.8757) (0.7,75.9296) (1,75.6423)};
    \nextgroupplot[title={\textsc{MMachine}-2D}, ymin=0.1119, ymax=46.863]
      \draw[black, dotted, line width=0.6pt] (axis cs:0.02,1) -- (axis cs:1,1);
      \addplot[color=acgcolor0, solid, line width=0.8pt, mark=*, mark size=1.0pt, mark options={solid}]
        coordinates {(0.02,0.2496) (0.05,0.3440) (0.1,0.3729) (0.2,0.3821) (0.3,0.3318) (0.5,0.4171) (0.7,0.3638) (1,0.3425)};
      \addplot[color=acgcolor0, dashed, line width=0.8pt, mark=*, mark size=1.0pt, mark options={solid}]
        coordinates {(0.02,0.1979) (0.05,0.2344) (0.1,0.3090) (0.2,0.3577) (0.3,0.3638) (0.5,0.2953) (0.7,0.3334) (1,0.3714)};
      \addplot[color=acgcolor0, dotted, line width=0.8pt, mark=*, mark size=1.0pt, mark options={solid}]
        coordinates {(0.02,0.2451) (0.05,0.3288) (0.1,0.3547) (0.2,0.3912) (0.3,0.3958) (0.5,0.3927) (0.7,0.3516) (1,0.3714)};
      \addplot[color=acgcolor0, dashdotted, line width=0.8pt, mark=*, mark size=1.0pt, mark options={solid}]
        coordinates {(0.02,0.1598) (0.05,0.2253) (0.1,0.3090) (0.2,0.3440) (0.3,0.3318) (0.5,0.3806) (0.7,0.3988) (1,0.4064)};
      \addplot[color=acgcolor1, solid, line width=0.8pt, mark=*, mark size=1.0pt, mark options={solid}]
        coordinates {(0.02,0.6789) (0.05,1.0199) (0.1,1.3198) (0.2,1.6486) (0.3,1.7840) (0.5,1.9195) (0.7,1.9287) (1,1.9180)};
      \addplot[color=acgcolor1, dashed, line width=0.8pt, mark=*, mark size=1.0pt, mark options={solid}]
        coordinates {(0.02,1.8952) (0.05,2.1920) (0.1,2.2544) (0.2,2.2316) (0.3,2.1539) (0.5,2.1083) (0.7,2.0124) (1,1.9865)};
      \addplot[color=acgcolor1, dotted, line width=0.8pt, mark=*, mark size=1.0pt, mark options={solid}]
        coordinates {(0.02,0.6500) (0.05,1.0290) (0.1,1.3639) (0.2,1.7034) (0.3,1.8038) (0.5,1.9500) (0.7,2.0215) (1,1.9545)};
      \addplot[color=acgcolor1, dashdotted, line width=0.8pt, mark=*, mark size=1.0pt, mark options={solid}]
        coordinates {(0.02,1.7932) (0.05,2.2179) (0.1,2.2544) (0.2,2.1920) (0.3,2.0580) (0.5,1.9058) (0.7,1.9804) (1,1.9926)};
      \addplot[color=acgcolor2, solid, line width=0.8pt, mark=*, mark size=1.0pt, mark options={solid}]
        coordinates {(0.02,1.0001) (0.05,1.2056) (0.1,1.1523) (0.2,1.0473) (0.3,1.0382) (0.5,0.9362) (0.7,0.9255) (1,1.0168)};
      \addplot[color=acgcolor2, dashed, line width=0.8pt, mark=*, mark size=1.0pt, mark options={solid}]
        coordinates {(0.02,0.3547) (0.05,0.4765) (0.1,0.6622) (0.2,0.7505) (0.3,0.8798) (0.5,0.9377) (0.7,1.0062) (1,0.9697)};
      \addplot[color=acgcolor2, dotted, line width=0.8pt, mark=*, mark size=1.0pt, mark options={solid}]
        coordinates {(0.02,1.1188) (0.05,1.2558) (0.1,1.1782) (0.2,1.0610) (0.3,1.0960) (0.5,0.9727) (0.7,0.9986) (1,0.9194)};
      \addplot[color=acgcolor2, dashdotted, line width=0.8pt, mark=*, mark size=1.0pt, mark options={solid}]
        coordinates {(0.02,0.3471) (0.05,0.4612) (0.1,0.6348) (0.2,0.7139) (0.3,0.8129) (0.5,1.0001) (0.7,1.0016) (1,0.9088)};
      \addplot[color=acgcolor3, solid, line width=0.8pt, mark=*, mark size=1.0pt, mark options={solid}]
        coordinates {(0.02,7.0890) (0.05,6.3446) (0.1,5.2441) (0.2,4.8117) (0.3,4.7417) (0.5,4.4023) (0.7,4.5180) (1,4.4525)};
      \addplot[color=acgcolor3, dashed, line width=0.8pt, mark=*, mark size=1.0pt, mark options={solid}]
        coordinates {(0.02,2.1783) (0.05,2.7263) (0.1,2.9531) (0.2,3.2682) (0.3,3.5833) (0.5,3.9806) (0.7,4.2440) (1,4.5621)};
      \addplot[color=acgcolor3, dotted, line width=0.8pt, mark=*, mark size=1.0pt, mark options={solid}]
        coordinates {(0.02,4.4510) (0.05,4.8178) (0.1,4.9990) (0.2,5.0583) (0.3,5.0736) (0.5,4.7539) (0.7,4.6245) (1,4.4495)};
      \addplot[color=acgcolor3, dashdotted, line width=0.8pt, mark=*, mark size=1.0pt, mark options={solid}]
        coordinates {(0.02,7.3067) (0.05,6.8546) (0.1,6.4283) (0.2,6.0463) (0.3,5.6916) (0.5,5.2456) (0.7,4.9366) (1,4.5149)};
      \addplot[color=acgcolor4, solid, line width=0.8pt, mark=*, mark size=1.0pt, mark options={solid}]
        coordinates {(0.02,30.7017) (0.05,30.4125) (0.1,29.6498) (0.2,29.5007) (0.3,29.4900) (0.5,29.5402) (0.7,29.3667) (1,29.3256)};
      \addplot[color=acgcolor4, dashed, line width=0.8pt, mark=*, mark size=1.0pt, mark options={solid}]
        coordinates {(0.02,28.7791) (0.05,28.9588) (0.1,29.0105) (0.2,29.2038) (0.3,29.0273) (0.5,29.1338) (0.7,29.4291) (1,29.3819)};
      \addplot[color=acgcolor4, dotted, line width=0.8pt, mark=*, mark size=1.0pt, mark options={solid}]
        coordinates {(0.02,29.0455) (0.05,28.6771) (0.1,29.0866) (0.2,29.2297) (0.3,29.1475) (0.5,29.2723) (0.7,29.5265) (1,29.4520)};
      \addplot[color=acgcolor4, dashdotted, line width=0.8pt, mark=*, mark size=1.0pt, mark options={solid}]
        coordinates {(0.02,31.2421) (0.05,30.5327) (0.1,30.4140) (0.2,30.7535) (0.3,30.5114) (0.5,29.9391) (0.7,29.9056) (1,29.3972)};
    \nextgroupplot[title={\textsc{KidneyStudy1}-2D}, ymin=1.0505, ymax=7.411]
      \draw[black, dotted, line width=0.6pt] (axis cs:0.02,1) -- (axis cs:1,1);
      \addplot[color=acgcolor0, solid, line width=0.8pt, mark=*, mark size=1.0pt, mark options={solid}]
        coordinates {(0.02,4.8391) (0.05,4.3237) (0.1,4.0092) (0.2,3.7678) (0.3,3.5471) (0.5,3.4467) (0.7,3.3839) (1,3.3764)};
      \addplot[color=acgcolor0, dashed, line width=0.8pt, mark=*, mark size=1.0pt, mark options={solid}]
        coordinates {(0.02,1.5304) (0.05,2.1569) (0.1,2.5814) (0.2,2.9550) (0.3,3.1677) (0.5,3.3040) (0.7,3.3409) (1,3.3526)};
      \addplot[color=acgcolor0, dotted, line width=0.8pt, mark=*, mark size=1.0pt, mark options={solid}]
        coordinates {(0.02,4.8378) (0.05,4.3390) (0.1,3.9661) (0.2,3.6945) (0.3,3.6219) (0.5,3.4087) (0.7,3.3579) (1,3.3594)};
      \addplot[color=acgcolor0, dashdotted, line width=0.8pt, mark=*, mark size=1.0pt, mark options={solid}]
        coordinates {(0.02,1.5007) (0.05,2.1412) (0.1,2.5251) (0.2,2.9580) (0.3,3.1059) (0.5,3.3154) (0.7,3.3535) (1,3.4053)};
      \addplot[color=acgcolor1, solid, line width=0.8pt, mark=*, mark size=1.0pt, mark options={solid}]
        coordinates {(0.02,4.9410) (0.05,4.4808) (0.1,4.2088) (0.2,4.0022) (0.3,3.8055) (0.5,3.7307) (0.7,3.6926) (1,3.6745)};
      \addplot[color=acgcolor1, dashed, line width=0.8pt, mark=*, mark size=1.0pt, mark options={solid}]
        coordinates {(0.02,1.9629) (0.05,2.6035) (0.1,2.9867) (0.2,3.3250) (0.3,3.5211) (0.5,3.6378) (0.7,3.6495) (1,3.6559)};
      \addplot[color=acgcolor1, dotted, line width=0.8pt, mark=*, mark size=1.0pt, mark options={solid}]
        coordinates {(0.02,4.9234) (0.05,4.4823) (0.1,4.1593) (0.2,3.9383) (0.3,3.8913) (0.5,3.6952) (0.7,3.6677) (1,3.6741)};
      \addplot[color=acgcolor1, dashdotted, line width=0.8pt, mark=*, mark size=1.0pt, mark options={solid}]
        coordinates {(0.02,1.9108) (0.05,2.5882) (0.1,2.9395) (0.2,3.3314) (0.3,3.4536) (0.5,3.6371) (0.7,3.6673) (1,3.6998)};
      \addplot[color=acgcolor2, solid, line width=0.8pt, mark=*, mark size=1.0pt, mark options={solid}]
        coordinates {(0.02,4.9262) (0.05,4.4542) (0.1,4.1786) (0.2,3.9612) (0.3,3.7630) (0.5,3.6831) (0.7,3.6393) (1,3.6221)};
      \addplot[color=acgcolor2, dashed, line width=0.8pt, mark=*, mark size=1.0pt, mark options={solid}]
        coordinates {(0.02,1.8777) (0.05,2.5260) (0.1,2.9123) (0.2,3.2572) (0.3,3.4601) (0.5,3.5838) (0.7,3.5989) (1,3.6042)};
      \addplot[color=acgcolor2, dotted, line width=0.8pt, mark=*, mark size=1.0pt, mark options={solid}]
        coordinates {(0.02,4.9098) (0.05,4.4595) (0.1,4.1272) (0.2,3.8964) (0.3,3.8535) (0.5,3.6491) (0.7,3.6199) (1,3.6265)};
      \addplot[color=acgcolor2, dashdotted, line width=0.8pt, mark=*, mark size=1.0pt, mark options={solid}]
        coordinates {(0.02,1.8254) (0.05,2.5062) (0.1,2.8654) (0.2,3.2651) (0.3,3.3934) (0.5,3.5838) (0.7,3.6123) (1,3.6476)};
      \addplot[color=acgcolor3, solid, line width=0.8pt, mark=*, mark size=1.0pt, mark options={solid}]
        coordinates {(0.02,4.6646) (0.05,4.1335) (0.1,3.8584) (0.2,3.6712) (0.3,3.5328) (0.5,3.4788) (0.7,3.4425) (1,3.4431)};
      \addplot[color=acgcolor3, dashed, line width=0.8pt, mark=*, mark size=1.0pt, mark options={solid}]
        coordinates {(0.02,2.5081) (0.05,2.8424) (0.1,3.0579) (0.2,3.2211) (0.3,3.3702) (0.5,3.4002) (0.7,3.4334) (1,3.4159)};
      \addplot[color=acgcolor3, dotted, line width=0.8pt, mark=*, mark size=1.0pt, mark options={solid}]
        coordinates {(0.02,4.5144) (0.05,4.0666) (0.1,3.7619) (0.2,3.5660) (0.3,3.5587) (0.5,3.4210) (0.7,3.4125) (1,3.4314)};
      \addplot[color=acgcolor3, dashdotted, line width=0.8pt, mark=*, mark size=1.0pt, mark options={solid}]
        coordinates {(0.02,2.8435) (0.05,3.0228) (0.1,3.1150) (0.2,3.2585) (0.3,3.3256) (0.5,3.4255) (0.7,3.4268) (1,3.4538)};
      \addplot[color=acgcolor4, solid, line width=0.8pt, mark=*, mark size=1.0pt, mark options={solid}]
        coordinates {(0.02,4.5186) (0.05,4.0851) (0.1,3.8890) (0.2,3.7825) (0.3,3.7032) (0.5,3.7098) (0.7,3.7177) (1,3.7604)};
      \addplot[color=acgcolor4, dashed, line width=0.8pt, mark=*, mark size=1.0pt, mark options={solid}]
        coordinates {(0.02,3.4778) (0.05,3.5570) (0.1,3.5906) (0.2,3.6161) (0.3,3.7077) (0.5,3.7345) (0.7,3.7428) (1,3.7292)};
      \addplot[color=acgcolor4, dotted, line width=0.8pt, mark=*, mark size=1.0pt, mark options={solid}]
        coordinates {(0.02,4.7558) (0.05,4.3856) (0.1,4.1359) (0.2,3.9514) (0.3,3.9304) (0.5,3.7800) (0.7,3.7589) (1,3.7553)};
      \addplot[color=acgcolor4, dashdotted, line width=0.8pt, mark=*, mark size=1.0pt, mark options={solid}]
        coordinates {(0.02,3.5649) (0.05,3.6048) (0.1,3.6306) (0.2,3.7056) (0.3,3.7062) (0.5,3.7617) (0.7,3.7343) (1,3.7562)};
    \nextgroupplot[title={\textsc{KidneyStudy2}-2D}, ymin=0.5833, ymax=12.818]
      \draw[black, dotted, line width=0.6pt] (axis cs:0.02,1) -- (axis cs:1,1);
      \addplot[color=acgcolor0, solid, line width=0.8pt, mark=*, mark size=1.0pt, mark options={solid}]
        coordinates {(0.02,1.0009) (0.05,1.5165) (0.1,2.0076) (0.2,2.6345) (0.3,2.8643) (0.5,3.1849) (0.7,3.2761) (1,3.3263)};
      \addplot[color=acgcolor0, dashed, line width=0.8pt, mark=*, mark size=1.0pt, mark options={solid}]
        coordinates {(0.02,6.5585) (0.05,5.6378) (0.1,4.8743) (0.2,4.2508) (0.3,3.6770) (0.5,3.4658) (0.7,3.4094) (1,3.3024)};
      \addplot[color=acgcolor0, dotted, line width=0.8pt, mark=*, mark size=1.0pt, mark options={solid}]
        coordinates {(0.02,1.0678) (0.05,1.6622) (0.1,2.1634) (0.2,2.6435) (0.3,2.8901) (0.5,3.1385) (0.7,3.2828) (1,3.3019)};
      \addplot[color=acgcolor0, dashdotted, line width=0.8pt, mark=*, mark size=1.0pt, mark options={solid}]
        coordinates {(0.02,6.5866) (0.05,5.6889) (0.1,4.9627) (0.2,4.0434) (0.3,3.8375) (0.5,3.4983) (0.7,3.3387) (1,3.3100)};
      \addplot[color=acgcolor1, solid, line width=0.8pt, mark=*, mark size=1.0pt, mark options={solid}]
        coordinates {(0.02,1.1462) (0.05,1.7210) (0.1,2.3521) (0.2,3.0434) (0.3,3.3354) (0.5,3.7119) (0.7,3.8117) (1,3.8447)};
      \addplot[color=acgcolor1, dashed, line width=0.8pt, mark=*, mark size=1.0pt, mark options={solid}]
        coordinates {(0.02,8.5221) (0.05,6.9278) (0.1,5.8332) (0.2,5.0243) (0.3,4.3048) (0.5,4.0702) (0.7,3.9627) (1,3.8437)};
      \addplot[color=acgcolor1, dotted, line width=0.8pt, mark=*, mark size=1.0pt, mark options={solid}]
        coordinates {(0.02,1.2250) (0.05,1.9035) (0.1,2.4720) (0.2,3.0597) (0.3,3.3268) (0.5,3.6459) (0.7,3.7998) (1,3.8729)};
      \addplot[color=acgcolor1, dashdotted, line width=0.8pt, mark=*, mark size=1.0pt, mark options={solid}]
        coordinates {(0.02,8.5455) (0.05,6.9708) (0.1,5.9254) (0.2,4.7577) (0.3,4.5198) (0.5,4.0797) (0.7,3.9240) (1,3.8533)};
      \addplot[color=acgcolor2, solid, line width=0.8pt, mark=*, mark size=1.0pt, mark options={solid}]
        coordinates {(0.02,0.8332) (0.05,1.2699) (0.1,1.6937) (0.2,2.2131) (0.3,2.3765) (0.5,2.6450) (0.7,2.7023) (1,2.7539)};
      \addplot[color=acgcolor2, dashed, line width=0.8pt, mark=*, mark size=1.0pt, mark options={solid}]
        coordinates {(0.02,4.6631) (0.05,4.3368) (0.1,3.8733) (0.2,3.4668) (0.3,2.9961) (0.5,2.8624) (0.7,2.7998) (1,2.7501)};
      \addplot[color=acgcolor2, dotted, line width=0.8pt, mark=*, mark size=1.0pt, mark options={solid}]
        coordinates {(0.02,0.9035) (0.05,1.4028) (0.1,1.8237) (0.2,2.2178) (0.3,2.4056) (0.5,2.6182) (0.7,2.6799) (1,2.7109)};
      \addplot[color=acgcolor2, dashdotted, line width=0.8pt, mark=*, mark size=1.0pt, mark options={solid}]
        coordinates {(0.02,4.6765) (0.05,4.3463) (0.1,3.9689) (0.2,3.2470) (0.3,3.1677) (0.5,2.9068) (0.7,2.7678) (1,2.7229)};
      \addplot[color=acgcolor3, solid, line width=0.8pt, mark=*, mark size=1.0pt, mark options={solid}]
        coordinates {(0.02,1.5580) (0.05,2.0387) (0.1,2.5552) (0.2,2.9656) (0.3,3.0764) (0.5,3.2948) (0.7,3.2551) (1,3.1834)};
      \addplot[color=acgcolor3, dashed, line width=0.8pt, mark=*, mark size=1.0pt, mark options={solid}]
        coordinates {(0.02,6.9784) (0.05,5.7682) (0.1,4.9039) (0.2,4.2016) (0.3,3.7109) (0.5,3.4395) (0.7,3.3378) (1,3.1720)};
      \addplot[color=acgcolor3, dotted, line width=0.8pt, mark=*, mark size=1.0pt, mark options={solid}]
        coordinates {(0.02,1.0024) (0.05,1.4114) (0.1,1.7855) (0.2,2.2604) (0.3,2.5222) (0.5,2.8385) (0.7,3.0344) (1,3.1767)};
      \addplot[color=acgcolor3, dashdotted, line width=0.8pt, mark=*, mark size=1.0pt, mark options={solid}]
        coordinates {(0.02,6.4447) (0.05,5.3330) (0.1,4.5260) (0.2,3.7367) (0.3,3.5704) (0.5,3.3029) (0.7,3.1944) (1,3.1892)};
      \addplot[color=acgcolor4, solid, line width=0.8pt, mark=*, mark size=1.0pt, mark options={solid}]
        coordinates {(0.02,3.2627) (0.05,3.5088) (0.1,3.7988) (0.2,4.0052) (0.3,4.0888) (0.5,4.2317) (0.7,4.2111) (1,4.2245)};
      \addplot[color=acgcolor4, dashed, line width=0.8pt, mark=*, mark size=1.0pt, mark options={solid}]
        coordinates {(0.02,7.4552) (0.05,6.4080) (0.1,5.5738) (0.2,5.0047) (0.3,4.5781) (0.5,4.4376) (0.7,4.3062) (1,4.2097)};
      \addplot[color=acgcolor4, dotted, line width=0.8pt, mark=*, mark size=1.0pt, mark options={solid}]
        coordinates {(0.02,3.1968) (0.05,3.4309) (0.1,3.6249) (0.2,3.8504) (0.3,3.9541) (0.5,4.1199) (0.7,4.1548) (1,4.2465)};
      \addplot[color=acgcolor4, dashdotted, line width=0.8pt, mark=*, mark size=1.0pt, mark options={solid}]
        coordinates {(0.02,7.3821) (0.05,6.3148) (0.1,5.5427) (0.2,4.8380) (0.3,4.7047) (0.5,4.3621) (0.7,4.2513) (1,4.2154)};
  \end{groupplot}
  % --- legend ---
  \path (group c1r1.south west) -- (group c5r1.south east) coordinate[midway] (legmid);
  \path let \p1=(group c1r1.south west),
           \p2=(group c5r1.south east) in
       node[anchor=north, font=\footnotesize, text width={\x2-\x1}, align=center, inner sep=0pt, yshift=-12mm] at (legmid) {%
    \mbox{\tikz[baseline=-0.5ex] \draw[acgcolor0, line width=1pt] (0,0) -- (0.3,0);~arithmetic mean}\hspace{6pt}\mbox{\tikz[baseline=-0.5ex] \draw[acgcolor1, line width=1pt] (0,0) -- (0.3,0);~angular mean}\hspace{6pt}\mbox{\tikz[baseline=-0.5ex] \draw[acgcolor2, line width=1pt] (0,0) -- (0.3,0);~geometric median}\hspace{6pt}\mbox{\tikz[baseline=-0.5ex] \draw[acgcolor3, line width=1pt] (0,0) -- (0.3,0);~Borda}\hspace{6pt}\mbox{\tikz[baseline=-0.5ex] \draw[acgcolor4, line width=1pt] (0,0) -- (0.3,0);~PSB} \\
    \mbox{\tikz[baseline=-0.5ex] \draw[gray, solid, line width=1pt] (0,0) -- (0.3,0);~$v$ at $0^\circ$}\hspace{6pt}\mbox{\tikz[baseline=-0.5ex] \draw[gray, dashed, line width=1pt] (0,0) -- (0.3,0);~$v$ at $45^\circ$}\hspace{6pt}\mbox{\tikz[baseline=-0.5ex] \draw[gray, dotted, line width=1pt] (0,0) -- (0.3,0);~$v$ at $90^\circ$}\hspace{6pt}\mbox{\tikz[baseline=-0.5ex] \draw[gray, dashdotted, line width=1pt] (0,0) -- (0.3,0);~$v$ at $135^\circ$}
  };
\end{tikzpicture}